\documentclass[12pt,oneside]{amsart}  
\usepackage{amssymb}
\usepackage{amsfonts}
\usepackage{amsthm}
\usepackage{amsthm}
\usepackage{amscd}
\usepackage{wrapfig}
\usepackage[matrix, arrow, curve]{xy}
\usepackage{graphicx}

\reversemarginpar         
\numberwithin{equation}{section}
\theoremstyle{plain}
\newtheorem{theorem}{Theorem}[section]

\newtheorem{corollary}[theorem]{Corollary}
\newtheorem{lemma}[theorem]{Lemma}
\theoremstyle{definition}
\newtheorem{definition}[theorem]{Definition}
\theoremstyle{remark}
\newtheorem{remark}{Remark}[section]
\newtheorem{example}[theorem]{Example}

\begin{document}
%

\newcommand{\MgNekp}{\mathcal{M}_{g,N+1}^{(k,p)}} 
\newcommand{\M}{\mathcal{M}}
\newcommand{\Teich}{\mathcal{T}_{g,N+1}^{(1)}}
\newcommand{\T}{\mathrm{T}}
\newcommand{\corr}{\bf}
\newcommand{\vac}{|0\rangle}
\newcommand{\Ga}{\Gamma}
\newcommand{\new}{\bf}
\newcommand{\define}{\def}
\newcommand{\redefine}{\def}
\newcommand{\Cal}[1]{\mathcal{#1}}
\renewcommand{\frak}[1]{\mathfrak{{#1}}}
\newcommand{\refE}[1]{(\ref{E:#1})}
\newcommand{\refS}[1]{Section~\ref{S:#1}}
\newcommand{\refSS}[1]{Section~\ref{SS:#1}}
\newcommand{\refT}[1]{Theorem~\ref{T:#1}}
\newcommand{\refO}[1]{Observation~\ref{O:#1}}
\newcommand{\refP}[1]{Proposition~\ref{P:#1}}
\newcommand{\refD}[1]{Definition~\ref{D:#1}}
\newcommand{\refC}[1]{Corollary~\ref{C:#1}}
\newcommand{\refL}[1]{Lemma~\ref{L:#1}}
\newcommand{\refR}[1]{Remark~\ref{R:#1}}
\newcommand{\refEx}[1]{Example~\ref{Ex:#1}}
\newcommand{\R}{\ensuremath{\mathbb{R}}}
\newcommand{\C}{\ensuremath{\mathbb{C}}}
\newcommand{\N}{\ensuremath{\mathbb{N}}}
\newcommand{\Q}{\ensuremath{\mathbb{Q}}}
\renewcommand{\P}{\ensuremath{\mathcal{P}}}
\newcommand{\Z}{\ensuremath{\mathbb{Z}}}
\newcommand{\kv}{{k^{\vee}}}
\renewcommand{\l}{\lambda}
\newcommand{\gb}{\overline{\mathfrak{g}}}
\newcommand{\hb}{\overline{\mathfrak{h}}}
\newcommand{\g}{\mathfrak{g}}
\newcommand{\h}{\mathfrak{h}}
\newcommand{\gh}{\widehat{\mathfrak{g}}}
\newcommand{\ghN}{\widehat{\mathfrak{g}_{(N)}}}
\newcommand{\gbN}{\overline{\mathfrak{g}_{(N)}}}
\newcommand{\tr}{\mathrm{tr}}
\newcommand{\sln}{\mathfrak{sl}}
\newcommand{\sn}{\mathfrak{s}}
\newcommand{\so}{\mathfrak{so}}
\newcommand{\spn}{\mathfrak{sp}}
\newcommand{\tsp}{\mathfrak{tsp}(2n)}
\newcommand{\gl}{\mathfrak{gl}}
\newcommand{\slnb}{{\overline{\mathfrak{sl}}}}
\newcommand{\snb}{{\overline{\mathfrak{s}}}}
\newcommand{\sob}{{\overline{\mathfrak{so}}}}
\newcommand{\spnb}{{\overline{\mathfrak{sp}}}}
\newcommand{\glb}{{\overline{\mathfrak{gl}}}}
\newcommand{\Hwft}{\mathcal{H}_{F,\tau}}
\newcommand{\Hwftm}{\mathcal{H}_{F,\tau}^{(m)}}

\newcommand{\car}{{\mathfrak{h}}}    
\newcommand{\bor}{{\mathfrak{b}}}    
\newcommand{\nil}{{\mathfrak{n}}}    
\newcommand{\vp}{{\varphi}}
\newcommand{\bh}{\widehat{\mathfrak{b}}}  
\newcommand{\bb}{\overline{\mathfrak{b}}}  
\newcommand{\Vh}{\widehat{\mathcal V}}
\newcommand{\KZ}{Kniz\-hnik-Zamo\-lod\-chi\-kov}
\newcommand{\TUY}{Tsuchia, Ueno  and Yamada}
\newcommand{\KN} {Kri\-che\-ver-Novi\-kov}
\newcommand{\pN}{\ensuremath{(P_1,P_2,\ldots,P_N)}}
\newcommand{\xN}{\ensuremath{(\xi_1,\xi_2,\ldots,\xi_N)}}
\newcommand{\lN}{\ensuremath{(\lambda_1,\lambda_2,\ldots,\lambda_N)}}
\newcommand{\iN}{\ensuremath{1,\ldots, N}}
\newcommand{\iNf}{\ensuremath{1,\ldots, N,\infty}}

\newcommand{\tb}{\tilde \beta}
\newcommand{\tk}{\tilde \kappa}
\newcommand{\ka}{\kappa}
\renewcommand{\k}{\kappa}

\newcommand{\Pif} {P_{\infty}}
\newcommand{\Pinf} {P_{\infty}}
\newcommand{\PN}{\ensuremath{\{P_1,P_2,\ldots,P_N\}}}
\newcommand{\PNi}{\ensuremath{\{P_1,P_2,\ldots,P_N,P_\infty\}}}
\newcommand{\Fln}[1][n]{F_{#1}^\lambda}
\newcommand{\tang}{\mathrm{T}}
\newcommand{\Kl}[1][\lambda]{\can^{#1}}
\newcommand{\A}{\mathcal{A}}
\newcommand{\U}{\mathcal{U}}
\newcommand{\V}{\mathcal{V}}
\renewcommand{\O}{\mathcal{O}}
\newcommand{\Ae}{\widehat{\mathcal{A}}}
\newcommand{\Ah}{\widehat{\mathcal{A}}}
\newcommand{\La}{\mathcal{L}}
\newcommand{\Lp}{\mathcal Lp}
\newcommand{\Le}{\widehat{\mathcal{L}}}
\newcommand{\Lh}{\widehat{\mathcal{L}}}
\newcommand{\eh}{\widehat{e}}
\newcommand{\Da}{\mathcal{D}}
\newcommand{\kndual}[2]{\langle #1,#2\rangle}
\newcommand{\cins}{\frac 1{2\pi\mathrm{i}}\int_{C_S}}
\newcommand{\cinsl}{\frac 1{24\pi\mathrm{i}}\int_{C_S}}
\newcommand{\cinc}[1]{\frac 1{2\pi\mathrm{i}}\int_{#1}}
\newcommand{\cintl}[1]{\frac 1{24\pi\mathrm{i}}\int_{#1 }}
\newcommand{\w}{\omega}
\newcommand{\ord}{\operatorname{ord}}
\newcommand{\res}{\operatorname{res}}
\newcommand{\nord}[1]{:\mkern-5mu{#1}\mkern-5mu:}
\newcommand{\ad}{\operatorname{ad}}
\newcommand{\Ad}{\operatorname{Ad}}
\newcommand{\codim}{\operatorname{codim}}
\newcommand{\Fn}[1][\lambda]{\mathcal{F}^{#1}}
\newcommand{\Fl}[1][\lambda]{\mathcal{F}^{#1}}
\renewcommand{\Re}{\mathrm{Re}}

\newcommand{\ha}{H^\alpha}

\define\ldot{\hskip 1pt.\hskip 1pt}
\define\ifft{\qquad\text{if and only if}\qquad}
\define\a{\alpha}
\redefine\d{\delta}
\define\w{\omega}
\define\ep{\epsilon}
\redefine\b{\beta} \redefine\t{\tau} \redefine\i{{\,\mathrm{i}}\,}
\define\ga{\gamma}
\define\cint #1{\frac 1{2\pi\i}\int_{C_{#1}}}
\define\cintta{\frac 1{2\pi\i}\int_{C_{\tau}}}
\define\cintt{\frac 1{2\pi\i}\oint_{C}}
\define\cinttp{\frac 1{2\pi\i}\int_{C_{\tau'}}}
\define\cinto{\frac 1{2\pi\i}\int_{C_{0}}}
\define\cinttt{\frac 1{24\pi\i}\int_C}
\define\cintd{\frac 1{(2\pi \i)^2}\iint\limits_{C_{\tau}\,C_{\tau'}}}
\define\cintdr{\frac 1{(2\pi \i)^3}\int_{C_{\tau}}\int_{C_{\tau'}}
\int_{C_{\tau''}}}
\define\im{\operatorname{Im}}
\define\re{\operatorname{Re}}
\define\res{\operatorname{res}}
\redefine\deg{\operatornamewithlimits{deg}}
\define\ord{\operatorname{ord}}
\define\rank{\operatorname{rank}}
\define\fpz{\frac {d }{dz}}
\define\dzl{\,{dz}^\l}
\define\pfz#1{\frac {d#1}{dz}}

\define\K{\Cal K}
\define\U{\Cal U}
\redefine\O{\Cal O}
\define\He{\text{\rm H}^1}
\redefine\H{{\mathrm{H}}}
\define\Ho{\text{\rm H}^0}
\define\A{\Cal A}
\define\Do{\Cal D^{1}}
\define\Dh{\widehat{\mathcal{D}}^{1}}
\redefine\L{\Cal L}
\newcommand{\ND}{\ensuremath{\mathcal{N}^D}}
\redefine\D{\Cal D^{1}}
\define\KN {Kri\-che\-ver-Novi\-kov}
\define\Pif {{P_{\infty}}}
\define\Uif {{U_{\infty}}}
\define\Uifs {{U_{\infty}^*}}
\define\KM {Kac-Moody}
\define\Fln{\Cal F^\lambda_n}
\define\gb{\overline{\mathfrak{ g}}}
\define\G{\overline{\mathfrak{ g}}}
\define\Gb{\overline{\mathfrak{ g}}}
\redefine\g{\mathfrak{ g}}
\define\Gh{\widehat{\mathfrak{ g}}}
\define\gh{\widehat{\mathfrak{ g}}}
\define\Ah{\widehat{\Cal A}}
\define\Lh{\widehat{\Cal L}}
\define\Ugh{\Cal U(\Gh)}
\define\Xh{\hat X}
\define\Tld{...}
\define\iN{i=1,\ldots,N}
\define\iNi{i=1,\ldots,N,\infty}
\define\pN{p=1,\ldots,N}
\define\pNi{p=1,\ldots,N,\infty}
\define\de{\delta}

\define\kndual#1#2{\langle #1,#2\rangle}
\define \nord #1{:\mkern-5mu{#1}\mkern-5mu:}
\define \sinf{{\widehat{\sigma}}_\infty}
\define\Wt{\widetilde{W}}
\define\St{\widetilde{S}}
\newcommand{\SigmaT}{\widetilde{\Sigma}}
\newcommand{\hT}{\widetilde{\frak h}}
\define\Wn{W^{(1)}}
\define\Wtn{\widetilde{W}^{(1)}}
\define\btn{\tilde b^{(1)}}
\define\bt{\tilde b}
\define\bn{b^{(1)}}
%
\define\eps{\varepsilon}    
\define\doint{({\frac 1{2\pi\i}})^2\oint\limits _{C_0}
       \oint\limits _{C_0}}                            
\define\noint{ {\frac 1{2\pi\i}} \oint}   
\define \fh{{\frak h}}     
\define \fg{{\frak g}}     
\define \GKN{{\Cal G}}   
\define \gaff{{\hat\frak g}}   
\define\V{\Cal V}
\define \ms{{\Cal M}_{g,N}} 
\define \mse{{\Cal M}_{g,N+1}} 
\define \tOmega{\Tilde\Omega}
\define \tw{\Tilde\omega}
\define \hw{\hat\omega}
\define \s{\sigma}
\define \car{{\frak h}}    
\define \bor{{\frak b}}    
\define \nil{{\frak n}}    
\define \vp{{\varphi}}
\define\bh{\widehat{\frak b}}  
\define\bb{\overline{\frak b}}  
\define\Vh{\widehat V}
\define\KZ{Knizhnik-Zamolodchikov}
\define\ai{{\alpha(i)}}
\define\ak{{\alpha(k)}}
\define\aj{{\alpha(j)}}
\newcommand{\laxgl}{\overline{\mathfrak{gl}}}
\newcommand{\laxsl}{\overline{\mathfrak{sl}}}
\newcommand{\laxso}{\overline{\mathfrak{so}}}
\newcommand{\laxsp}{\overline{\mathfrak{sp}}}
\newcommand{\laxs}{\overline{\mathfrak{s}}}
\newcommand{\laxg}{\overline{\frak g}}
\newcommand{\bgl}{\laxgl(n)}
\newcommand{\tX}{\widetilde{X}}
\newcommand{\tY}{\widetilde{Y}}
\newcommand{\tZ}{\widetilde{Z}}

\vspace*{-1cm}
%
%
%
\vspace*{2cm}

\title[Lax operator algebras and integrable systems]
{
Lax operator algebras and integrable systems}
\author[O.K.Sheinman]{O.K.Sheinman}
\thanks{This work is supported by the Russian Science Foundation under grant
14-50-00005.
}

\begin{abstract}
A new class of infinite-dimensional Lie algebras given a name of Lax operator algebras, and the related unifying approach to finite-dimensional integrable systems with spectral parameter on a Riemann surface, such as Calogero--Moser and Hitchin systems, are presented. In particular, our approach includes (the non-twisted) Kac--Moody algebras and integrable systems with rational spectral parameter. The presentation is based on the quite simple ideas related to gradings of semisimple Lie algebras, and their interaction with the Riemann--Roch theorem. The basic properties of the Lax operator algebras, as well as the basic facts of the theory of integrable systems in question, are treated (and proven) from this general point of  view, in particular existence of commutative hierarchies and their Hamiltonian properties. We conclude with an application of the Lax operator algebras to prequantization of finite-dimensional integrable systems.

\end{abstract}
\maketitle
\tableofcontents

\section{Introduction}\label{S:intro}

Methods of algebraic geometry, and methods of Lie group and algebra theory are two classical approaches in the theory of integrable systems originating in the 70th of the 20th century \cite{Nov_KdV1,Met_alg_g,rKNU,OPer_Inv}. A broad and, as people say, horizonless literature is devoted to them. However, by our opinion, the subject is sufficiently well characterized by the following fundamental surveys and monographs: \cite{ZMNP,IntSys1,AsPer,OPer,IntSys2,R_S_TSh,Dr_Soc}. In this work we will focus on finite-dimensional integrable systems.

The interrelation between algebro-geometric and group theory approaches in the theory of finite-dimensional integrable systems has been brightly manifested in creation of the Hitchin systems \cite{Hit}. It is based on the application of holomorphic bundles on Riemann surfaces with a simple Lie group as a structure group, and finally on the remarkable interrelation between the Riemann--Roch theorem and the structure of semi-simple Lie algebras. Applying of holomorphic $G$-bundles on elliptic curves led to many results on Calogero--Moser systems (for example, see \cite{Levin_Ol_Zot} and references therein).

In our work, we would like to elucidate a new direction related to integrable systems with spectral parameter on a Riemann surface \cite{Klax,Zakhar_Mikh}, and emergence of a new class of infinite-dimensional Lie algebras of algebro-geometric nature called \emph{Lax operator algebras} \cite{KSlax,Sh_DGr,Lax_Grad,Sh_DAN_LOA&gr}.

In \cite{Zakhar_Mikh}, it was shown that a naive generalization of Lax pairs with a spectral parameter onto the case when the last belongs to a Riemann surface results in overdetermined systems, as it follows from the Riemann--Roch theorem. In \cite{Klax}, I.Krichever got over the obstruction. In his work, the finite-zone integration technique previously elaborated for the Kadomtsev--Petviashvili equation \cite{rKNU}, had been applied to  description of a certain class of Lax equations including Hitchin equations, some their generalizations, and also zero curvature systems with one spatial variable. The description descended to pointing out the form of the Laurent expansions of Lax equations in terms of Tyurin parameters of holomorphic vector bundles on curves. For the introduced class of equations, Krichever has constructed hierarchies of commuting flows, proved they are Hamiltonian, and developed the methods of algebro-geometric integration, namely method of Baker--Akhieser functions, and method of deformation of Tyurin parameters. In this work I.Krichever actually has formulated a certain program of creating a general theory of finite-dimensional integrable systems, which passed through several generalizations already.

In \cite{KSlax} I.Krichever and the author have shown that the Lax operators introduced in \cite{Klax}, form a Lie algebra generalizing the loop algebra for $\gl(n)$, and constructed its analogs for the orthogonal and symplectic algebras. It was the main point also here that the form of the Laurent expansions of the Lax operators in terms of Tyurin parameters was invariant with respect to the commutator, and its codimension in the space of all formal $\g$-valued Laurent expansions was equal to $k\dim\g$, $k$ being the pole order. It was shown that the obtained infinite-dimensional Lie algebras named \emph{Lax operator algebras} by us, possess an almost graded structure and central extensions similar to those for loop algebras. In the coming next works (see a survey and references in \cite{Sh_DGr}) the author proved the existence theorem for commuting flows, and the Hamiltonian property for them, in case the Lax operators belong to those Lie algebras. The proofs were given for the classic Lie algebras and depended on their type, though a similarity between them was obvious. It drows attention that all the integrability theorems finally appeals to the same relations on Tyurin parameters that implied closeness of the space of the Lax operators with respect to the commutator. These results, as well as some others, like a relation of the Lax equations in question to conformal field theory, has been summarized in the monograph \cite{Sh_DGr}.

In \cite{Sh_G2,Sh_G2_AMS}, there were constructed Lax operator algebras for the Lie algebra of type $G_2$, again in terms of Tyurin parameters. Uniqueness of central extensions, and Lax operator algebras for an arbitrary number of marked points on the Riemann surface have been studied in \cite{SSlax} and \cite{Sch_Laxmulti}, respectively.

In \cite{Sh_DAN_LOA&gr,Lax_Grad}, there was proposed a construction of Lax operator algebras in terms of roots, suitable for arbitrary semi-simple Lie algebras. It came true due to discussions with E.B.Vinberg who associated the Laurent expansions in the existing examples of Lax operator algebras  with $\Z$-gradings of the corresponding finite-dimensional Lie algebras (see relation \refE{ga_expan} below). Starting from this observation, in \cite{Sh_DAN_LOA&gr,Lax_Grad} there was constructed the Lax operator algebra, its almost graded structure and central extensions corresponding to an arbitrary semi-simple Lie algebra and its $\Z$-grading. It is shown how retrieve the Tyurin parameters or their analogs by a $\Z$-grading.

In the consequent works \cite{Sh_TMPh_14,Tr_MIAN_15,Sh_MMJ_15}, at the same level of generality which was given by works \cite{Sh_DAN_LOA&gr,Lax_Grad} (i.e. for a Riemann surface of an arbitrary genus, and arbitrary complex semi-simple Lie algebra), there were considered basic question of the theory of finite-dimensional integrable systems, i.e. existence of commutative hierarchies and their Hamiltonian properties.

The present work is devoted to systematic treatment of the above outlined circle of  questions.

In \refS{LaxAlgs} we give a general construction of Lax operator algebras, and establish its correspondence to the construction in terms of Tyurin parameters, previously proposed for classical Lie algebras. We introduce  almost graded structures and describe almost graded central extensions of Lax operator algebras.

In \refS{Hierarchies} we define $M$-operators -- the counterparts of  $L$-operators in Lax pairs. Given a Lax operator, we construct a family of $M$-operators determining the hierarchy of commuting flows related to $L$.

In \refS{Ham_th}, for the systems in question, we construct the Krichever--Phong symplectic structure \cite{Sh_DGr,Klax}, and the family of Hamiltonians in involution corresponding to the above constructed flows. Every Hamiltonian is defined by an invariant polynomial on the Lie algebra, a point of the Riemann surface, and an integer. The same data define $M$-operators constructed in \refS{Hierarchies}. The relation between them is as follows: the main part of an $M$-operator is given by the gradient of the corresponding invariant polynomial. This relation is well-known in the group-theoretic approach to Lax integrable systems \cite{R_S_TSh,Goldman}. Further on, we consider examples in \refS{Ham_th}: the Hitchin and Calogero--Moser systems for classical Lie algebras. It should be noted here that the obtained results are also applicable to many known systems with a rational spectral parameter such as classical gyroscopes, integrable cases of a solid in a flow.

In the concluding \refS{LaxCFT}, following \cite{Sh_L_KZ,Sh_DGr}, we treat the interrelation between integrable systems in question and 2-dimensional conformal field theories. Namely, with every integrable system we associate a unitary representation of the algebra of its classical observables by Knizhnik--Zamolodchikov-type operators on the space of spectral curves.

The author is grateful to I.M.Krichever and M.Schlichenmaier for numerous, long-term discussions, and joint works, and to E.B.Vinberg for discussions on appliance of gradings of semi-simple Lie algebras to the construction of Lax operator algebras. 
\section{Lax operator algebras}\label{S:LaxAlgs}
In addition to the Introduction, let us stress that the current state of the theory of Lax operator algebras by the end of 2013 is summarized in \cite{Sh_DGr}. In that monograph, a construction of Lax operator algebras is given for classical Lie algebras. The presentation was based on Tyurin parameters. Here, we will give a general construction suitable for an arbitrary semi-simple Lie algebra. We introduce almost graded structures and consider central extensions of Lax operator algebras. We show how to retrieve the Tyurin parameters from $\Z$-gradings for classical Lie algebras.

The two basic results of this chapter are given by \refT{almgrad} of  \refSS{constr}, and \refT{central} of \refSS{constr}. The first of them gives a description of the Lie structure, and the almost graded structures of Lax operator algebras, while the second treats construction and classification of their central extensions. In \refSS{Tyupar} we consider a number of examples, in particular classical Lie algebras, and $G_2$, in order to point out the correspondence with the previous approach, in particular,  explain emergence of the Tyurin parameters.

\subsection{Current algebra and its almost graded structures} \label{SS:constr}

Let $\g$ be a semi-simple Lie algebra over $\C$, $\h$ be its Cartan subalgebra, and $h\in\h$ be such element that  $p_i=\a_i(h)\in\Z_+$ for every simple root $\a_i$ of $\g$. Let $\g_p=\{ X\in\g\ |\ (\ad h)X=pX \}$, and $k=\max\{p\ |\ \g_p\ne 0\}$. Then the decomposition $\g=\bigoplus\limits_{i=-k}^{k}\g_p$ gives a $\Z$-grading on $\g$. For the theory and classification results on such kind of gradings we refer to \cite{Vin}. Call $k$ a \emph{depth} of the grading. Obviously, $\g_p=\bigoplus\limits_{\substack{\a\in R\\ \a(h)=p}}\g_\a$ where $R$ is the root system of $\g$. Define also the following filtration on $\g$:  $\tilde\g_p=\bigoplus\limits_{q=-k}^p\g_q$. Then $\tilde\g_p\subset\tilde\g_{p+1}$ ($p\ge -k$), $\tilde\g_{-k}=\g_{-k},\ldots,\tilde\g_k=\g$, $\tilde\g_p=\g$, $p>k$.

Let $\Sigma$ be a complex compact Riemann surface with two given finite sets of marked points: $\Pi$ and $\Gamma$. Let $L$ be a meromorphic mapping $\Sigma\to\g$ holomorphic outside the marked points which may have poles of an arbitrary order at the points in $\Pi$, and has the decomposition of the following form at the points in $\Gamma$:
\begin{equation}\label{E:ga_expan}
   L(z)=\sum\limits_{p=-k}^\infty L_pz^p,\ L_p\in\tilde\g_p
\end{equation}
where $z$ is a local coordinate in the neighborhood of a $\ga\in\Gamma$. In general, the grading element $h$ may vary from one to another point of $\Gamma$. For simplicity, we assume that $k$ is the same all over $\Gamma$, though it would be no difference otherwise.

We denote by $\L$ a linear space of all such mappings. Since the relation \refE{ga_expan} is preserved under commutator, $\L$ is a Lie algebra. Below, we fix this important fact as  assertion $1^\circ$ of \refT{almgrad}. The Lie algebra $\L$, its almost graded structures and central extensions are the main subject of the present section. Sometimes we use the notation $\gb$ instead $\L$, in order to stress the relation of this algebra to a certain semi-simple Lie algebra $\g$. We keep also the name of {\it Lax operator algebras} for just constructed current algebras, in order to stress their succession to those in \cite{KSlax,Sh_DGr}. The observation that the Laurent expansions of the elements of Lax operator algebras considered in  \cite{KSlax,Sh_DGr}, have the form (2.1) at the points in $\Gamma$ is due to E.B.Vinberg \footnote{Private communication}.
\begin{definition}
Given a Lie algebra $\L$, by its {\it almost graded structure} we mean a system of its finite-dimensional subspaces $\L_m$, and two nonnegative integers $R$, $S$ such that $\L =\bigoplus\limits_{m=-\infty}^\infty \L_m$, and $[\L_m,\L_n]\subseteq\bigoplus\limits_{r=m+n-R}^{m+n+S}\L_r$ ($R$, $S$ don't depend on $m$, $n$).
\end{definition}
The knowledge of almost graded structure on associative and Lie algebras is introduced by I.M.Krichever and S.P.Novikov in \cite{KNFa}. For Lax operator algebras, in the case when $\Pi$ consists of two points (we distinguish the two-point and the many point cases depending on the number of elements in $\Pi$), it was investigated in \cite{KSlax}. Under more general assumptions, almost graded structures for Krichever--Novikov algebras, and for Lax operator algebras as well, were considered by M.Schlichenmaier \cite{SLa,SLb,Sch_Laxmulti}.

The above defined Lie algebra $\L$ admits several almost graded structures. To define such a structure we give a splitting of $\Pi$ to two disjoint subsets: $\Pi=\{ P_i\ |\ i=1,\ldots,N \}\cup \{ Q_j\ |\ j=1,\ldots,M \}$. Following \cite{SLa,SLb,Sch_Laxmulti}, for every $m\in\Z$ we consider three divisors:
\begin{equation}\label{E:3div}
 D^P_m=-m\sum\limits_{i=1}^NP_i,\ D^Q_m=\sum_{j=1}^M(a_jm+b_{m,j})\,Q_j,\ D^\Gamma=k\sum_{\ga\in\Gamma}\ga ,
\end{equation}
where $a_j,b_{m,j}\in\Q$, $a_j>0$, and $a_jm+b_{m,j}$ is an icreasing $\Z$ -valued function of $m$, and also there exists $B\in\R_+$ such that
$|b_{m,j}|\le B, \forall m\in\Z, j=1,\ldots,M$\label{bjmbound}.
We require that
\begin{equation}\label{E:condd}
\sum_{j=1}^Ma_j=N,\qquad
\sum_{i=j}^Mb_{m,j}=N+g-1.
\end{equation}
Let
\begin{equation}\label{E:Dm}
  D_m=D_m^P+D_m^Q+D^\Gamma
\end{equation}
and
\begin{equation}\label{E:homos}
   \L_m=\{L\in\L\ |(L)+D_m\ge 0\},
\end{equation}
where $(L)$ is the divisor of a $\g$-valued function $L$.
With relation to the definition of $(L)$ we stress that by order of a meromorphic \emph{vector-valued} function at a point we call the minimal order of its entries.

We call $\L_m$ {\it (homogeneous, grading) subspace of degree $m$} of the Lie algebra $\L$.
\begin{theorem}\label{T:almgrad}
\begin{itemize}
  \item[]
  \item[$1^\circ $] The subspace $\L$ is closed with respect to the point-wise commutator $[L,L'](P)=[L(P),L'(P)]$ $(P\in\Sigma)$.
  \item[$2^\circ$] $\dim\,\L_m=N\dim\,\g$;
  \item[$3^\circ$] $\L =\bigoplus\limits_{m=-\infty}^\infty \L_m$ ;
  \item[$4^\circ$] $[\L_m,\L_n]\subseteq\bigoplus\limits_{r=m+n}^{m+n+ S}\L_r$ where $S$ is a positive integer depending on $N$, $M$, and
$g$, and independent of $m$, $n$.
\end{itemize}
\end{theorem}
\begin{proof}
Proof of $1^\circ $ is obvious as it has been noted above. For the proofs of $3^\circ$, $4^\circ$ we refer to \cite{Sch_Laxmulti} (where they are given for classical Lie algebras but are true in the set-up of the present work in fact). This is only assertion $2^\circ$ which requires any special proof, and we will give it now.

Let $L(D_m)=\{L\ |\ (L)+D_m\ge 0\}$ where $L:\Sigma\to\g$ still is a meromorphic function but the requirement $L\in\L$ is relaxed. Let $l_m=\dim L(D_m)$. In case the points in $\Pi$, $\Gamma$ are in generic position the value $l_m$ is given by the Riemann--Roch theorem:
\[
  l_m=(\dim\g)(\deg\, D_m-g+1).
\]
Observe that
\[
  \deg D_m=-mN+m\sum_{i=1}^Ma_i+\sum_{i=1}^Mb_{m,i}+k|\Gamma|
\]
where $|\Gamma|$ denotes cardinality of $\Gamma$.
By \refE{condd} we obtain
$
  \deg D_m=N+g-1+k|\Gamma|
$.
Hence
\[
 l_m=(\dim\g)(N+k|\Gamma|).
\]

By definition, $\L_m$ is a subspace in $L(D_m)$ distinguished by means the conditions \refE{ga_expan} which must hold at every $\ga\in\Gamma$. For any $\ga\in\Gamma$ the codimension of expansions of the form \refE{ga_expan}, given at $\ga$, in the space of all $\g$-valued Laurent expansions in $z$ starting with  $z^{-k}$ can be computed as follows:
$c_\ga=\sum\limits_{p=-k}^{k-1} \codim_{\,\g}\,\tilde\g_p$ (taking into account that $\codim_{\,\g}\tilde\g_p=0$ for $p\ge k$). The grading $\g=\bigoplus\limits_{p=-k}^k\g_p$ is symmetric in the sense  that $\dim\g_p=\dim\g_{-p}$ \cite{Vin} (moreover, $\g_p$ and $\g_{-p}$ are contragredient $\g_0$-modules). By definition, $\codim_{\,\g}\tilde\g_p=\sum\limits_{q=p+1}^k\dim\g_q$, by symmetry $\sum\limits_{q=p+1}^k\dim\g_q=\dim\tilde\g_{-p-1}$, hence $\codim_{\,\g}\,\tilde\g_p+\codim_{\,\g}\tilde\g_{-p-1}=\dim\g$. Obviously, $c_\ga=\sum\limits_{p=-k}^{-1} (\codim_{\,\g}\,\tilde\g_p+\codim_{\,\g}\tilde\g_{-p-1})$.   Hence $c_\ga=k\dim\g$.

Further on, $\codim_{L(D_m)}\L_m=\sum\limits_{\ga\in\Gamma} c_\ga=k(\dim\g)|\Gamma|$, and finally
\[
 \dim\L_m=l_m-k(\dim\g)|\Gamma|=N\dim\g .
\]
\end{proof}
This computation gives one more example of non-trivial interaction between the Riemann--Roch theorem and the structure of semi-simple Lie algebras (the first one was mentioned in the Introduction). Its idea can be tracked back to \cite{KN_2point} (via \cite{Sch_Laxmulti,Sh_DGr,KSlax}) where a similar computation has been carried out with no relation to any Lie algebra theory.

\subsection{Central extensions}
\label{SS:cexst}
In this section, we will construct almost graded central extensions of the Lie algebra $\L$.

The \emph{central extension} of a Lie algebra $\L$ is a short exact sequence of Lie algebras
\begin{equation}
\begin{CD}
0@>>>\C@>i>>\widehat\L@>p>>\L@>>>0
\end{CD}
\end{equation}
where ${\rm im}(i)=\ker(p)$ is a center of $\widehat\L$. The Lie algebra $\widehat\L$ itself is also often called a central extension of the Lie algebra $\L$.

Two central extensions $\widehat\L$ and $\widehat\L'$  are called equivalent if there exists an isomorphism $e$ ({\it equivalence}) such that the following diagram is commutative:
\begin{equation}
\xymatrix{ &   &   \widehat\L\ar[dr]^p \ar[dd] & &  \\
0\ar[r] & \C\ar[ur]^i\ar[dr]_{i'} &             &   \L\ar[r] & 0. \\
  &  &  \widehat\L'\ar[ur]_{p'} \ar[uu]_{e} &  &
}
\end{equation}

By a \emph{2-cocycle} on $\L$ we mean a bilinear skew-symmetric form $\eta$ satisfying the relation
\begin{equation}\label{E:cocycle}
\eta([f,g],h)+\eta([g,h],f)+\eta([h,f],g)=0,\quad\forall f,g,h\in\L.
\end{equation}
If $\eta(f,g)=\phi([f,g])$ where $\phi\in\L^*$ the $\eta$ is called a
\emph{coboundary} (and denoted by $\d\phi$).
If $\eta-\eta'=\d\phi$ then $\eta$ and $\eta'$ are called {\it
cohomological}.

Every 2-cocycle $\eta$ on $\L$ gives a central extension, and vise versa, by means the relations
$\widehat\L=\L\oplus\C t$ and
\begin{equation}\label{E:centext}
[f,g] =[f,g]_{{}_\L}+\eta(f,g)\cdot t,\quad [t,\L]=0
\end{equation}
where $f,g\in\L$, $t$ is a formal generator of the 1-dimensional space $\C t$, $[\cdot,\cdot]$ is a commutator on $\widehat\L$, and $[\cdot,\cdot]_{{}_\L}$ on $\L$. Two central extensions are equivalent if, and only if their defining  cocycles are cohomological. Thus the equivalence classes of central extensions are in a one-to-one correspondence with the elements of the space $\H^2(\L,\C)$.
\begin{example}
Let $\V=\C^{2n}$ be thought as a commutative Lie algebra, and let
$\eta$ be a symplectic form on $\V$. Then $\widehat\V$ is the Heisenberg algebra.
\end{example}

A central extension is called almost graded if it inherits the almost graded structure from the original Lie algebra while central elements are relegated to the degree~$0$ subspace.

Almost graded central extensions are given by \emph{local cocycles}.
We remind from \cite{KNFa,KSlax,SSlax,Sh_DGr,Sch_Laxmulti} that a 2-cocycle $\eta$ on $\L$ is called local if there exists $M\in\Z_+$ such that for every $m,n\in\Z$, $|m+n|>M$, and every $L\in\L_m$, $L'\in\L_n$ one has $\eta(L,L')=0$.

Let $\langle\cdot ,\cdot\rangle$ denote an invariant symmetric bilinear form on $\g$. In abuse of notation, we use the same symbol to denote a natural extension of this form to $\g$-valued functions and $1$-forms on $\Sigma$. For example, for $L,L'\in{\mathcal L}$ we denote by $\langle L,L'\rangle$ the scalar function on $\Sigma$ taking value $\langle L(P),L'(P)\rangle$ at an arbitrary $P\in\Sigma$.

Finally, let ${\omega}$ be a  $\g_0$-valued 1-form on $\Sigma$ having the expansion of the form
\[
          {\omega}(z)=\left(\frac{h}{z}+{\omega}_0+\ldots\right)dz
\]
at every $\ga\in\Gamma$, where $h\in\h$ is the element giving the grading of $\g$ at $\ga$.

It is our main purpose in this section to give a proof of the following theorem.
\begin{theorem}\label{T:central}
\begin{itemize}
  \item[]
  \item[$1^\circ $] For any $L,L'\in{\mathcal L}$ the 1-form $\langle L,(d - \ad{\omega})L'\rangle$ is holomorphic outside the points
      $\{ P_i\ |\ i=1,\ldots,N \}$ and $\{ Q_j\ |\ j=1,\ldots,M \}$.
  \item[$2^\circ$] For any invariant symmetric bilinear form $\langle\cdot , \cdot\rangle$ on~$\g$, the relation \[\eta(L,L')=\sum\limits_{i=1}^N\res_{P_i}\langle L,(d - \ad{\omega})L'\rangle\] gives a local cocycle on ${\mathcal L}$.
  \item[$3^\circ$] Up to equivalence, almost graded central extensions of the Lie algebra ${\mathcal L}$ are in one-to-one correspondence with invariant symmetric bilinear forms on~$\g$. In partcular, if $\g$ is simple then the central extension given by a cocycle $\eta$ is unique in the class of almost graded central extensions up to equivalence and normalization of the central generator.
\end{itemize}
\end{theorem}
\begin{proof}[Proof of \refT{central},\,$1^\circ$]
In the neighborhood of a point $\ga\in\Gamma$, let
\[
   L=\sum_{p\ge -k} L_pz^p, \quad L'=\sum_{q\ge -k} L'_qz^q, \quad L_p\in{\tilde\g_p},\ L_q\in{\tilde\g_q}
\]
where $\tilde\g_p\subset\g$ be a filtration subspace, i.e. $\tilde\g_p=\bigoplus\limits_{s\le p}\g_s$, and $\g_s$'s are grading subspaces. Then
\[
  dL'=\sum_{q\ge -k} qL'_qz^{q-1}\, dz,
\]
and in the neighborhood of $\ga$
\begin{equation}\label{E:stand}
   \langle L,dL'\rangle=\sum_{p,q\ge -k} q\langle L_p,L'_q\rangle z^{p+q-1}\, dz.
\end{equation}
Observe that $(\ad h)L_p=pL_p+\tilde L_{p-1}$, где $\tilde L_{p-1}\in\tilde\g_{p-1}$. For this reason
\[
\begin{split}
  (\ad{\omega})L' &=\ad(hz^{-1}+{\omega}_0+\ldots)\sum_{q\ge -k} L'_qz^q\, dz= \\
  &=\sum_{q\ge -k}(qL'_q+\tilde L_{q-1})z^{q-1}\, dz +\sum_{\substack{q\ge -k\\l\ge 0\hphantom{ii}}} L_{q,l}z^{q+l}\, dz
\end{split}
\]
where $L_{q,l}=(\ad{\omega}_l)L_q'\in\tilde\g_q$ (since ${\omega}_l\in\g_0$). The summands $\tilde L_{q-1}z^{q-1}$ of the first sum have the same form as the summands of the second sum (with $l=0$). By abuse of notation, we omit them regarding them to as having been shifted into the second sum (where we don't change the notation nevertheless):
\begin{equation*}
 (\ad{\omega})L' = \sum_{q\ge -k} qL'_qz^{q-1}\, dz +\sum_{\substack{q\ge -k\\ l\ge 0\hphantom{ii}}} L_{q,l}z^{q+l}\, dz .
\end{equation*}
Therefore
\begin{equation}\label{E:cobound}
 \langle L,(\ad{\omega})L'\rangle = \sum_{p,q\ge -k} q\langle L_p,L'_q\rangle z^{p+q-1}\, dz +\sum_{\substack{p,q\ge -k\\ l\ge 0\hphantom{i,ii}}}\langle L_p, L_{q,l}\rangle z^{p+q+l}\, dz .
\end{equation}
Subtracting \refE{cobound} from \refE{stand} we obtain
\begin{equation}
 \langle L,(d-\ad{\omega})L'\rangle = -\sum_{\substack{p,q\ge -k\\ l\ge 0\hphantom{i,ii}}}\langle L_p, L_{q,l}\rangle z^{p+q+l}\, dz .
\end{equation}
We will show that the last expression is holomorphic in a neighborhood of $z=0$. Assume that $p+q+l < 0$. Since $l\ge 0$, this would imply $p+q < 0$. By definition, $L_p\in\bigoplus\limits_{i\le p}\g_i$, $L_{q,l}\in\bigoplus\limits_{j\le q}\g_j$. Therefore $i+j\le p+q<0$, hence $\langle \g_i,\g_j\rangle = 0$. This implies $\langle L_p,L_{q,l}\rangle =0$.
\end{proof}
\begin{proof}[Proof of \refT{central},$2^\circ$]
Here we will give a proof of the locality of the cocycle.
The technique we use for that became a kind of standard already \cite{KSlax,SSlax,Sh_DGr,Sh_G2}.

As above, let $D_m$ be defined as follows:
\begin{equation}\label{E:Dm}
  D_m=-m\sum\limits_{i=1}^NP_i+\sum_{j=1}^M(a_jm+b_{m,j})\,Q_i+k\sum_{s=1}^K\ga_s.
\end{equation}
Denote the divisor of a function, or a 1-form, by means of the notation of this function (resp., 1-form) enclosed in brackets. Let $L\in{\mathcal L}_m$, $L'\in{\mathcal L}_{m'}$. Then
\[
  (\langle L,dL'\rangle)\ge (m+m'-1)\sum\limits_{i=1}^NP_i-\sum_{j=1}^M(a_j(m+m')+b_{m,j}+b_{m',j}-1)\,Q_j+D_\Gamma ,
\]
where $D_\Gamma$ is a divisor supported on $\Gamma$.

Assume that $m_1^+,\ldots,m_N^+$, $m_1^-,\ldots,m_M^-$ are such integers that $({\omega})\ge \sum\limits_{i=1}^Nm_i^+P_i-\sum_{j=1}^Mm_j^-Q_j-\sum_{s=1}^K\ga_s$. Then
$
  (\langle L,(\ad{\omega})L'\rangle )\ge  \sum\limits_{i=1}^N(m+m'+m_i^+)P_i-\sum_{j=1}^M(a_j(m+m')+b_{m,j}+b_{m',j}+m_j^-)\,Q_j+D'_\Gamma
$
where, by \refT{central},$1^\circ$, $D_\Gamma-D'_\Gamma\ge 0$.

In order the 1-form $\rho=\langle L,(d-\ad{\omega})L'\rangle$ had a nontrivial residue at least at one of the points $P_i$, it is necessary that
\[
  \min_{i=1,\ldots,N} \{ m+m'-1,m+m'+m_i^+\}\le -1,
\]
in other words,
\begin{equation}\label{E:upperb}
  m+m'\le -1-\min_{i=1,\ldots,N} \{ -1,m_i^+\}.
\end{equation}

Further on, it is necessary that $\sum\limits_{i=1}^N \res_{P_i}\rho\ne 0$, otherwise the cocycle $\eta$ is equal to zero. But then $\sum\limits_{j=1}^M \res_{Q_j}\rho\ne 0$ also, hence $\rho$ has a nontrivial residue at least at one of the points~$Q_j$. It implies that
\[
  \max_{j=1,\ldots,M} \{a_j(m+m')+b_{m,j}+b_{m',j}-1 ,a_j(m+m')+b_{m,j}+b_{m',j}+m_j^-\}\ge 1.
\]
Since $b_{j,m}\le B, \forall j,m$ (as required in \refE{3div}), the last inequality implies that
\[
  \max_{j=1,\ldots,M} \{a_j(m+m')+2B-1 ,a_j(m+m')+2B+m_j^-\}\ge 1,
\]
and further on
\[
  \max_{j=1,\ldots,M} \{a_j(m+m')+2B-1 ,a_j(m+m')+2B+\max_{j=1,\ldots,M}m_j^-\}\ge 1,
\]
then
\[
  \max_{j=1,\ldots,M} \{a_j(m+m')\}\ge 1-\max\{  2B-1\, ,\, 2B+\! \max_{j=1,\ldots,M}m_j^- \},
\]
and finally
\begin{equation}\label{E:lowerb}
   m+m' \ge \min_{j=1,\ldots,M}\{ a_j^{-1}(1-\max\{  2B-1\, ,\, 2B+\! \max_{j=1,\ldots,M}m_j^- \})\}.
\end{equation}
By \refE{upperb} and \refE{lowerb} we conclude that $\eta(L,L')$ is a local cocycle.
\end{proof}
\begin{remark}
By $\g$-invariance and symmetry of the bilinear form
\[\sum_{i=1}^N\res_{P_i}\langle L,(\ad{\omega})L'\rangle=-\sum_{i=1}^N\res_{P_i}\langle {\omega},[L,L']\rangle .
\]
Therefore the corresponding part of the cocycle $\eta$ is a coboundary.  Hence \refT{central},$2^\circ$ can be formulated in the following form: the standard cocycle given by the 1-form $\langle L,dL'\rangle$ is local up to a coboundary.
\end{remark}

\begin{proof}[Proof of \refT{central},$3^\circ$] For a proof of uniqueness  in the case when $\g$ is simple, we refer to \cite{SSlax} where it is given for an arbitrary simple Lie algebra $\g$ in the 2-point case ($N=M=1$), and to \cite{Sch_Laxmulti} where it is given for arbitrary $N$ and $M$. The remainder of the assertion $3^\circ$ easily follows from this result.
\end{proof}

\subsection{Examples: gradings of classical Lie algebras and $G_2$} \label{SS:ex_grad}
Here, we will consider gradings of depth 1 and 2 of classical Lie algebras, and of depth 2 and 3 for $G_2$. In these cases we will reproduce the expansions of Lax operators obtained earlier in \cite{Klax,KSlax,Sh_DGr,Sh_G2}, and then (in \refSS{Tyupar}) their Tyurin parameters. We also will point out new Lax operator algebras arising in these cases. We will focus on the gradings given by simple roots. For a simple root $\a_i$ it is a grading given by the element $h\in\h$ such that $\a_i(h)=1$, and $\a_j(h)=0$ ($j\ne i$). Therefore the grading subspace $\g_p$ is a direct sum of the root subspaces $\g_\a$ such that $\a_i$ is contained in the expansion of $\a$ over simple roots with multiplicity $p$. There is a dual grading  when this multiplicity is set to $-p$. Observe that the depth of the grading given by a simple root, is equal to its multiplicity in the expansion of the highest root. Below, we keep the following conventions: $\g_i\subset\g$ is the eigenspace with the eigenvalue $(-i)$; in the figures, the lines going beyond the matrices, correspond to their medians. For the detailed information on $\Z$-gradings of semi-simple Lie algebras we refer to \cite[Ch.2,\ \S 3.5]{Vin}. Below, $e_1,\ldots,e_n$ denote the elements of the orthonormal basis in the space the root system is embedded in, with respect to the Cartan--Killing form.

\subsubsection{The case of $A_n$}\label{SS:An}

In this case $\g$ has $\left[\frac{n}{2}\right]$ gradings of depth 1 (and has no grading of depth 2 and more, given by simple roots). The grading number $r$ ($1\le r\le\left[\frac{n}{2}\right]$) is given by assigning the degree $-1$ to the simple root $\a_r=e_r-e_{r+1}$. To begin with, consider the grading number $1$.
The block structure
\begin{figure}[h]           
\begin{picture}(0,0)
\put(165,57){\text{--}\ $\g_0$}
\put(165,74){\text{--}\ $\g_{-1}$}
\put(165,38){\text{--}\ $\g_{1}$}
\put(67,-10){\text{a}}
\end{picture}
  \includegraphics[width=6cm]{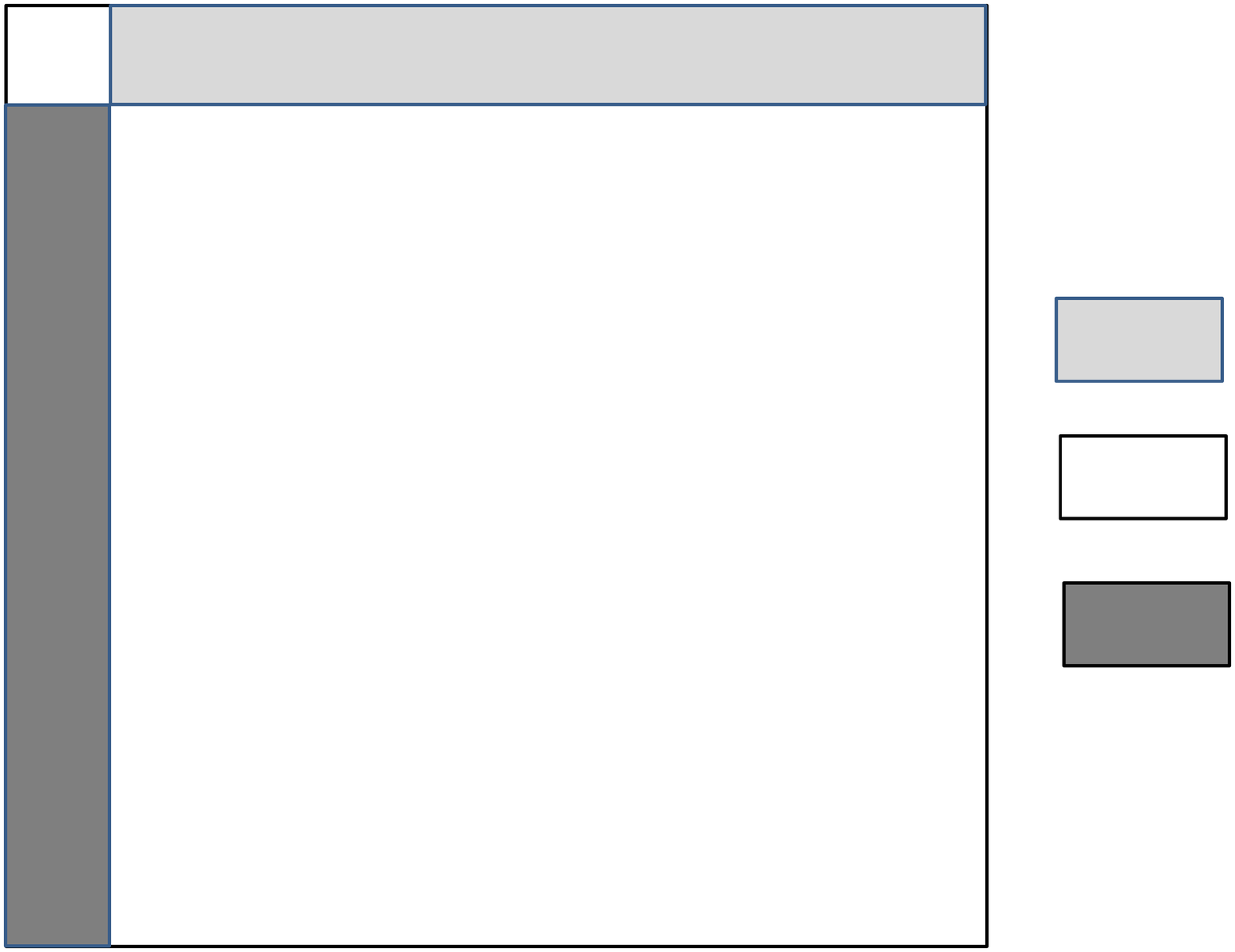}
\begin{picture}(30,0)
\end{picture}
\includegraphics[width=6cm]{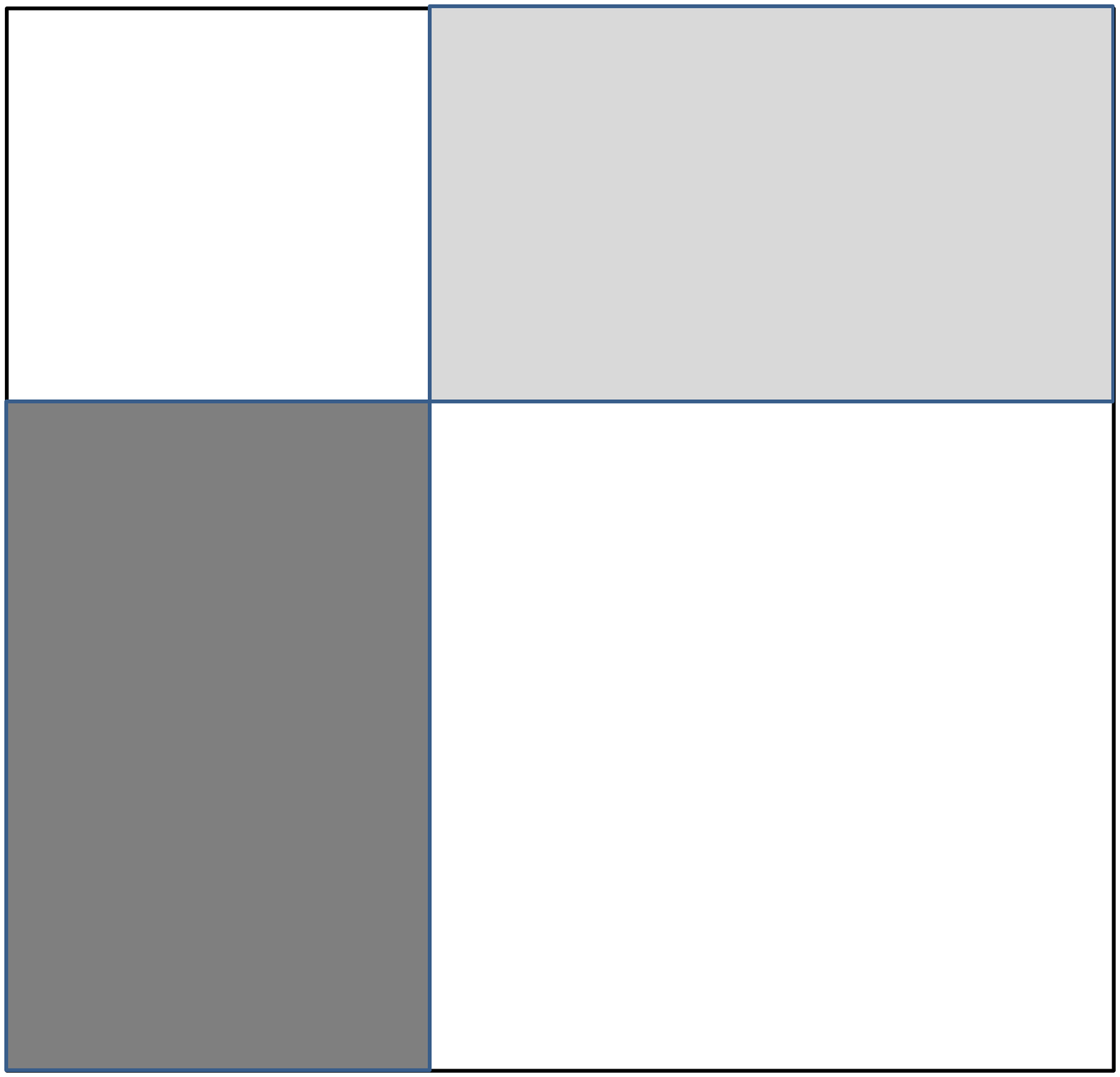}
\begin{picture}(0,0)
\put(-112,-10){\text{b}}
\end{picture}
  \caption{Case $A_n$}\label{An}
\end{figure}
of grading subspaces in this case is drown in Figure \ref{An},a. The matrices corresponding to the subspace $\g_{-1}=\tilde\g_{-1}$ can be represented in the form $\a\b^t$ where $\a\in\C^n$, the transposed vector to $\a$ is $\a^t=(1,0,\ldots,0)$, and $\b\in\C^n$ is arbitrary.
Such matrix belongs to $\sln(n)$ if $\b^t\a=0$. Next, an element $L_0\in\g$ belongs to the filtration subspace $\tilde\g_0=\g_{-1}\oplus\g_0$ if, and only if $\a$ is its eigenvector. Since all elements defining the grading on $\g$ are given up to an inner automorphism, $\a=g(1,0,\ldots,0)^t$ ($g\in GL(n)$) can take an arbitrary value in $\C^n$, while the above relations between $\a$, $\b$ and $L_0$ are preserved. This way, we arrive to the following expansion of the Lax operator for $\sln(n)$ at the point $\ga\in\Gamma$, first proposed in \cite{Klax} (see also \cite{KSlax,Sh_DGr}):
\begin{equation}\label{E:pTgln}
   L(z)=\a\b^tz^{-1}+L_0+\ldots ,
\end{equation}
where $\b^t\a=0$, and there exists $\varkappa\in\C$ such that $L_0\a=\varkappa\a$.

The grading number $r$ is given by simple root $\a_r=e_r-e_{r+1}$. The corresponding matrix realization is represented in Figure \ref{An},b. This root is contained in the expansions of the the roots $e_i-e_j$ for $i=1,\dots,r$, $j=r+1\ldots,n$. The sum of the corresponding root subspaces gives the grading subspace $\g_{-1}$. Therefore the upper block $\g_0$ has dimension $r\times r$, blocks $\g_{-1}$ and $\g_1$ have dimensions $r\times (n-r)$ and $(n-r)\times r$, respectively. The subspace $\g_{-1}$ consists of matrices of the form $\tilde\a_1\b_1^t+\ldots + \tilde\a_r\b_r^t$ where $\tilde\a_i=(0,\ldots,0,1,0,\ldots,0)$ ($1$ in $i$th positions). Vectors $\tilde\a_i$ and $\b_j$ are mutually orthogonal, the linear span of the vectors $\tilde\a_1,\ldots,\tilde\a_r$ is an invariant subspace of the subalgebra $\tilde\g_0$.

\subsubsection{The case of $D_n$}\label{S:Dn}
The Dynkin diagram $D_n$ has the form:

\begin{figure}[h]  
\begin{picture}(100,45)
\put(-40,10){
\begin{picture}(100,30)
\put(0,10){\circle*{3}}
\put(0,10){\line(1,0){30}}
\put(30,10){\circle*{3}}
\put(30,10){\line(1,0){30}}
\put(60,10){\circle*{3}}
\put(60,10){\line(1,0){15}}
\put(100,10){\line(1,0){15}}
\put(115,10){\circle*{3}}
\put(115,10){\line(1,0){30}}
\put(145,10){\circle*{3}}
\put(145,10){\line(1,1){20}}
\put(165,30){\circle*{3}}
\put(145,10){\line(1,-1){20}}
\put(165,-10){\circle*{3}}
\put(80,9){$\ldots$}
\put(-5,0){$\a_1$}
\put(25,0){$\a_2$}
\put(150,7){$\a_{n-2}$}
\put(170,27){$\a_{n-1}$}
\put(170,-13){$\a_n$}
\end{picture}   }
\end{picture}
\end{figure}
where
\[
  \a_1=e_1-e_2,\ \ldots, \ \a_{n-1}=e_{n-1}-e_n,\ \a_n=e_{n-1}+e_n,
\]
and the full set of positive roots is given by the following list: $e_i\pm e_j$, $1\le i<j\le n$. We will need the expressions of the positive roots via simple ones, they are as follows:
\[
  e_i-e_j=\a_i+\ldots+\a_{j-1}\quad (1\le i<j\le n),
\]
and
\[
e_i+e_j=\left\{
   \begin{array}{ll}
            \a_i+\ldots+\a_{j-1}+2\a_j+\ldots+2\a_{n-2}+\a_{n-1}+\a_n, & \hbox{$i<j\le n-2$;} \\
            \a_i+\ldots+\a_{n-1}+\a_n, & \hbox{$i<j= n-1$;} \\
            \a_i+\ldots+\a_{n-2}+\a_n, & \hbox{$i\le n-2, j=n$;} \\
            \a_n, & \hbox{$i=n-1, j=n$.}
   \end{array}
        \right.
\]
The highest root is equal to $\theta=\a_1+2\a_2+\ldots+2\a_{n-2}+\a_{n-1}+\a_n$, hence there are three simple roots $\a_1$, $\a_{n-1}$ and $\a_n$ giving gradings of depth~$1$, but the last two are equivalent under an outer automorphism.

Consider the grading corresponding to the simple root $\a_1$. This simple root is contained in the expansions of the roots $e_1\pm e_j$, $j=2,\ldots,n$. The corresponding root subspaces form the grading subspace $\g_{-1}$. The blocks corresponding to the matrix realization of $\g=\so(2n)$ with respect to the quadratic form $\s=\begin{pmatrix}
          0 & E \\
          E & 0 \\
    \end{pmatrix}
$ are represented in Figure~\ref{Dn},a:
\begin{figure}[h]           
\begin{picture}(0,0)
\put(165,57){\text{--}\ $\g_0$}
\put(165,74){\text{--}\ $\g_{-1}$}
\put(165,38){\text{--}\ $\g_{1}$}
\put(67,-10){\text{a}}
\end{picture}
  \includegraphics[width=6cm]{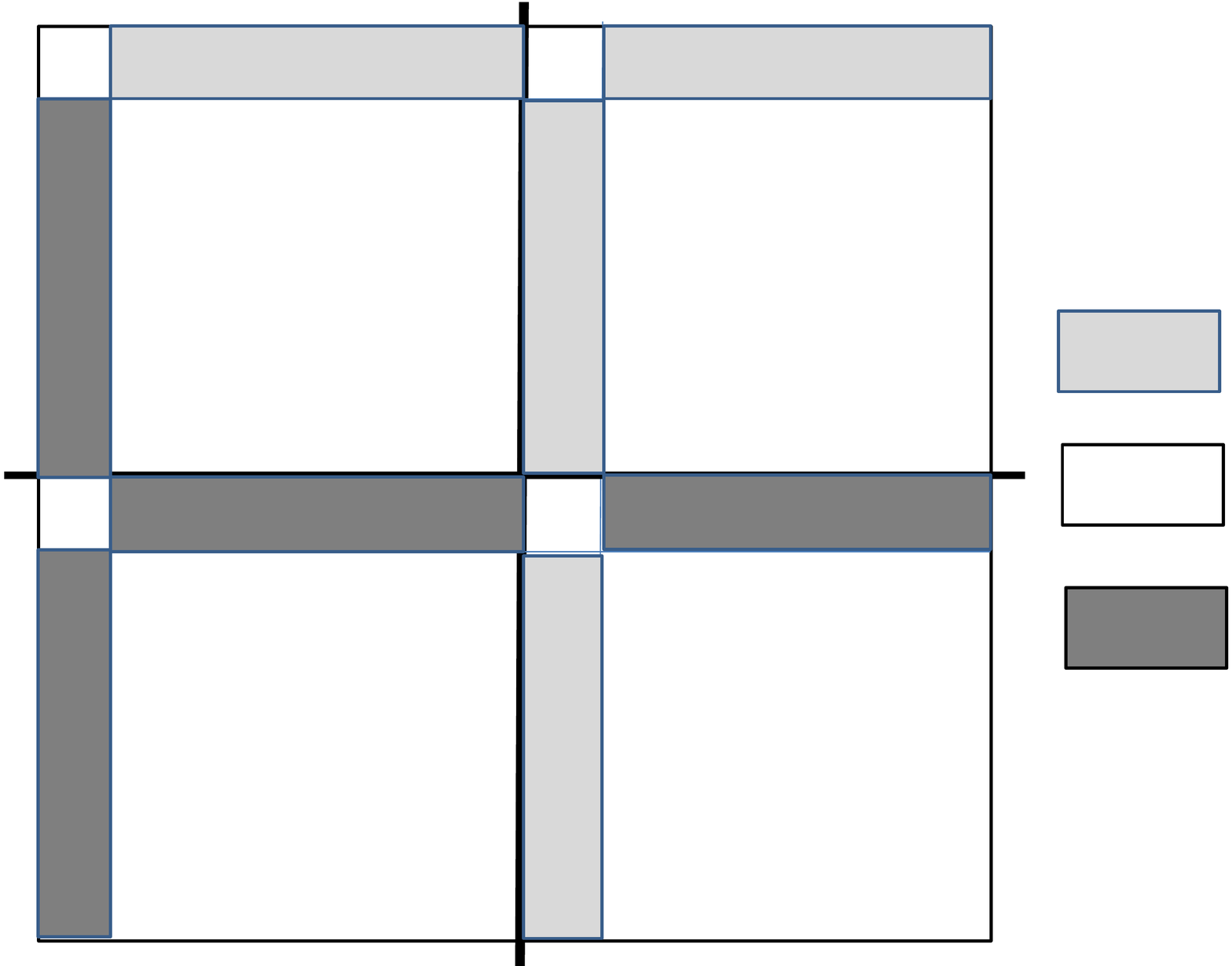}
\begin{picture}(30,0)
\put(-167,53){\small $0$}
\put(-105,109){\small $0$}
\end{picture}
\includegraphics[width=6cm]{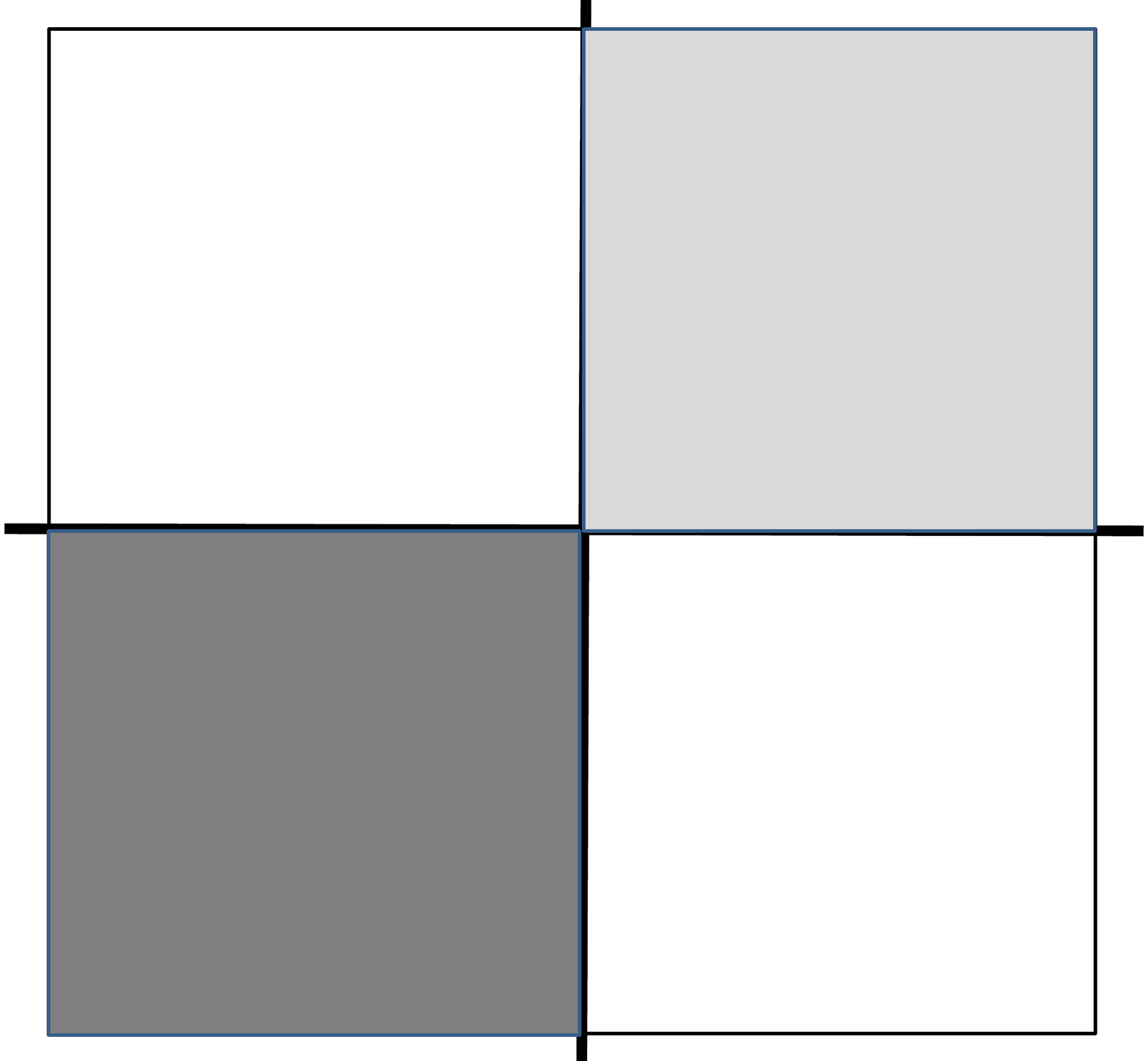}
\begin{picture}(0,0)
\put(-112,-10){\text{b}}
\end{picture}
  \caption{Case of $D_n$}\label{Dn}
\end{figure}

Observe that $\g_{-1}$ can be represented as the space of rank 2 matrices of the form $(\a\b^t-\b\a^t)\s$ where $\a,\b\in\C^{2n}$, $\a^t=(1,0,\ldots,0)$, and $\b^t\s\a=0$. Observe also that $\a$ is an eigenvector of the subalgebra $\tilde\g_0$, and $\a^t\s\a=0$. Hence we obtain the following expansion (first found out in \cite{KSlax}) for $L$ at $\ga\in\Gamma$:
\begin{equation}\label{E:pTso2n}
   L(z)=(\a\b^t-\b\a^t)\s z^{-1}+L_0+\ldots
\end{equation}
where $\a$, $\b$ satisfy the above conditions, and there exists a $\varkappa\in\C$ such that $L_0\a=\varkappa\a$.

Considering next the grading given by the simple root $\a_n$, we see that this root is contained in the expansions of the positive roots $e_i+e_j$ for $i<j$. The remainder of simple roots forms the Dynkin diagram $A_{n-1}$. The blocks corresponding to the grading subspaces are represented in Figure \ref{Dn},b. In particular, the subspace $\g_{-1}$ is represented by matrices of the form $(\tilde\a_1\b_1^t-\b_1\tilde\a_1^t)\s+\ldots+(\tilde\a_n\b_n^t-\b_n\tilde\a_n^t)\s$, where $\tilde\a_i,\b_i\in\C^{2n}$ for $i=1,\ldots,n$, vector $\tilde\a_i$ is given by its entries $\tilde\a_i^j=\d_i^j$ (where $j=1,\ldots,2n$, $\d_i^j$ is the Kronecker symbol), $\b_i=(\b_i^1,\ldots,\b_i^n,0,\ldots,0)$ (wher $\b_i^j\in\C$ are arbitrary). Observe that $\tilde\a_i^t\s\b_j=\tilde\a_i^t\s\tilde\a_j=\b_i^t\s\b_j=0$ for all $i,j=1,\ldots,n$. Using these relations, we can independently show by methods of  \cite{KSlax} that that the property of $L$ to
have simple poles at the points $\ga\in\Gamma$ is conserved under the commutator.
\subsubsection{The case of $C_n$}\label{S:Cn}
The Dynkin diagram $C_n$ is as follows:
\begin{figure}[h]
\begin{picture}(100,40)
\put(-40,10){
\begin{picture}(100,30)
\put(0,10){\circle*{3}}
\put(0,10){\line(1,0){30}}
\put(30,10){\circle*{3}}
\put(30,10){\line(1,0){30}}
\put(60,10){\circle*{3}}
\put(60,10){\line(1,0){15}}
\put(100,10){\line(1,0){15}}
\put(115,10){\circle*{3}}
\put(115,10){\line(1,0){30}}
\put(145,10){\circle*{3}}
\put(145,11){\line(1,0){30}}
\put(175,10){\circle*{3}}
\put(145,9){\line(1,0){30}}
\put(175,10){\line(-2,1){10}}
\put(175,10){\line(-2,-1){10}}
\put(80,9){$\ldots$}
\put(-5,0){$\a_1$}
\put(25,0){$\a_2$}
\put(130,0){$\a_{n-1}$}
\put(170,0){$\a_n$}
\end{picture}   }
\end{picture}
\end{figure}
\newline where
\[
  \a_1=e_1-e_2,\ \ldots, \ \a_{n-1}=e_{n-1}-e_n,\ \a_n=2e_n,
\]
and all positive roots are given by the following list: $e_i\pm e_j$, $1\le i<j\le n$, and $2e_i$ ($i=1,\ldots,n$).
The expressions of the positive roots via simple ones are as follows:
\[
   \begin{array}{lll}
       e_i-e_j&=\a_i+\ldots+\a_{j-1}, & \hbox{$1\le i<j\le n$;} \\
       2e_i&=2\a_i+\ldots\hphantom{+\a_{j-1}}+2\a_{n-1}+\a_n, & \hbox{$i=1,\ldots,n-1$;}\\
       2e_n&=\hphantom{2\a_i+\ldots+\a_{j-1}+2\a_{n-1}+\ }\a_n,
   \end{array}
\]
and
\[
e_i+e_j=\left\{
   \begin{array}{ll}
            \a_i+\ldots+\a_{j-1}+2\a_j+\ldots+2\a_{n-1}+\a_n, & \hbox{$i<j\le n-1$;} \\
            \a_i+\ldots+\a_{n-1}+\a_n, & \hbox{$i<j= n$;}.
   \end{array}
        \right.
\]
The highest root is equal to $\theta=2e_1=2\a_1+2\a_2+\ldots+2\a_{n-1}+\a_n$. We will consider here the grading of depth $2$ given by the root $\a_1$, and the grading of depth $1$ given by the root $\a_n$. The Lax operator algebra corresponding to the first of them has been found out in \cite{KSlax}, see also \cite{Sh_DGr}. The algebra corresponding to the second is found out in \cite{Lax_Grad}.

In the matrix realization of the Lie algebra $\g=\spn(2n)$ with respect to the symplectic form
$\s=\begin{pmatrix}
          0 & E \\
          -E & 0 \\
    \end{pmatrix}
$, the blocks corresponding to the grading subspaces are represented in Figure~\ref{Cn},a.
\begin{figure}[h]          
\begin{picture}(0,0)
\put(165,92){\text{--}\ $\g_{-2}$}
\put(165,21){\text{--}\ $\g_2$}
\put(165,57){\text{--}\ $\g_0$}
\put(165,74){\text{--}\ $\g_{-1}$}
\put(165,38){\text{--}\ $\g_{1}$}
\put(67,-10){\text{a}}
\end{picture}
  \includegraphics[width=6cm]{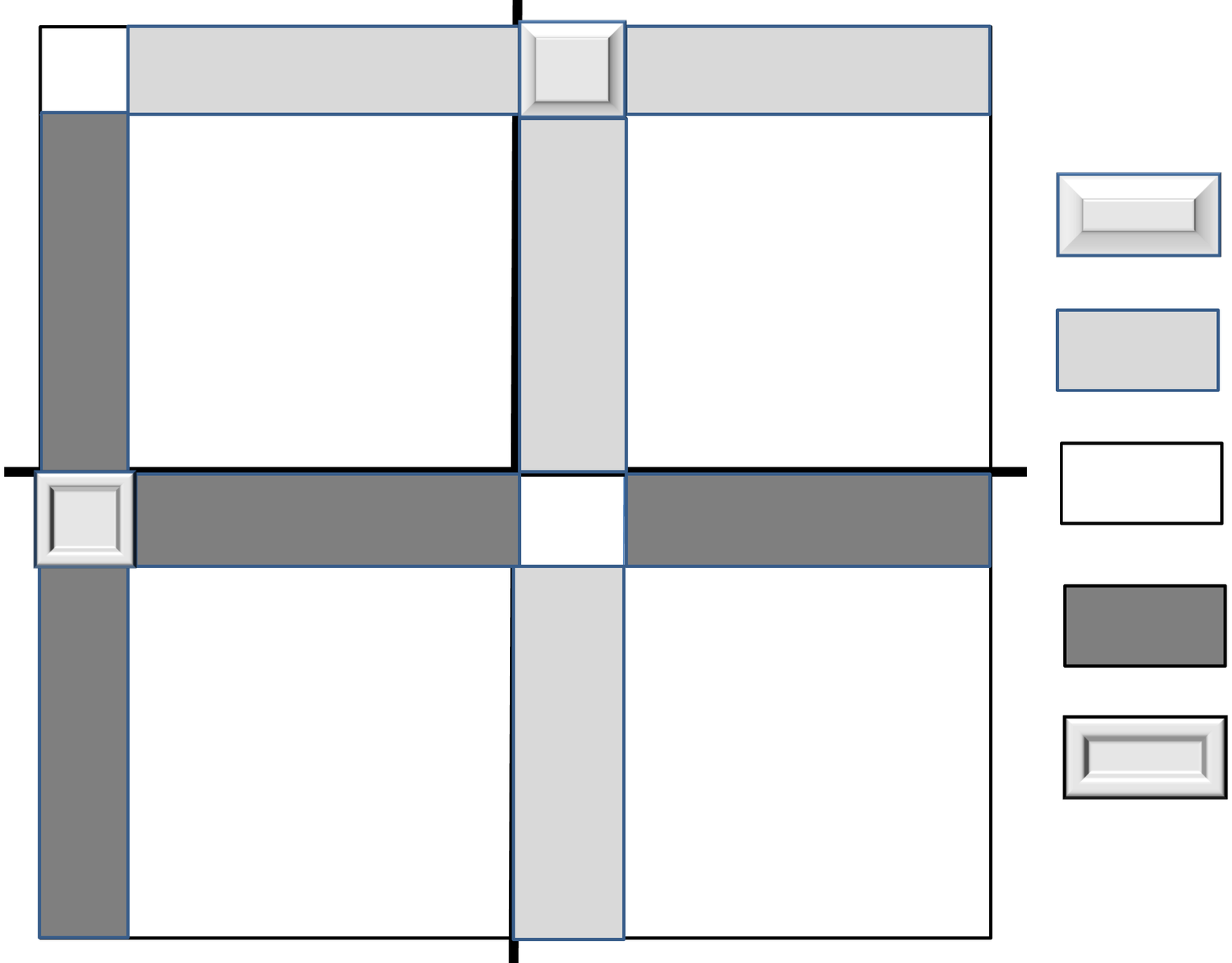}
\begin{picture}(30,0)
\end{picture}
\includegraphics[width=6cm]{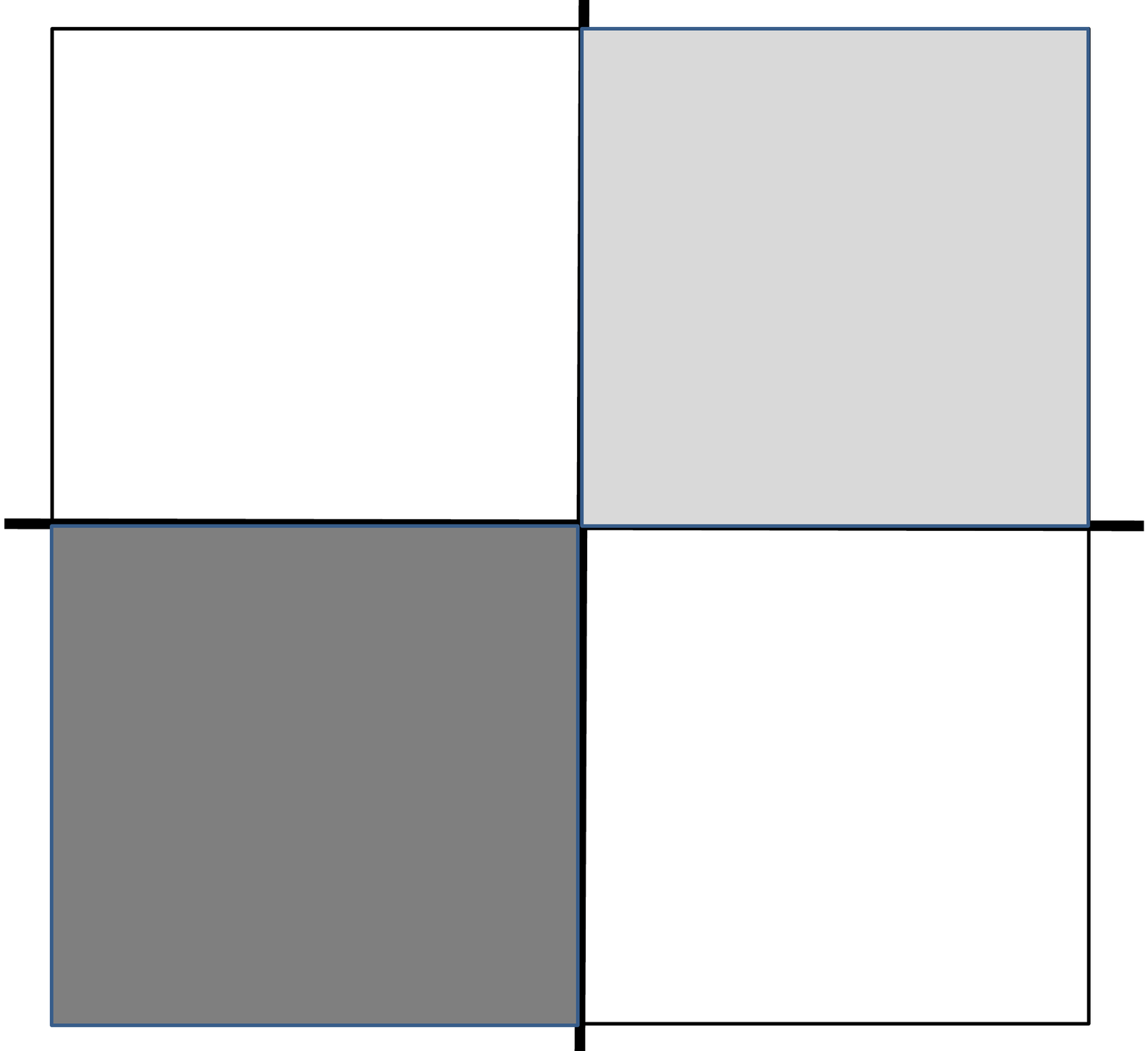}
\begin{picture}(0,0)
\put(-112,-10){\text{b}}
\end{picture}
  \caption{Case $C_n$}\label{Cn}
\end{figure}
In particular, the 1-dimensional subspace $\g_{-2}$, corresponding to the highest root, consists of the matrices of the form $\nu\a\a^t\s$ where $\a\in\C^{2n}$, $\a^t=(1,0,\ldots,0)$, $\nu\in\C$. The subspace $\g_{-1}$ consists of the matrices of the form $(\a\b^t+\b\a^t)\s$ where $\a,\b\in\C^{2n}$, $\a^t=(1,0,\ldots,0)$, and $\b^t\s\a=0$. Observe also that $\a$ is an eigenvector of the subalgebra $\tilde\g_0$, i.e. for every $L_0\in\tilde\g_0$ there exists $\varkappa\in\C$ such that $L_0\a=\varkappa\a$. Observe finally that for every $L_1\in\tilde\g_1$ we have $\a^t\s L_1\a=0$. Therefore we arrive to the following form of the expansion of the element $L$ at $\ga\in\Gamma$:
\begin{equation}\label{E:pTsp2n}
   L(z)=\nu\a\a^t\s z^{-2}+(\a\b^t+\b\a^t)\s z^{-1}+L_0+L_1z+\ldots
\end{equation}
where $\a$, $\b$, $L_0$, $L_1$ satisfy the above relations (see also \cite{KSlax,Sh_DGr}).

The matrix realization of the grading given by the simple root $\a_n$ is given in Figure \ref{Cn},b. It is quite similar to the one for $D_n$, with the only distinction that the matrices in $\g_{-1}$ have the form $(\tilde\a_1\b_1^t+\b_1\tilde\a_1^t)\s+\ldots+(\tilde\a_n\b_n^t+\b_n\tilde\a_n^t)\s$.
As earlier, $\tilde\a_i$ and $\b_j$ satisfy the orthogonality relations which enable one to independently derive by methods of \cite{KSlax} that the simple poles at the points $\ga\in\Gamma$ keep simple after commutation (like in \refS{Dn}).


\subsubsection{The case of $B_n$}\label{S:Bn}
The Dynkin diagram $B_n$ has the following form:
\begin{figure}[h]
\begin{picture}(100,45)
\put(-40,10){
\begin{picture}(100,30)
\put(0,10){\circle*{3}}
\put(0,10){\line(1,0){30}}
\put(30,10){\circle*{3}}
\put(30,10){\line(1,0){30}}
\put(60,10){\circle*{3}}
\put(60,10){\line(1,0){15}}
\put(100,10){\line(1,0){15}}
\put(115,10){\circle*{3}}
\put(115,10){\line(1,0){30}}
\put(145,10){\circle*{3}}
\put(145,11){\line(1,0){30}}
\put(175,10){\circle*{3}}
\put(145,9){\line(1,0){30}}
\put(145,10){\line(2,1){10}}
\put(145,10){\line(2,-1){10}}
\put(80,9){$\ldots$}
\put(-5,0){$\a_1$}
\put(25,0){$\a_2$}
\put(133,0){$\a_{n-1}$}
\put(170,0){$\a_n$}
\end{picture}   }
\end{picture}
\end{figure}
\newline where
\[
  \a_1=e_1-e_2,\ \ldots, \ \a_{n-1}=e_{n-1}-e_n,\ \a_n=e_n,
\]
and all positive roots are given by the following list: $e_i\pm e_j$, $1\le i<j\le n$, and $e_i$ ($i=1,\ldots,n$).
The expressions of the positive roots via simple ones are as follows:
\[
   \begin{array}{lll}
       e_i-e_j&=\a_i+\ldots+\a_{j-1}, &  \\
       e_i&=\a_i+\ldots\hphantom{+\a_{j-1}}+\a_{n-1}+\a_n, & \hbox{$1\le i<j\le n$;}\\
       e_i+e_j&=\a_i+\ldots+\a_{j-1}+2\a_j+\ldots+2\a_{n-1}+2\a_n,&
   \end{array}
\]
The highest root is as follows: $\theta=e_1+e_2=\a_1+2\a_2+\ldots+2\a_n$. We will consider here the grading of depth 1 given by the simple root $\a_1$, and the grading of depth $2$ corresponding to the simple root $\a_n$. The Lax operator algebra corresponding to the first of them has been found out in \cite{KSlax}, see also \cite{Sh_DGr}. The algebra corresponding to the second one is found out in \cite{Lax_Grad}.

In the matrix realization of the Lie algebra $\g=\so(2n+1)$ corresponding to the quadratic form
$\s~=~\begin{pmatrix}
          0 & 0 & E \\
          0 & 1 & 0 \\
          E & 0 & 0 \\
    \end{pmatrix} ,
$ the blocks corresponding to the grading subspaces  are represented in Figure \ref{Bn},a.
\begin{figure}[h] 
\begin{picture}(0,0)
\put(165,92){\text{--}\ $\g_{-2}$}
\put(165,21){\text{--}\ $\g_2$}
\put(165,57){\text{--}\ $\g_0$}
\put(165,74){\text{--}\ $\g_{-1}$}
\put(165,38){\text{--}\ $\g_{1}$}
\put(67,-10){\text{a}}
\end{picture}
  \includegraphics[width=6cm]{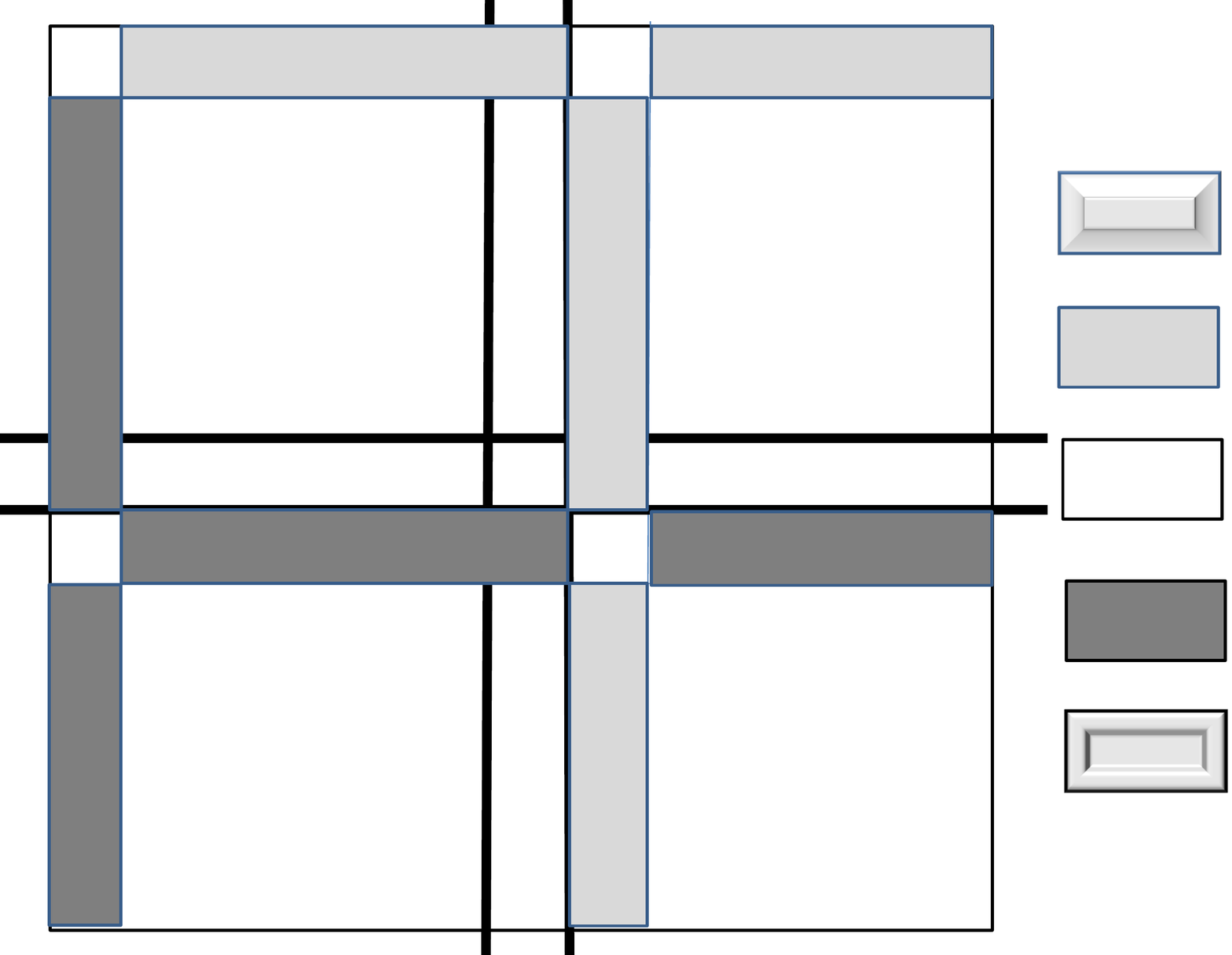}
\begin{picture}(30,0)
\put(-101,109){\small $0$}
\put(-111,57){\small $0$}
\put(-167,48){\small $0$}
\end{picture}
\includegraphics[width=6cm]{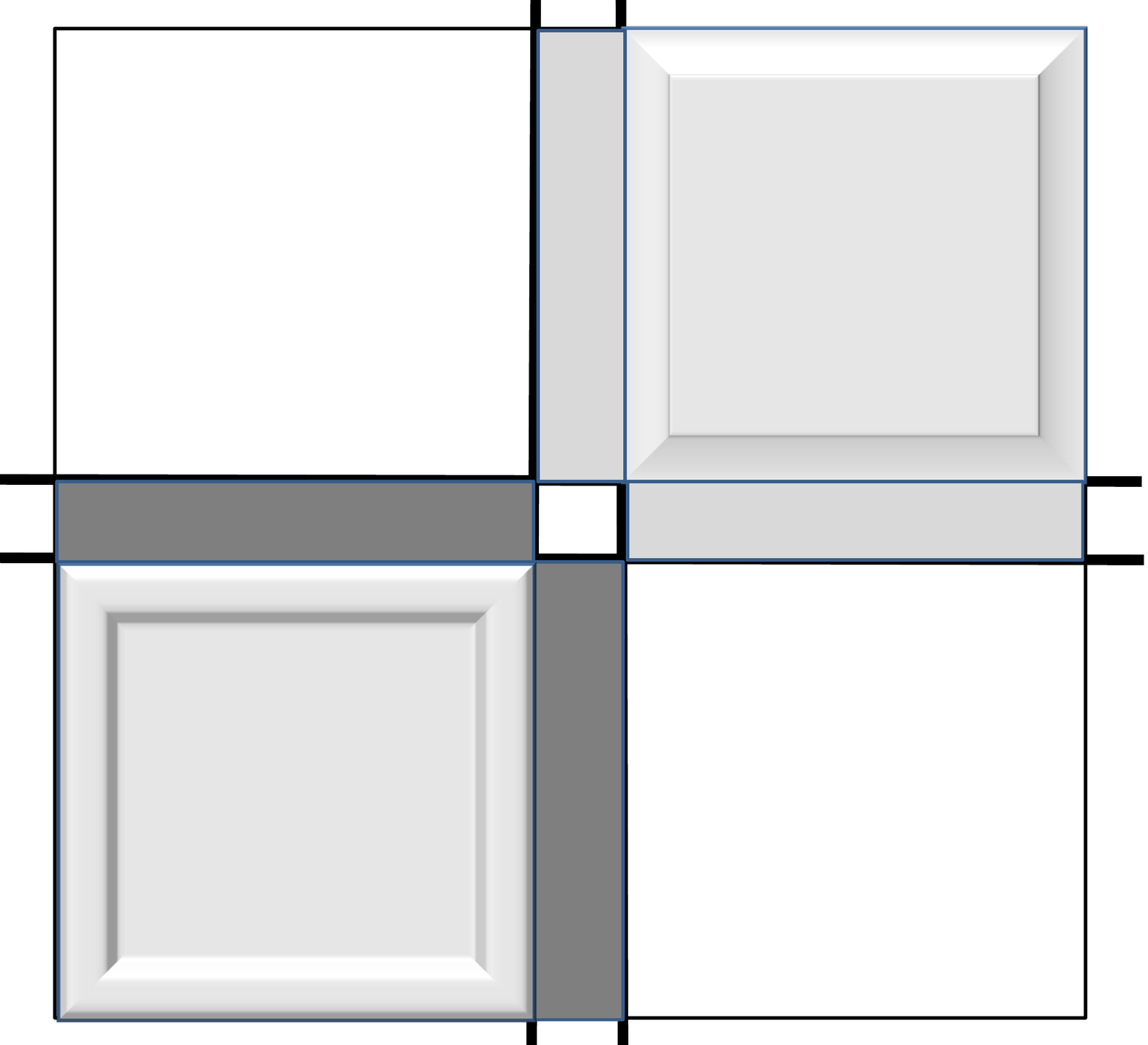}
\begin{picture}(0,0)
\put(-111,57){\small $0$}
\put(-112,-10){\text{b}}
\end{picture}
  \caption{Case $B_n$}\label{Bn}
\end{figure}

The subspace $\g_{-1}$ consists of matrices of the form $(\a\b^t-\b\a^t)\s$ where $\a,\b\in\C^{2n+1}$, $\a^t=(1,0,\ldots,0)$, and $\b^t\s\a=0$ (observe that also $\a^t\s\a=0$). As above, $\a$ is an eigenvector of the subalgebra $\tilde\g_0$. Therefore we arrive to the following form of Laurent expansions of the element $L$ at $\ga\in\Gamma$:
\begin{equation}\label{E:pTso2n+1}
   L(z)=(\a\b^t-\b\a^t)\s z^{-1}+L_0+\ldots
\end{equation}
where $\a$, $\b$ satisfy the just listed relations, and there exists $\varkappa\in\C$ such that $L_0\a=\varkappa\a$ (see also \cite{KSlax,Sh_DGr}).

The matrix realization of the grading given by the simple root $\a_n$ is presented in Figure \ref{Bn},b. The subspace $\g_{-1}$ is a direct sum of the root subspaces of the roots $e_i$ ($i=1,\ldots,n$). The subspace $\g_{-2}$ is a direct sum of the root subspaces of the roots $e_i+e_j$ for all $i,j=1,\ldots,n$. The matrices from $\g_{-1}$ have the form $\tilde\a_0\b_0^t-\b_0\tilde\a_0^t$ where $\tilde\a_0,\b_0\in\C^{2n+1}$, $(\tilde\a_0)_i=(\d_{i,n+1})$, $\b_0=(\b_0^1,\ldots,\b_0^n,0,\ldots,0)$ ($\b_0^j\in\C$ are arbitrary). The matrices from $\g_{-2}$ have the form $(\tilde\a_1\b_1^t-\b_1\tilde\a_1^t)\s+\ldots+(\tilde\a_n\b_n^t-\b_n\tilde\a_n^t)\s$
where $\tilde\a_i,\b_i\in\C^{2n+1}$ for $i=1,\ldots,n$, vector $\tilde\a_i$ is given by its coordinates $\tilde\a_i^j=\d_i^j$ (where $j=1,\ldots,2n+1$, $\d_i^j$ is the Kronecker symbol), $\b_i=(\b_i^1,\ldots,\b_i^n,0,\ldots,0)$ (where $\b_i^j\in\C$ is arbitrary).

\subsubsection{The case of $G_2$}\label{S:G2}
The Dynkin diagram $G_2$ is as follows:
\begin{figure}[h]
\begin{picture}(100,45)
\put(-40,10){
\begin{picture}(100,30)
\put(0,10){\circle*{4}}
\put(0,10){\line(1,0){30}}
\put(0,12){\line(1,0){30}}
\put(0,8){\line(1,0){30}}
\put(30,10){\circle*{4}}
\put(20,10){\line(-2,-1){10}}
\put(20,10){\line(-2,1){10}}

\put(-5,0){$\a_1$}
\put(25,0){$\a_2$}
\end{picture}   }
\end{picture}
\end{figure}
\newline where $\a_1,\a_2\in\C^2$,
\[
  \a_1=(1,0),\  \ \a_2=(-3/2,\sqrt{3}/2),
\]
and all positive roots are located in the vertices of two regular hexagons with common center at $(0,0)$, having vertices at the ends of the vectors $\a_1$, $\a_2$, respectively. The Lie algebra $G_2$ has an exact $7$-dimensional representation by the matrices of the form presented in the Figure \ref{G2},b. In the figure, the dependent blocks have the same color (bright gray, dark gray or white). By $[x]$ (where $x\in\C^3$, $x^T=(x_1,x_2,x_3)$) we denote the skew-symmetric matrix $[x]=\left(\begin{smallmatrix} 0&x_3&-x_2\\
-x_3&0&x_1\\
x_2&-x_1&0\end{smallmatrix}\right)$. Below, we give the full list of positive roots, and their correspondence with matrix elements of the $7$-dimensional representation. For every positive root, we give the value of an entry corresponding to this root, and the list of all other corresponding entries (in the form $(i,j)$).
\[
   \begin{array}{lll}
       \a_1, & (a_1)_1=\sqrt{2}, & (2,1),(1,5),(6,4), \\
       \a_2, &  A_{21}=1,& (3,2),(5,6),\\
       \a_1+\a_2, & (a_1)_2=\sqrt{2},  &  (3,1),(1,6),(5,4),  \\
      2\a_1+\a_2, & (a_2)_3=\sqrt{2},   &  (7,1),(1,4),(2,6),  \\
      3\a_1+\a_2, & A_{13}=1,  &  (2,4),(7,5),  \\
     3\a_1+2\a_2, & A_{23}=1,  &  (3,4),(7,6).
   \end{array}
\]

The highest root is equal to $\theta=3\a_1+2\a_2$. Consider first the grading of depth $2$ given by the simple root $\a_2$. The blocks, corresponding to the grading subspaces in the matrix realization, are presented in Figure \ref{G2},a.
\begin{figure}[h]  
\begin{picture}(0,0)
\put(165,92){\text{--}\ $\g_{-2}$}
\put(165,21){\text{--}\ $\g_2$}
\put(165,57){\text{--}\ $\g_0$}
\put(165,74){\text{--}\ $\g_{-1}$}
\put(165,38){\text{--}\ $\g_{1}$}
\end{picture}
 \includegraphics[width=6cm]{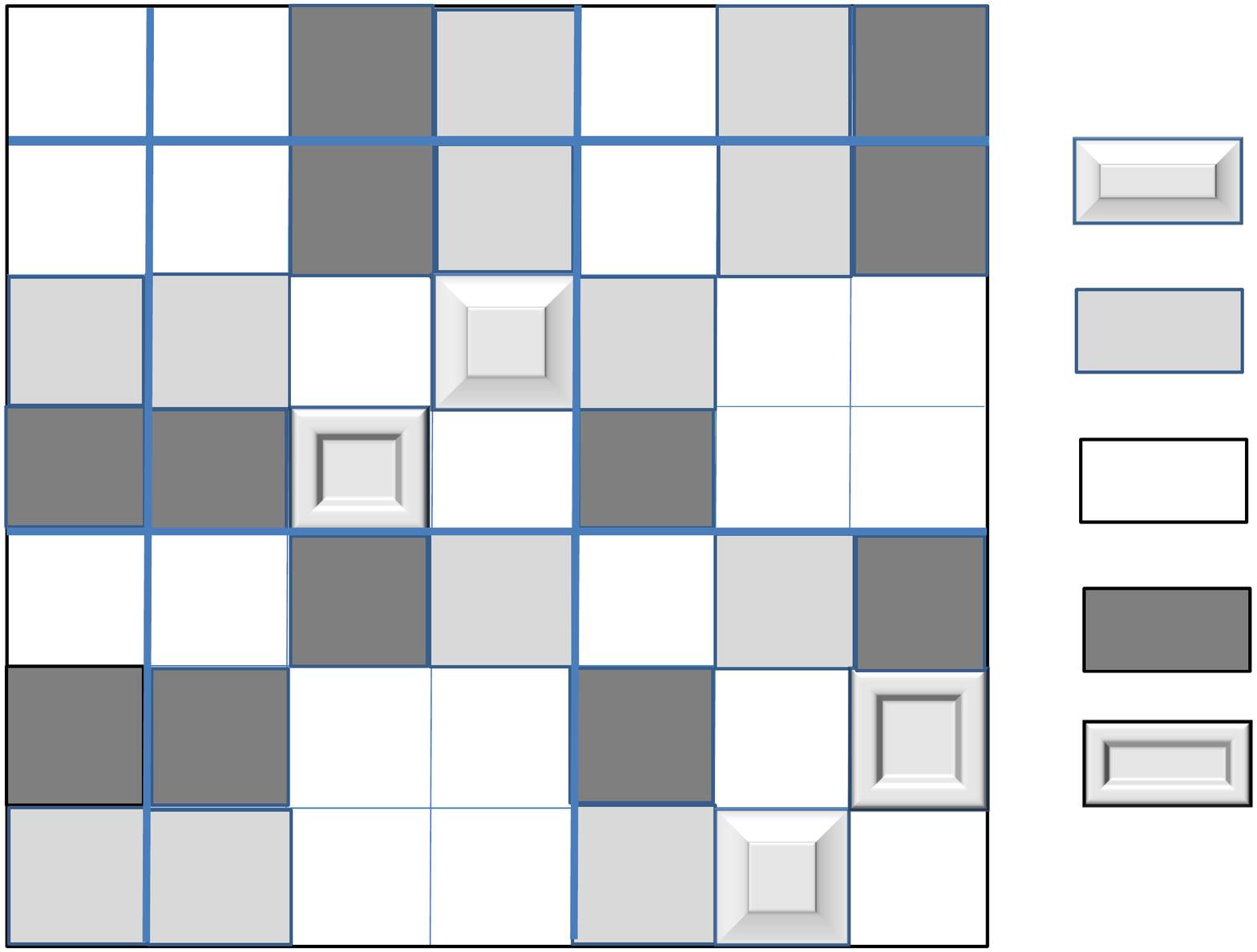}
\begin{picture}(30,0)
\put(-163,105){\small $0$}
\put(-95,90){\small $0$}
\put(-79,75){\small $0$}
\put(-60,57){\small $0$}
\put(-147,42){\small $0$}
\put(-130,25){\small $0$}
\put(-112,8){\small $0$}
\put(-105,-10){\text{a}}
\end{picture}
\includegraphics[width=6cm]{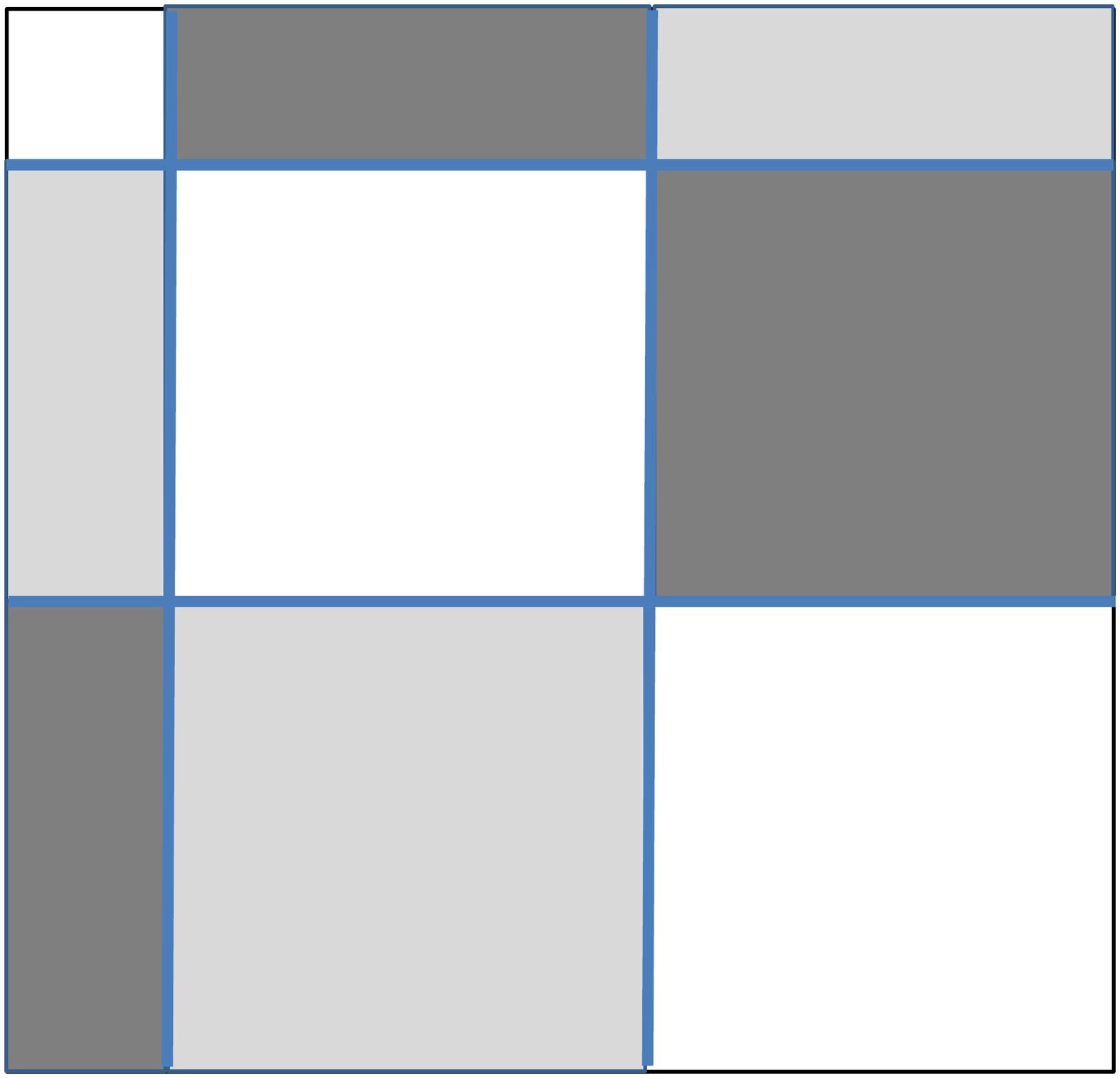}
\begin{picture}(0,0)
\put(-163,105){\small $0$}
\put(-80,105){\small $-a_1^t$}
\put(-165,75){\small $a_1$}
\put(-165,25){\small $a_2$}
\put(-130,105){\small $-a_2^t$}
\put(-130,75){$A$}
\put(-85,25){$-A^t$}
\put(-140,25){$\frac{1}{\sqrt{2}}[a_1]$}
\put(-90,75){$\frac{1}{\sqrt{2}}[a_2]$}
\put(-112,-10){\text{b}}
\end{picture}
\caption{Case of $G_2$: depth $2$}\label{G2}
\end{figure}

It is easy to check that the subspace $\g_{-2}$ consists of matrices of the form
\begin{equation}\label{E:L{-2}}
   L_{-2}=\mu\begin{pmatrix}
       0 & 0 & 0\\
       0 & \tilde\a_1\tilde\a_2^t &  0   \\
       0 &  0  & -\tilde\a_2\tilde\a_1^t    \\
       \end{pmatrix},\quad \mu\in\C ,
\end{equation}
where $\tilde\a_1=(0,1,0)$, $\tilde\a_2=(0,0,1)$, while the subspace $\g_{-1}$ consist of matrices of the form
\begin{equation}\label{E:resid}
   L_{-1}=\begin{pmatrix}
       0 & -\sqrt{2}\b_{02}\tilde\a_2^t & -\sqrt{2}\b_{01}\tilde\a_1^t \\
       \sqrt{2}\b_{01}\tilde\a_1 & \tilde\a_1\b_2^t-\b_1\tilde\a_2^t & \b_{02}[\tilde\a_2] \\
       \sqrt{2}\b_{02}\tilde\a_2 & \b_{01}[\tilde\a_1] & \tilde\a_2\b_1^t-\b_2\tilde\a_1^t \\
       \end{pmatrix},
\end{equation}
where $\b_{01},\b_{02}\in\C$ are arbitrary, $\b_1,\b_2\in \C^3$ satisfy to the following orthogonality relations: $\tilde\a_1^T\b_2=0$, $\tilde\a_2^T\b_1=0$. Observe also that $\tilde\a_1^T\tilde\a_2=0$, and if
$L_0\in\tilde\g_0$ is as in Figure \ref{G2},b, then
\begin{equation}\label{E:eigen}
\tilde\a_1^Ta_2=0, \quad \tilde\a_2^Ta_1=0, \quad A\tilde\alpha_1=\varkappa_1\tilde\alpha_1,  \quad -A^T\tilde\alpha_2=\varkappa_2\tilde\alpha_2 ,
\end{equation}
where $\varkappa_1,\varkappa_2\in\C$.

As a result we obtain the Lax operator algebra found in \cite{Sh_G2}. So we claim that the last corresponds to the grading of depth $2$ of the Lie algebra $G_2$ given by the simple root $\a_2$.

Besides, the Lie algebra $G_2$ has a grading of depth $3$ given by the simple root $\a_1$. The matrix realization of this grading is presented in Figure \ref{G2c}.
\begin{figure}[h] 
\includegraphics[width=6cm]{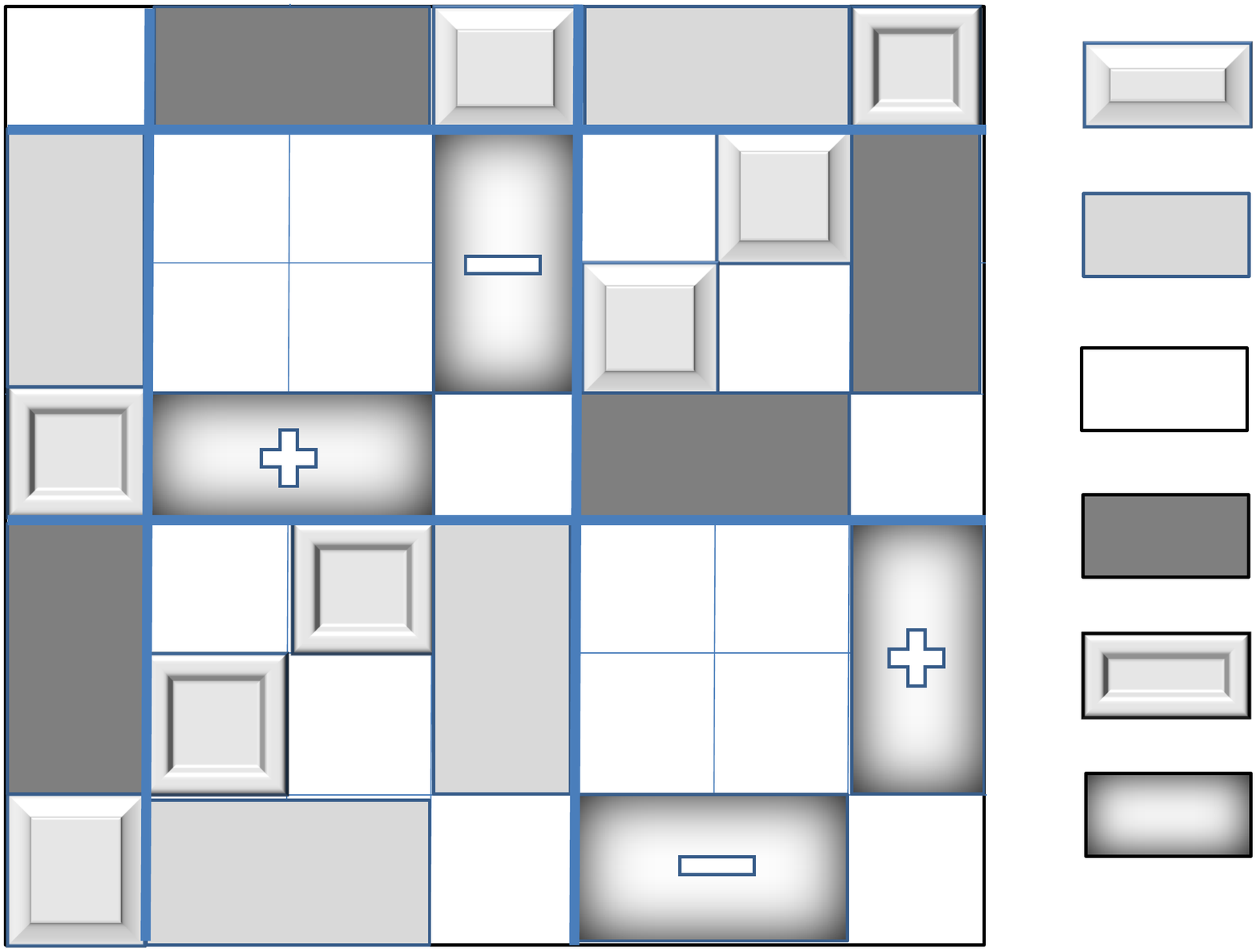}
\begin{picture}(30,0)
\put(-163,105){\small $0$}
\put(-95,90){\small $0$}
\put(-79,75){\small $0$}
\put(-60,57){\small $0$}
\put(-147,42){\small $0$}
\put(-130,25){\small $0$}
\put(-112,8){\small $0$}
\put(-10,102){\text{--}\ $\g_{-2}$}
\put(-10,85){\text{--}\ $\g_{-1}$}
\put(-10,68){\text{--}\ $\g_0$}
\put(-10,51){\text{--}\ $\g_1$}
\put(-10,31){\text{--}\ $\g_2$}
\put(-10,14){\text{--}\ $\g_{\pm 3}$}
\end{picture}
\caption{Case $G_2$: depth $3$}\label{G2c}
\end{figure}
\subsection{Tyurin parameters}\label{SS:Tyupar}
A Lax operator algebra $\L$ is being given by a choice of a Cartan subalgebra, and a $\Z$-grading of the Lie algebra $\g$ at every $\ga\in\Gamma$. Such pair of objects is defined up to an inner automorphism which depends on $\ga$. Operating by it on the above objects we actually operate on local expansions \refE{ga_expan} of elements $L\in\L$.

We call an inner automorphism of the Lie algebra of $\g$-valued Laurent expansions at a point $\ga\in\Gamma$, given by a constant (in $z$) element $g\in G$, a \emph{local inner automorphism}.

Operating by local inner automorphisms, in general, we deform the Lax operator algebra. In this section, for some of examples considered in \refSS{ex_grad}, we will point out independent parameters giving the Lax operator algebra. They are called Tyurin parameters. For the case of $\g=\gl(n)$, they emerged in classification of holomorphic vector bundles on Riemann surfaces \cite{Tyvb,rKNU}. In the context of integrable systems with spectral parameter on a Riemann suface, and Lax operator algebras, they emerged in \cite{Klax}, and in \cite{KSlax,Sh_DGr}, respectively.

Consider the expansion \refE{pTgln} of an element $L$ in the case of  $\g=\gl(n)$. A local automorphism operates on this expansion as follows:
\[
   L\to L',\ \a\to\a' ,\ \b\to\b' ,\ L_0\to L_0',
\]
where
\begin{equation}\label{E:loc_auto}
  L'=g^{-1}Lg, \ \a'=g^{-1}\a , \ \b'^t=\b^tg,\ L_0'= g^{-1}L_0'g.
\end{equation}
It is clear that the expansion \refE{ga_expan}, and the relations $\b^t\a=0$, $L_0\a=\varkappa\a$ are preserved under the transformation.

The set $\Gamma$ and the set of parameters $\a'$ for all $\ga\in\Gamma$ define the Lax operator algebra uniquely in this case. They are called \emph{Tyurin parameters}.

In the same way the Tyurin parameters are defined in case of local expansions \refE{pTso2n}, \refE{pTsp2n}, \refE{pTso2n+1} for orthogonal and symplectic algebras \cite{KSlax,Sh_DGr}.

We will need the parametrization of Lax operator algebras by means the local automorphisms (not only in the case of existence of Tyurin parameters) in \refSS{LP}, with relation to the definition of Lax equations.


\section{Lax equations and commutative hierarchies}\label{S:Hierarchies}
In this chapter we define a certain class of finite-dimensional evolution systems corresponding to Lax operator algebras, and start investigating their integrability. It is one of definitions of integrability that there exists a full, in some sense, set of commuting flows the evolution system commutes with (\emph{a commutative hierarchy of flows}). Given a Lax operator of the class in question we construct here the corresponding family of commuting flows.

\subsection{Lax pairs}\label{SS:LP}

Elements of the Lie algebra $\L$, which the last chapter was devoted to, are called $L$-operators here. We complete the notation of the algebra by inventing an indication to the sets $\Gamma$ and $\{ h\}$ defining it (see \refSS{constr}), and will denote it by $\L_{\Gamma,\{ h\}}$ from now on. Next, we formulate an important property of $L$-operators generalizing \refT{almgrad},$2^\circ$. Let $D$ be a nonspecial nonnegative divisor on $\Sigma$. We associate it with the subspace
\[
  \L^D_{\Gamma,\{h\}}=\{ L\in\L_{\Gamma,\{h\}}\ |\ (L)+D +k\sum_{\ga\in\Gamma} \ga\ge 0 \},
\]
where $(L)$ denotes the divisor of the mapping $L$. We would like make more precise that by divisor of a vector-valued function we mean the pointwise minimum of divisors of its entries. Then the following \emph{dimension formula} for $L$-operators holds:
\begin{equation}\label{E:dimL^D}
  \dim\L^D_{\Gamma,\{h\}}=(\dim\g)(\deg D-g+1).
\end{equation}
Proof of the relation \refE{dimL^D} is also quite similar to the proof of \refT{almgrad},$2^\circ$. Namely, the dimension of the space of all meromorphic $\g$-valued functions with divisor $D$ outside $\Gamma$, and the poles of order $k$ at $\ga\in\Gamma$, generically is equal to $(\dim\g)(\deg D+k|\Gamma|-g+1)$, by Riemann-Roch theorem. In addition, the $L$-operators satisfy the relations \refE{ga_expan} which have a codimension $k(\dim\g)|\Gamma|$, as it is shown in \refSS{constr}, which proves \refE{dimL^D}.

The relation \refE{dimL^D} is applied below (in \refSS{Hitchin}) to the computation of dimensions of phase spaces of integrable systems.

Next, we allow the elements of $\Gamma$ to vary in such way that all $\ga\in\Gamma$ remain mutually different, and $\Gamma\cap\Pi=\emptyset$. The elements of the set $\{h\}$ are also assumed variables with the following range. For every $\ga\in\Gamma$, we fix $h_\ga^0\in\h$ satisfying the integrality and positivity requirements formulated in \refSS{constr}, and assume that $h_\ga=(\Ad g_\ga) h_\ga^0$ where $g_\ga\in G$, $G$ is a connected Lie group with the Lie algebra $\g$. Supplying the obtained this way space of data sets with an appropriate topology (and even a complex structure), we can consider a sheaf $\L$ of Lax operator algebras $\L_{\Gamma,\{h\}}$ on it, and its subsheaf $\L^D$ of subspaces $\L^D_{\Gamma,\{h\}}$. The set $\Pi$ is assumed to be fixed. The sheaf $\L^D$ plays the role of the phase space of the Lax equation below.

We notice that considering the variable sets $\Gamma$ and $\{ h\}$ generalizes the \emph{method of deformation of Tyurin parameters}, applying of which in the theory of Kadomtsev--Petviashvili equation goes back to \cite{rKNU}, and which is heavily used in \cite{Klax}.

More formally, we can define the base of the sheaves $\L$ и $\L^D$ as follows. Let $H$ be the centralizer of the element $\h^0_\ga$ in $G$, then $g_\ga\in G/H$. For a fixed finite set $\Gamma$, denote by $G_H^\Gamma$ the set of all mappings $\Gamma\to G/H$, and by $\Sigma^\Gamma$ the set of all embeddings $\Gamma\to\Sigma$. Then we regard to $\Sigma^\Gamma\times G_H^\Gamma$ as to the base.

A meromorphic mapping $M:\ \Sigma\to\g$,  holomorphic outside $\Pi$ and $\Gamma$, is said to be an \emph{$M$-operator} if at any $\ga\in\Gamma$ it has a Laurent expansion
\begin{equation}\label{E:M_oper}
  M(z)=\frac{\nu h}{z}+\sum_{i=-k}^\infty M_iz^i
\end{equation}
where $M_i\in\tilde\g_i$ for $i<0$, $M_i\in\g$ for $i\ge 0$, $h\in\h$ is the element giving the grading on $\g$ at the point $\ga$, and $\nu\in\C$. We denote the collection of $M$-operators by $\M_{\Gamma,\{h\}}$.  Obviously, $\L_{\Gamma,\{h\}}\subset\M_{\Gamma,\{h\}}$.

For an arbitrary nonspecial, nonnegative divisor $D$ let
\[
  \M^D_{\Gamma,\{h\}}=\{ M\in\M_{\Gamma,\{h\}}\ |\ (M)+D +k\sum_{\ga\in\Gamma} \ga\ge 0 \}.
\]
According to \cite{Sh_TMPh_14}, the following \emph{dimension formula for $M$-operators} holds:
\begin{equation}\label{E:dimM^D2''}
\dim\M^D_{\Gamma,\{h\}} =  (\dim\g)(\deg D+l-g+1)
\end{equation}
where $l\in\Z_+$ is determined by the number of points in $\Gamma$, and dimensions of the filtration spaces of the Lie algebra $\g$. We will prove \refE{dimM^D2''} in \refSS{constrM}.

In particular cases $M$-operators, like $L$-operators, can be given by Tyurin parameters, see \cite{Sh_DGr} for the detailes.

In a similar way it was done for $L$-operators, we will consider the sheaves $\M$ and $\M^D$, with the same base.

\begin{definition}\label{D:Lax}
A \emph{Lax pair} is a pair consisting of smooth sections of the sheaves $\L^D$ and $\M^D$. From these two, the section of the sheaf $\L^D$ is called the \emph{Lax operator}.
\end{definition}

Let two smooth curves $\Gamma(t)$ in $\Sigma^\Gamma$, and $h(t)$ in $G_H^\Gamma$ to be given. Giving a Lax pair, defines a pull back of the curves to the sheaves $\L^D$ and $\M^D$. We will obtain the curves $L(t)$ and $M(t)$ where $L(t)\in\L_{\Gamma(t),\{h(t)\}}$,  $M(t)\in\M_{\Gamma(t),\{h(t)\}}$.

Let the curves $L=L(t)$, $M=M(t)$ to be obtained this way. The equation
\begin{equation}\label{E:Lax_eq0}
 \dot L=[L,M]
\end{equation}
where $\dot L=\frac{dL}{dt}$, is called \emph{Lax equation}. This is a system of ordinary differential equations on the curves $\Gamma(t)$ and $h(t)$, and on the main parts of meromorphic functions $L(t)$ and $M(t)$ at the points of $\Pi$ and $\Gamma$. In order this system was closed, it is necessary to give $M$ as a function of~$L$.

\begin{example}
Consider an example which will be investigated in more detail later (\refSS{CMoser}), namely the elliptic Calogero--Moser system for the root system $A_n$. From the physical point of view this is a system of pairwise interacting particles on a torus with coordinates $q_1,\ldots,q_n$, and momenta $p_1,\ldots,p_n$. Its Lax operator is a meromorphic function in $z$ on the torus taking values in $\gl(n)$, having the entries of the form:
\begin{equation}\label{E:matrLgl0}
 L_{ij}=f_{ij}\frac{\s(z+q_j-q_i)\s(z-q_j)\s(q_i)}{\s(z)\s(z-q_i)\s(q_i-q_j)\s(q_j)}\ (i\ne j), \ \
 L_{jj}=p_j ,
\end{equation}
where $\s$ is the Weierstra\ss\ function. In this example,  $\Gamma=\{q_1,\ldots,q_n\}$, and, at the point $q_i$, the element $h_i$ is given by a constant (in time) diagonal matrix equal to $diag(0,\ldots,0,1,0,\ldots,0)$ ($1$ at $i$'s position). The equation \refE{Lax_eq0} gives a motion in the phase space  $\{p_1,\ldots,p_n,q_1,\ldots,q_n\}$.
\end{example}

\subsection{Dimension formula for $M$-operators}\label{SS:constrM}
We choose an arbitrary nonnegative nonspecial divisor
\[
  D=\sum_{i=1}^N m_iP_i, \quad m_i\ge 0\ (i=1,\ldots,N)
\]
and compute the dimension of the space $\M^D_{\Gamma,\{h\}}=\{M\in\M_{\Gamma,\{h\}}\ |\ (M)+D\ge~0\}$.
By the Riemann-Roch theorem, taking account of the codimension of the expansions of $M$-operators at the points $\ga\in\Gamma$, and of the additional parameter $\nu$ at every of those points, we have:
\begin{equation}\label{E:dimM^D1}
\dim\M^D_{\Gamma,\{h\}}= (\dim\g)(\deg D+k|\Gamma|-g+1)-\sum_{\ga\in\Gamma}\sum_{i=-k}^{-1}\codim\tilde\g_i^\ga+|\Gamma| \end{equation}
where $\g_i^\ga$ is a space of the grading at $\ga$ (here we don't assume the gradings to be the same at different points), $\tilde\g_i^\ga$ is the corresponding filtration space. Replacing  $\codim\tilde\g_i^\ga$ with $\dim\g-\dim\tilde\g_i^\ga$ in \refE{dimM^D1}, we obtain
\begin{equation*}
\begin{aligned}
\sum_{\ga\in\Gamma}\sum_{i=-k}^{-1}\codim\tilde\g_i^\ga &=
\sum_{\ga\in\Gamma}\sum_{i=-k}^{-1}(\dim\g-\dim\tilde\g_i^\ga)\\
&= (\dim\g)k|\Gamma|-\sum_{\ga\in\Gamma}\sum_{i=-k}^{-1}\dim\tilde\g_i^\ga .
\end{aligned}
\end{equation*}
Therefore
\begin{equation*}
\dim\M^D_{\Gamma,\{h\}} = (\dim\g)(\deg D-g+1)+ \sum_{\ga\in\Gamma}\sum_{i=-k}^{-1}\dim\tilde\g_i^\ga+|\Gamma|,
\end{equation*}
and finally
\begin{equation}\label{E:dimM^D2}
\dim\M^D_{\Gamma,\{h\}} =  (\dim\g)(\deg D-g+1) +\sum_{\ga\in\Gamma}\left(\sum_{i=-k}^{-1}\dim\tilde\g_i^\ga +1\right) .
\end{equation}
Next, assume that the gradings are the same up to inner automorphisms at all points $\ga~\in~\Gamma$. Then
\begin{equation}\label{E:dimM^D2'}
\dim\M^D_{\Gamma,\{h\}} =  (\dim\g)(\deg D-g+1)+\left(\sum_{i=-k}^{-1}\dim\tilde\g_i^\ga +1\right)|\Gamma| .
\end{equation}
Choose $|\Gamma|$ so that the last summand is equal to $(\dim\g)l$ where $l\in\Z_+$.  It is always possible, and can be done in several ways. Then \cite{Sh_TMPh_14}
\begin{equation}\label{E:dimfo}
\dim\M^D_{\Gamma,\{h\}} =  (\dim\g)(\deg D+l-g+1).
\end{equation}
Below, we always assume that
\[
            l-g+1\ge 0.
\]
For example, for the Lie algebras $\gl(n)$, $\so(2n)$, $\spn(2n)$, $G_2$ and a certain choice of gradings, the following relation holds
\begin{equation}\label{E:mist}
\dim\g-\left(\sum_{i=-k}^{-1}\dim\tilde\g_i +1\right)n=0,
\end{equation}
and we can take $|\Gamma|=ng$ where $n=\rank\g$. Then $l=g$,  and we obtain
\begin{equation}\label{E:dimM^D3}
\dim\M^D_{\Gamma,\{h\}} =  (\dim\g)(\deg D+1) .
\end{equation}
In this form, the dimension formula for $\M^D_{\Gamma,\{h\}}$, in the case $\g=\gl(n)$, has been obtained by I.M.Krichever in \cite{Klax}. For the classical Lie agebras it has been obtained by author, see \cite{Sh_DGr} and references therein. The validity of the relation \refE{mist} in the above listed cases is verified in the next four examples.
\begin{example}\label{Ex:gr_gln}
$\g=\gl(n)$, a grading of depth 1 is given by the simple root $\a_1$. Then $k=1$, $\dim\tilde\g_{-1}=n-1$, $\dim\g-(\dim\tilde\g_{-1}+1)n=\dim\g-n^2=0$.
\end{example}
\begin{example}\label{Ex:gr_so2n}
$\g=\so(2n)$, a grading of depth 1 is given by the simple root $\a_1$, $k=1$, $\dim\tilde\g_{-1}=2n-2$, $\dim\g-(\dim\tilde\g_{-1}+1)n=\dim\g-(2n-1)n=0$.
\end{example}
\begin{example}\label{Ex:gr_sp2n}
$\g=\spn(2n)$, a grading of depth 2 is given by the simple root $\a_1$, $k=2$, $\dim\g_{-2}=1$, $\dim\g_{-1}=2n-2$, $\dim\tilde\g_{-1}=\dim\g_{-2}+\dim\g_{-1}=2n-1$, $\dim\g-(\dim\tilde\g_{-2}+\dim\tilde\g_{-1}+1)n=\dim\g-(2n+1)n=0$.
\end{example}
\begin{example}\label{Ex:gr_G2}
$\g=G_2$, a grading of depth 2 is given by the simple root $\a_1$, $k=2$, $n=2$, $\dim\g_{-2}=1$, $\dim\g_{-1}=4$, $\dim\tilde\g_{-1}=\dim\g_{-2}+\dim\g_{-1}=5$, $\dim\g-(\dim\tilde\g_{-2}+\dim\tilde\g_{-1}+1)n=\dim\g-7\cdot 2=0$.
\end{example}
Consider the case $g=1$. Then the example \refEx{gr_gln} corresponds to the Calogero--Moser system for the root system $A_n$. The two following examples give the Calogero--Moser systems for the root systems $D_n$, $C_n$ and $B_n$.
\begin{example}\label{Ex:gr_ell}
$\g=\so(2n)$ or $\g=\spn(2n)$, $|\Gamma|=2n$. We consider the same gradings on $\g$ as in the examples \ref{Ex:gr_so2n}, \ref{Ex:gr_sp2n}, respectively. Since the relation \refE{mist} still holds, $l=2$. We obtain
\begin{equation}\label{E:dimM^D3'}
   \dim\M^D_{\Gamma,\{h\}}=(\dim\g)(\deg D+2).
\end{equation}
\end{example}
\begin{example}\label{Ex:gr_so2n+1}
\label{so_2n+1}
$\g=\so(2n+1)$, the grading of depth 1 is given by a simple root $\a_1$,
$k=1$, $\dim\tilde\g_{-1}=2n-1$, $\dim\g=n(2n+1)$. Therefore $2\dim\g=(\dim\tilde\g_{-1}+1)(2n+1)$, and we can take $|\Gamma|=2n+1$, $l=2$. The relation \refE{dimM^D3'} holds also in this case.
\end{example}
Observe that neither requirement \refE{mist}, nor the dimension formula \refE{dimM^D3} are satisfied by $\g=\so(2n+1)$, as well as by $\g=\sln(n)$. In the last case $k=1$, $\dim\g_{-1}=n-1$, $\dim\g=n^2-1$, while $(\dim\g_{-1}+1)\rank\g=n(n-1)$.

For $g=1$ and $\g=\so(2n+1)$ there is an alternative to the dimension formula \refE{dimM^D3}:
\begin{equation}\label{E:alt}
  \dim\M^D_{\Gamma,\{h\}}=(\dim\g)(\deg D)+\dim\g-\rank\g .
\end{equation}
This relation immediately follows say from \refE{dimM^D2}.

Below, we make use of the dimension formula \refE{dimM^D2''}. It will be substantial starting from \refSS{hierarch}.

\subsection{$M$-operators and vector fields}\label{SS:DimMs}

Consider the Lax equation
\begin{equation}\label{E:Lax_eq}
 \dot L=[L,M].
\end{equation}
In order the pair $L$, $M$ be its solution, it is necessary that the following relations take place at every $\ga\in\Gamma$:
\begin{equation}\label{E:Tyur_eq}
  \dot z=-\nu,\ \dot L_p=\sum_{i+j=p} [L_i,M_j]+\nu\sum_{s=-k}^p(p+1-s)L_{p+1}^s\quad (p=-k,\ldots,0)
\end{equation}
where $L_{p+1}^s$ is a projection of $L_{p+1}$ onto $\g_s$.
To obtain them, it is sufficient to compare the expansions
\begin{equation}\label{E:exp_dot}
  \dot L=-kL_{-k}\dot zz^{-k-1}+\sum_{p=-k}^\infty(\dot L_p+(p+1)L_{p+1}\dot z)z^p,
\end{equation}
and
\begin{equation}\label{E:exp_comm}
\begin{aligned}
  \phantom{a}[L,M]=&kL_{-k}\nu z^{-k-1}+
  \sum_{p=-k}^\infty\bigg(\sum_{i+j=p} [L_i,M_j]+ \\ &+\nu\sum_{s=-k}^p(p+1-s)L_{p+1}^s
  -(p+1)L_{p+1}\nu\bigg)z^p.
\end{aligned}
\end{equation}
In abuse of notation, we write here $\dot z$ instead $\dot z_\ga$. Formally, it would be necessary to write the expansions for $L$, $M$ in degrees of $z-z_\ga$ where $z_\ga$ is the coordinate of the point $\ga$ which is also assumed to be  depending on~$t$. Differentiating such expansion in $t$ we would obtain $\dot z_\ga$.

Observe that the summands containing $L_{p+1}^s$ and $L_{p+1}$ in \refE{exp_comm} come from the commutator $L$ with $\nu h/z$, the last coming from the expansion of $M$.

Observe also that, due to the conditions $M_i\in\tilde\g_i$ ($i<0$), the terms of degree $p<-k-1$ vanish on both hand sides of the equation \refE{Lax_eq}, i.e. the equation does not give any new relation in these degrees.

Let $T_{L}\L^D_{\Gamma,\{ h\}}$ be the tangent space to $\L^D_{\Gamma,\{ h\}}$ at the point $L$, and $M\in\M$. We stress that it is not assumed above that $\dot L\in\T_L\L^D_{\Gamma,\{ h\}}$. As well, we can not tell anything about $\dot L_p$, except that it is an element of $\g$. The last -- because the element $h$ giving the grading, depends on $t$.
\begin{theorem}\label{T:Lax_corr}
Assume that $\dot L$ and $M$ satisfy the relations \refE{Tyur_eq} for every $\ga\in\Gamma$. Then $[L,M]\in T_L\L^D_{\Gamma,\{ h\}}$ if, and only if $([L,M])+D\ge 0$ outside $\Gamma$.
\end{theorem}
\begin{proof}
The plan of the proof is as follows. First, we construct an auxiliary subspace $\mathcal T^D$, and, applying the Riemann-Roch theorem, derive that $\dim\mathcal T^D=\dim T_L\L^D_{\Gamma,\{ h\}}$. Then we observe that $T_L\L^D_{\Gamma,\{ h\}}\subseteq\mathcal T^D$, hence these spaces coincide. Finally, $[L,M]\in\mathcal T^D\Leftrightarrow ([L,M])+D\ge 0$ outside $\Gamma$.

We define $\mathcal T^D$ as the subspace of meromorphic mappings $T: \Sigma~\to~\g$ holomorphic outside the sets $\Pi$, $\Gamma$, satisfying the condition
\[
 (T)+D\ge 0
\]
outside $\Gamma$, and having the expansion
\begin{equation}\label{E:Texp}
  T=\sum_{i=-k-1}^\infty T_iz^i
\end{equation}
at every $\ga\in\Gamma$ where $T_i\in\g$, $i=-k-1,k,\ldots,\infty$,
\begin{equation}\label{E:mathcalT}
 \begin{array}{llll}
   T_i     &=  & \dot L_i+(i+1)L_{i+1}\dot z ,  & (i=-k,\ldots,-1), \\
   T_0     &=  & \dot L_0+L_1\dot z,            &
 \end{array}
 \end{equation}
$L_i\in\tilde\g_i$ are fixed, $\dot L_i\in\g$ are free parameters. Comparing the right hand sides of \refE{exp_dot} and \refE{mathcalT}, conclude that $T_L\L^D_{\Gamma,\{ h\}}\subseteq\mathcal T^D$.

Next, compute the dimension of the space $\mathcal T^D$. If we relax the conditions \refE{mathcalT}, then the dimension of the obtained space  $\mathcal T$ of meromorphic functions can be computed by the Riemann-Roch theorem: $\dim\mathcal T = (\dim\g) (\deg D+(k+1)|\Gamma|-g+1)$. The conditions \refE{mathcalT} give $(\dim\g)(k+1)$ relations at every point $\ga\in\Gamma$, i.e. $(\dim\g)(k+1)|\Gamma|$ relations in total. The summarized dimension of the parameters $\dot L_i$, $i=-k,\ldots,0$ is also equal to $(\dim\g)(k+1)|\Gamma|$, but they satisfy to the same number of relations \refE{Tyur_eq}. Therefore the subspace $\mathcal T^D$ is distinguished in the space $\mathcal T$ by means $(\dim\g)(k+1)|\Gamma|$ effective relations, hence
\[
   \dim \mathcal T^D=\dim \mathcal T-(\dim\g)(k+1)|\Gamma|=(\dim\g)(\deg D-g+1).
\]
But $\dim T_L\L^D_{\Gamma,\{ h\}}$ is the same, because $\dim T_L\L^D_{\Gamma,\{ h\}}=\dim\L^D_{\Gamma,\{ h\}}$ (these are finite-dimensional spaces), and by the relation \refE{dimL^D}. Therefore we have proved that $T_L\L^D_{\Gamma,\{ h\}}=\mathcal T^D$.

Comparing \refE{exp_comm} and \refE{mathcalT}, we see that, under conditions \refE{Tyur_eq}, the expansion \refE{exp_comm} satisfies \refE{mathcalT}. Therefore it is necessary and sufficient for $[L,M]\in\mathcal T^D$ that $([L,M])+D\ge 0$ (outside $\Gamma$).
\end{proof}
Consider the next property of $M$-operators. For any $M$-operator we will write down its Laurent expansion in the neighborhood of $\ga\in\Gamma$ in the form $M(z)=\frac{\nu h}{z}+\sum_{i=-k}^\infty M_iz^i$, and also in the form $M=M^-+\frac{\nu h}{z}+M^+$ where $M^-=\sum_{i=-k}^{-1} M_iz^i$, $M^+=\sum_{i=0}^\infty M_iz^i$.
\begin{lemma}\label{L:mult}
Given two $M$-operators $M_a$ and $M_b$, consider two vector fields $\partial_a$, $\partial_b$ given by the equalities $\partial_aL=[L,M_a]$ and $\partial_bL=[L,M_b]$ on $\L^D$, and continued to $\M^D$ so that
\begin{equation}\label{E:local}
\partial_a M_b^-=[M_b^-,M_a]^-+L_a^-,\
\partial_b M_a^-=[M_a^-,M_b]^-+L_b^- ,
\end{equation}
and
\begin{equation}\label{E:local1}
\begin{aligned}
\partial_a z&=-\nu_a,\quad  \partial_ah=[h,M_{a,0}]+O(z),\\
\partial_b z&=-\nu_b,\quad\partial_bh=[h,M_{b,0}]+O(z),
\end{aligned}
\end{equation}
where $L_{a,i}^-,L_{b,i}^-\in\tilde\g_i$ ($i=-k,\ldots,-1$), and the upper minus at the commutator means its main part\footnote{The relations for $h$ in \refE{local1} are additional conditions the continuation of the vector fields satisfies to; it does not follow from the above because $h$, as a coefficient at $z^{-1}$, appears only in $M$-operators. The equations for $h$ are made more precise below by means the relation \refE{dvih}.}. Then
\[
\partial_aM_b-\partial_bM_a+[M_a,M_b]\in\M.
\]
\end{lemma}
\begin{proof}
The lemma will be proved by straightforward computation of the main part of the expression $\partial_aM_b-\partial_bM_a+[M_a,M_b]$ at an arbitrary point $\ga\in\Gamma$.

We write the Laurent expansions of the operators $M_a$ and $M_b$ in the neighborhood of $\ga\in\Gamma$ in the form
\begin{equation}\label{E:M_loc}
   M_a=M_a^-+\frac{\nu_ah}{z}+\sum_{i\ge 0}M_{a,i}z^i,\quad
   M_b=M_b^-+\frac{\nu_bh}{z}+\sum_{i\ge 0}M_{b,i}z^i
\end{equation}
where $M_a^-=\sum_{i=-k}^{-1}M_{a,i}z^i$, $M_b^-=\sum_{i=-k}^{-1}M_{b,i}z^i$.

At the first step we have
\[
\begin{aligned}
\partial_aM_b &=\partial_a\left( M_b^-+\frac{\nu_bh}{z}+\sum_{i\ge 0}M_{b,i}z^i  \right)=\\
&=\partial_aM_b^- + \frac{(\partial_a\nu_b)h}{z}+\frac{\nu_b(\partial_ah)}{z}- \frac{\nu_bh(\partial_az)}{z^2}+O(1).
\end{aligned}
\]
Here, $O(1)$ is the expansion in nonnegative degrees with coefficients in $\g$. We make substitutions $\partial_az=-\nu_a$, and $\partial_ah=[h,M_{a,0}]$ (by \refE{local1}). At the next step replace $\partial_aM_b^-$ with $[M_b^-,M_a^-]+[M_b^-,\frac{\nu_ah}{z}+M_a^+]+L_a^-+O(1)$ by \refE{local}. The second commutator contracts with $[\frac{\nu_ah}{z}+M_a^+,M_b^-]$ in $[M_a,M_b]$, which follows from \refE{M_loc}. Observe that the term $\nu_bh(\partial_az)z^{-2}=-\nu_a\nu_bhz^{-2}$ is symmetric in $a$ and $b$, and annihilates with the corresponding term in $\partial_bM_a$. Therefore, replacing with dots the terms which definitely disappear in the result, we obtain
\[
    \partial_aM_b = [M_b^-,M_a^-]+L_a^-+\frac{(\partial_a\nu_b)h}{z} + \frac{\nu_b(\partial_ah)}{z} +\ldots+O(1).
\]
Computing $\partial_bM_a$ similarly, we obtain
\begin{equation}\label{E:poldela1}
\begin{aligned}
  \partial_aM_b-\partial_bM_a = &2[M_b^-,M_a^-]+ \frac{(\partial_a\nu_b-\partial_b\nu_a)h}{z} + \frac{\nu_b(\partial_ah)-\nu_a(\partial_bh)}{z}\\
  &+L_a^--L_b^-+\ldots+O(1).
\end{aligned}
\end{equation}
Next, we compute $[M_a,M_b]$ starting from \refE{M_loc}:
\begin{equation}\label{E:poldela2}
[M_a,M_b]=[M_a^-,M_b^-]+\sum_{i=0}^\infty \nu_a[h,M_{b,i}]z^{i-1} +
\sum_{i=0}^\infty \nu_b[M_{a,i},h]z^{i-1} + \ldots + O(1).
\end{equation}
Observe that by \refE{local1}, $\nu_b(\partial_ah)z^{-1}$ in \refE{poldela1} annihilates with the term with $i=0$ of the secomd sum in \refE{poldela2}, and similarly $-\nu_a(\partial_bh)z^{-1}$ annihilates with the number zero term of the first sum. The remainder of those two sums in \refE{poldela2} is $O(1)$.

Therefore
\[
  \partial_aM_b-\partial_bM_a+[M_a,M_b]=[M_b^-,M_a^-] +L_a^--L_b^- + \frac{(\partial_a\nu_b-\partial_b\nu_a)h}{z} + O(1).
\]
Since $M_{a,i}\in\tilde\g_i$ and $M_{b,i}\in\tilde\g_i$ for $i<0$, the $[M_b^-,M_a^-]$ possesses this property too. By assumption, $L_a^-$ and $L_b^-$ also possess it, hence $\partial_aM_b-\partial_bM_a+[M_a,M_b]$ is an $M$-operator.
\end{proof}
We notice that the continuation of vector fields from $\L^D$ to $\M$, constructed in \cite{Sh_DGr} for classical Lie algebras by means Tyurin parameters, satisfy the conditions \refE{local}. Since it is not completely obvious, we will show it by the example $\g=\gl(n)$.
\begin{example}
For $\g=\gl(n)$, in the neighborhood of a $\ga$ we have
\[
   M_b^-=\frac{\a\mu_b^t}{z},\quad M_a=\frac{\a\mu_a^t}{z}+M_{a0}+\ldots,
\]
where $\a$ are Tyurin parameter, $\mu_b^t\a=0$ (for $\mu_a$ it is not true because the matrix $\a\mu_a^t$ contains $\nu_ah$). Therefore
\begin{equation}\label{E:prod1}
 [M_b^-,M_a]^-=-\frac{(\mu_a^t\a)\a\mu_b^t}{z^2}+\frac{\a\mu_b^tM_{a0}-M_{a0}\a\mu_b^t}{z}.
\end{equation}
If we compute $(\partial_aM_b)^-$ another way, formally applying the derivation $\partial_a$ and using the Leibnitz formula, we will obtain
\begin{equation}\label{E:prod2}
 (\partial_aM_b)^-=\frac{(\partial_az)\a\mu_b^t}{z^2}+ \frac{(\partial_a\a)\mu_b^t+\a(\partial_a\mu_b^t)}{z}.
\end{equation}
The motion equations of the Tyurin parameters are as follows \cite{Klax,Sh_DGr}:
\begin{equation}\label{E:Tyupa}
        \partial_az=-\mu_a\a,\quad\partial_a\a=-M_{a0}\a+\l\a\ (\l\in\C).
\end{equation}
(for the way to obtain them see also example \ref{PT_ex},  \refSS{Holomorphy}).
By these equations
\[
   (\partial_aM_b)^- - [M_b^-,M_a]^-=\frac{\l\a\mu_b^t+\a(\partial_a\mu_b^t-\mu_b^tM_{a0})}{z}.
\]
The first summand in the nominator belongs to the space $\g_{-1}$ by assumption. To show that the second one also belongs to it, we check that $(\partial_a\mu_b^t-\mu_b^tM_{a0})\a=0$. It can be derived by differentiation of the relation $\mu_b^t\a=0$, and applying the second of the relations \refE{Tyupa}.
\end{example}

\subsection{Hierarchies of Lax equations}\label{SS:hierarch}
The indices $a$, $b$ introduced in the last section, denote here triples of the form $\{ \chi, P\in\Pi, m>-m_P\}$ where $\chi$ is an invariant polynomial of the Lie algebra $\g$, $m_P$ is a multiplicity of the point $P$ in the divisor $D$.
\begin{example}\label{pol_inv_cl}
For $\g=\gl(n)$ we can take $\chi(L)=\tr L^p$, $p\in\Z_+$. In the cases $\g=\so(n)$, $\g=\spn(2n)$ we can manage the same way, assuming that $p\in 2\Z_+$.
\end{example}
We will introduce also $l-g+1$ fixed points $P_j\notin(\Pi\cup\Gamma)$, $j=1,\ldots,l-g+1$ to normalize $M$-operators  ($l$ is the same here as in \refE{dimM^D2''}).

The following fragment, till the end of the proof of \refL{central}, is inspired by \cite{Goldman}. There is the only obstruction to establishing an equivalence between the results here, and there, namely, nonuniqueness of the logarithm mapping on a Lie group.

We define the \emph{gradient} $\d \chi(L)\in\g$ of the polynomial $\chi$ at the point $L\in\g$ by means the equality
\begin{equation}\label{E:variation}
   d\chi(L)=\langle \d \chi(L),\d L\rangle
\end{equation}
where $d\chi$ is the differential of $\chi$ as a function on $\g$, $\langle \cdot\, ,\cdot\rangle$ is a non-degenerate invariant bilinear form on $\g$. If $L\in\L$, i.e. it is considered as a meromorphic function on $\Sigma$ taking values in $\g$, then such will be also $\d \chi(L)$. If this function is considered as a function of a local coordinate $w$ on $\Sigma$, then we write $\d \chi(w)$.
\begin{lemma}\label{L:central}  Let $\chi$ be an invariant polynomial on the Lie algebra $\g$. Then
\[
    [\d \chi(L),L]=0.
\]
\end{lemma}
\begin{proof}
By definition, invariance of $\chi$ means that for any $g\in\exp\g$
\[
    \chi((\Ad g)L)=\chi(L).
\]
Taking differential of the both sides of the last equality, by its invariance with respect to substitutions, and also by \refE{variation}, we have
\[
   \langle\d\chi((\Ad g)L),\d((\Ad g)L)\rangle=\langle\d\chi(L),\d L\rangle .
\]
Since $\d((\Ad g)L)=(\Ad g)\d L$, and by invariance of the bilinear form (which implies that $\Ad g$ is an orthogonal operator) we have
\[
   \langle (\Ad g^{-1})\d\chi((\Ad g)L),\d L\rangle=\langle\d\chi(L),\d L\rangle .
\]
Since the last equality holds for every $\d L\in\g$, and the bilinear form is non-degenerate, we obtain that $\d\chi(L)$ is equvariant:
\begin{equation}\label{E:equivar}
   \d\chi((\Ad g)L)=(\Ad g)\d\chi(L) .
\end{equation}
Put $g=\exp(tL)$ here, and differentiate the obtained equality at  $t=0$. Since the left hand side is equal to $\d\chi(L)$, and does not depend on  $t$, we will arrive to the statement of the lemma.
\end{proof}
\begin{lemma}\label{L:Lax_eq}    Given $L\in\L$, $M\in\M$ let $\xi_M$ be a vector field on $\L$ defined by
\[
    \xi_ML=[L,M],
\]
$\chi$ be an invariant polynomial on the Lie algebra $\g$. Then
\[
    \xi_M\d \chi(L)=[\d \chi(L),M].
\]
\end{lemma}
\begin{proof} Let $g=\exp(-tM)$ in \refE{equivar}. The assertion of the lemma is obtained by differentiation of both parts of the obtained equality at $t=0$.
\end{proof}
\begin{lemma}\label{L:Ma}
For every $a=\{ \chi, P, m\}$, and every $L\in\L$ there exists the only $M$-operator $M_a$ having a unique pole outside $\Gamma$, namely at $P$, where the following relation holds:
\begin{equation}\label{E:sravn}
     M_a(w)=w^{-m}\d \chi(L(w))+O(1)
\end{equation}
($w$ is a local parameter in the neighborhood of $P$), and $M_a(P_j)=0$, $j=1,\ldots,l-g+1$.
If $L\in\L^D$ then $([L,M_a])+D\ge 0$ outside $\Gamma$.
\end{lemma}
\begin{proof}
Outside $\Gamma$, the divisor of $M_a$  consists of one point $P$, and has the order $d=\ord_P w^{-m}\d \chi(L(w))$ there. By \refE{dimM^D2''} the dimension of the space of such operators is equal to $(\dim\g)(d+l-g+1)$. The equality \refE{sravn} fixes $(\dim\g)d$ degrees of freedom, and  $(\dim\g)(l-g+1)$ more relation are given by the normalization conditions. Thus the existence, and uniqueness of $M_a$ are proved.

By \refE{sravn}, and \refL{central}, in the neighborhood of the point $P$ we have $[L,M_a]=[L,O(1)]$. The same is true at the other points of $\Pi$ because $M_a$ is holomorphic there. For this reason, $([L,M_a])\ge (L)$ outside $\Gamma$. The $(L)+D\ge 0$ implies $([L,M_a])+D\ge 0$, which proves the second assertion of the lemma.
\end{proof}
\refL{Ma} defines $M_a$ as a function of $L$ which we will denote by $M_a(L)$.

$P\in\Sigma$ is said to be a \emph{regular point} for a Lax operator $L$ if $L(P)$ is well-defined, and is a regular element of the Lie algebra $\g$. $P$ is said to be \emph{non-regular} if this value is well-defined, but is not regular. The set of non-regular points of a given Lax operator is finite. This implies that the set of the Lax operators for which the set of non-regular points has an empty intersection with the set $\Pi$, is open.
\begin{theorem}\label{T:hierarch}
Relations
\begin{equation}\label{E:part_a}
      \partial_aL=[L,M_a]
\end{equation}
where $M_a=M_a(L)$, define a set of commuting vector fields on an open subset in $\L^D$, consisting of Lax operators for which the set of non-regular points has an empty intersection with $\Pi$.
\end{theorem}
\begin{proof}
By \refL{Ma} and \refT{Lax_corr}, $\partial_a$ is a tangential vector field on $\L^D$.

The commutator $[\partial_a,\partial_b]$ is being computed as follows:
\[
\partial_a\partial_bL=
\partial_a[L,M_b]=
[\partial_aL,M_b]+[L,\partial_aM_b]=
[[L,M_a],M_b]+[L,\partial_aM_b].
\]
Therefore,
\[
(\partial_a\partial_b-\partial_b\partial_a)L=
[L,\partial_aM_b-\partial_bM_a]+[[L,M_a],M_b]-[[L,M_b],M_a].
\]
Proceeding with help of the Jacoby identity, we obtain
\[
    [\partial_a,\partial_b]L=[L,\partial_aM_b-\partial_bM_a+[M_a,M_b]].
\]

In order $\partial_a$ and $\partial_b$ commuted, it is sufficient that $\partial_aM_b-\partial_bM_a+[M_a,M_b]=0$. Next, we verify that it is really the case. First, we show that the expression on the left hand side is holomorphic at an arbitrary $P\in\Pi$. We assume first that $a$ and $b$ correspond to the same $P\in\Pi$, i.e. $a=(\chi_a,P,m)$, $b=(\chi_b,P,m')$. We denote $M_a-w^{-m}\d\chi_a(L)$ by $M_a^+$, and $M_b-w^{-m'}\d\chi_b(L)$ by $M_b^+$. Then, by \refE{sravn}, $M_a^+$ и $M_b^+$ are holomorphic at~$P$. We have $\partial_aM_b=w^{-m'}\partial_a\d\chi_b(L) + \partial_aM_b^+$. Applying \refL{Lax_eq}, we obtain
\[
  \partial_aM_b=w^{-m'}[\d\chi_b(L),M_a] + \partial_aM_b^+
               =w^{-m'}[\d\chi_b(L),w^{-m}\d\chi_a(L)+M_a^+] + \partial_aM_b^+ .
\]
By \refL{central}, $\d\chi_a(w)$ and $\d\chi_b(w)$ commute with $L(w)$ for any $w$, where  $w$ is a local coordinate in the neighborhood of $P$, as above. According to the assumptions of the theorem, $L(w)\in\g$ is a regular element for sufficiently small $w$. For this reason, its centralizer coincides with the Cartan subalgebra containing it, hence it is commutative. Therefore
\begin{equation}\label{E:vcomm}
   [\d\chi_a(w),\d\chi_b(w)]=0.
\end{equation}
Applying \refE{vcomm} to the previous relation, we obtain
\[
\partial_aM_b = w^{-m'}[\d\chi_b(L),M_a^+] + \partial_aM_b^+ ,
\]
and a similar formula holds for $\partial_bM_a$.

Further on,
\[
\begin{aligned}
  \left[M_a,M_b\right]&=[M_a^+ + w^{-m}\d\chi_a(L),M_b^+ + w^{-m'}\d\chi_b] \\
  &= [M_a^+,M_b^+]+[M_a^+,w^{-m'}\d\chi_b] + [w^{-m}\d\chi_a,M_b^+].
\end{aligned}
\]
Therefore $\partial_aM_b-\partial_bM_a+[M_a,M_b]=\partial_aM_b^+ -\partial_bM_a^++[M_a^+,M_b^+]$. The function on the right hand side of this equality is holomorphic in the neighborhood of $P$.

By \refL{mult}, $\partial_aM_b-\partial_bM_a+[M_a,M_b]$ is an $M$-operator. By just proved, this $M$-operator has zero divisor outside $\Gamma$. According to \refE{dimM^D2''} the dimension of the space of such $M$-operators is equal to $l-g+1$. Because $M_a$ and $M_b$ satisfy the normalizing condition $M_a(P_j)=M_b(P_j)=0$, $j=1,\ldots,l-g+1$, our operator also satisfies this condition. Therefore $\partial_aM_b-\partial_bM_a+[M_a,M_b]~=~0$.

In the case when $a$ and $b$ correspond to different points, the proof is similar.
\end{proof}
\begin{corollary}\label{C:coroll}
The Lax equations $\dot L=[L,M_a]$ give commuting flows on $\L^D$.
\end{corollary}
Indeed, these equations can be written as $\dot L=\partial_aL$ where $\partial_a$ are commuting tangential vector fields on $\L^D$.
\begin{example}\label{Ex:M_cl}
In the case of classical Lie algebras from the example \ref{pol_inv_cl} the space of invariant polynomials is generated by the polynomials of the form $\chi_p(L)=\tr L^p$, $p\in\Z_+$ where $p$ is even for orthogonal and symplectic algebras. Then $\d\chi_p(L)=pL^{p-1}$.
\refL{central} and relation \refE{vcomm} become obvious, and \refL{Lax_eq} can be easy proven by induction.
\end{example}

\section{Hamiltonian theory}\label{S:Ham_th}

In this chapter, we construct a symplectic structure, and involutive system of Hamiltonians, giving a commutative hierarchy of flows $\dot L=[L,M_a]$ constructed in the previous chapter. As above, all constructions are given in terms of semisimple Lie algebras, and their invariants.

It should be stressed that the basic principles of the Hamiltonian theory of Lax equations with the spectral parameter on a Riemann surfaces are established in \cite{Klax}. Here, we actually show that these principles work in the more general context.

The symplectic structure we use below is invented by I.M.Krichever in \cite{Klax} for $\g=\gl(n)$ and is generalized here to the case of an arbitrary semisimple Lie algebra without any modification (for the classical Lie algebras see \cite{Sh_DGr}, in more general context of soliton theory see \cite{KrPhong}). We call it the  \emph{Krichever--Phong symplectic structure}. For the classical Lie algebras everything presented here has been formulated in terms of Tyurin parameters in \cite{Sh_DGr}, for this reason we omit here some details characteristic for that way of presentation (for example, the contribution of the points $\ga\in\Gamma$ into the symplectic structure in terms of the Tyurin parameters, and similar).

In this section, we keep the following convention on the notation. We do not distinguish between notation of elements of Lie groups and Lie algebras, and notation of the operators of their action. An element, and its operator  are denoted by the same letter. In other words, we keep the notation conventional for the case of classical groups where the operator is the same as the action of the element in the standard representation. In general, we could consider the action in the adjoint action instead (see \cite{Sh_MMJ_15} for example): the presentation below does not depeend on it.

\subsection{Symplectic structure}\label{SS:KrPh}

Following \cite{Klax}, we introduce here a closed skew-simmetric  2-form (the Krichever--Phong form) on~$\L^D$, becoming nondegenerate (i.e. a symplectic structure) on some submanifold $\P^D\subset\L^D/G$ (where $G$ is a connected Lie group with the Lie algebra $\g$).

The Krichever--Phong form, as well as an invariant scalar product on $\L_{\Gamma,\{h\}}$, depends on a choice of the holomorphic differential $\varpi$ on the Riemann surface. In the context of the theory of integrable systems the nonuniqueness of the invariant scalar product in the case of loop algebras has been noticed in \cite[p. 66]{R_S_TSh}, for example. In  \cite{GKMMM} the necessity of the choice of $\varpi$ has been explained from the point of view of Seiberg--Witten theory.

Let $\Psi\, :\, \Sigma\to G$ be the function diagonalizing   $L$\label{cannorm} at a generic point of the Riemann surface, $\d L$ and $\d\Psi$ are external differentials of the corresponding functions, which are 1-forms on $\L^D$. Similarly, we consider the function $K\, :\, \Sigma\to \exp\h$ defined by the relation
\[
   \Psi L=K\Psi ,
\]
and $\h$-valued 1-form $\d K$. $K$ and $\Psi$ are defined up to the action of the Weyl group which does not affect anything below.

Let $\Omega$ be a 2-form on $\L^D$ taking values in the space of meromorphic functions on $\Sigma$, defined by the relation
\[
   \Omega=\tr(\Psi^{-1}\d\Psi\wedge\d L-\d
   K\wedge\d\Psi\cdot\Psi^{-1}).
\]

We choose a holomorphic differential $\varpi$ on $\Sigma$ and define a scalar 2-form $\w$ on $\L^D$ by the relation
\[
    \w=\sum\limits_{\ga\in\Gamma} \res_\gamma\Omega \varpi+
    \sum\limits_{p\in\Pi}\res_p\Omega \varpi.
\]
$\Omega$ has one more representation:
\[
 \Omega=2\d\,\tr\left( \d\Psi\cdot\Psi^{-1}K\right)
\]
which obviously implies that $\w$ is closed. As it is pointed out above, we do not consider here the question of nondegeneracy of $\w$, and refer to \cite{Klax} for the corresponding discussion. We only notice that one of the conditions defining symplectic submanifolds in $\L^D$ is holomorphy of the 1-form $\d K\varpi$ on $\Pi$.

\subsection{System of Hamiltonians in involution}\label{SS:Ham_sys}

We will establish here that the hierarchies of commuting flows defined by \refT{hierarch} are Hamiltonian with respect to the Krichever--Phong symplectic structure.

Let $a=\{ \chi, P,m\}$ be as introduced in \refSS{hierarch}, and let
\[
H_a(L)=\res_P w^{-m}\chi(L(w))\varpi(w).
\]

For a vector field $e$ on $\L^D$, let $i_e\w$ be the 1-form defined by the relation $i_e\w(X)=\w(e,X)$ (where $X$ is an arbitrary vector field). By definition, $e$ is Hamiltonian if $i_e\w=\d H$ where $H$ is a function called Hamiltonian of $e$. The following theorem asserts the vector fields $\partial_a$ are Hamiltonian. For the classical Lie algebras it is proved in \cite{Klax,Sh_DGr}.
\begin{theorem}\label{T:Hamil}
Let $\partial_a$ be the vector field defined by the relation \refE{part_a}.
Then
\[
   i_{\partial_a}\w=\d H_a .
\]
\end{theorem}
Before we start proving the theorem, we remark the following. The operators $L$, $\partial_a+M_a$ commute by the Lax equation, for this reason they are diagonalized by the same function $\Psi$. The diagonal forms of these two operators are
\begin{equation}\label{E:eigenPsi}
  K=\Psi L\Psi^{-1},\ F_a=\Psi (\partial_a+M_a)\Psi^{-1}.
\end{equation}
An important role in the proof of the theorem is played by the holomorphy of $K$ and $F_a$ as functions on $\Sigma$. It follows from the following two lemmas which we prove below in \refSS{Holomorphy}.
\begin{lemma}\label{L:elimin}
The poles of the elements $L\in\L$ at the points of the set $\Gamma$ can be eliminated by means of conjugation by a local holomorphic function on $\Sigma$ taking values in $\exp\h$. In a neighborhood of a $\ga\in\Gamma$, this function has the form $e^{-h\ln z}$ where $h$ gives the grading at~$\ga$.
\end{lemma}
\begin{lemma}\label{L:specMhol}
The spectrum $F_a$ of the operator $\partial_a+M_a$ is holomorphic at the points of $\Gamma$ along solutions of the equation  $\partial_a L=[L,M_a]$.
\end{lemma}
We assume also that $\Psi$ is holomorphic, and has holomorphic inverse on $\Pi$ which is a generic position requirement for $L$.
\begin{proof}[Proof of \refT{Hamil}]
By definition
\[
i_{\partial_a}\w=\w(\partial_a,\cdot)=-{1\over 2}\left(
\sum_{\ga\in\Gamma}^K\res_\ga\Lambda+\sum_{p\in\Pi}\res_p\Lambda\right),
\]
where $\Lambda=\Omega(\partial_a,\cdot)$. Since
$\d\Psi(\partial_a)=\partial_a\Psi$ and $\d
L(\partial_a)=\partial_aL$ (the evaluation of a differential on a vector field is equal to the derivative along the vector field), we obtain
\[
\Lambda=\tr\left(\partial_a\Psi\cdot\d
L\cdot\Psi^{-1}-\d\Psi\cdot\partial_a
L\cdot\Psi^{-1}-\partial_aK\cdot\d\Psi\cdot\Psi^{-1}+\d
K\cdot\partial_a\Psi\cdot\Psi^{-1}\right).
\]
By the Lax equation
\begin{align*}
\Lambda=&\tr\left((\Psi M_a-F_a\Psi)\d
L\cdot\Psi^{-1}-\d\Psi[L,M_a]\Psi^{-1}+\d K(\Psi
M_a-F_a\Psi)\Psi^{-1}\right)\\
=&\tr\left( M_a\d L-F_a\Psi\d
L\cdot\Psi^{-1}-\d\Psi[L,M_a]\Psi^{-1}+\d K\Psi M_a\Psi^{-1}-\d
KF_a\right).
\end{align*}
Next, transform the middle term. From $\Psi L=K\Psi$ we derive
$\d\Psi\cdot L=-\Psi\d L+\d K\Psi+K\d\Psi$. Therefore
\begin{align*}
\tr\,\d\Psi[L,M_a]\Psi^{-1}= &\,\tr\left( (\d\Psi\cdot
L)M_a\Psi^{-1}-\d\Psi M_aL\Psi^{-1}\right)\\
=&\,\tr\left( (-\Psi\d L+\d K\Psi+K\d\Psi)M_a\Psi^{-1}-\d\Psi
M_aL\Psi^{-1}\right)\\
=&\,\tr\left( -\Psi\d LM_a\Psi^{-1}+\d K\Psi M_a\Psi^{-1}+K\d\Psi
M_a\Psi^{-1}-\d\Psi M_aL\Psi^{-1}\right).
\end{align*}
The last two terms annihilate because
\[
\tr\left( \d\Psi M_aL\Psi^{-1}\right)=\tr\left( \d\Psi
M_a\Psi^{-1}(\Psi L\Psi^{-1})\right)=\tr\left( \d\Psi
M_a\Psi^{-1}K\right),
\]
and we obtain
\[
\tr\,\d\Psi[L,M_a]\Psi^{-1}=\tr\left( -\d LM_a+\d K\Psi
M_a\Psi^{-1}\right).
\]
Substituting this to the last expression for $\Lambda$, we obtain
\[
\Lambda=\tr\left( 2M_a\d L-F_a\Psi\d L\cdot\Psi^{-1}-\d
KF_a\right).
\]
The last two terms in brackets are equal under symbol of trace,  which can be derived by replacing $\Psi\d L$ with $-\d\Psi L+\d
K\Psi+K\d\Psi$, and making use of commutativity $K$ and $F_a$. Finally
\[
\Lambda=\tr\left( 2M_a\d L-2\d KF_a\right).
\]
This implies
\begin{equation}\label{E:final}
i_{\partial_a}\w=\sum_{p\in\Pi}\res_p\tr(\d K\, F_a)\varpi-R_a,
\end{equation}
where
\begin{equation}\label{E:R_a}
R_a=\sum_{\ga\in\Gamma}\res_\ga\tr(\d LM_a)\varpi+\sum_{p\in\Pi}
\res_p\tr(\d LM_a)\varpi.
\end{equation}
In general, there should be the sum of residues of the 1-form $\tr(\d K\,
F_a)\varpi$ over the points $\ga\in\Gamma$  in \refE{final}, but it vanishes since $\d K$ and $F_a$ are holomorphic there. Observe that the function $F_a$ has singularities outside $\Pi\cup\Gamma$. Indeed,
$F_a=-\partial_a\Psi\cdot\Psi^{-1}+\Psi M_a\Psi^{-1}$, and
$\Psi^{-1}$ has poles at the branch points of the spectrum of $L$.

In contrary, $L$ and $M_a$ are holomorphic everywhere except at the points $\Pi\cup\Gamma$. Therefore $R_a=0$ as the sum of residues of a meromorphic 1-form over all its poles. Moreover, by construction of $M_a$, the function $F_a$ is holomorphic at all points of the set $\Pi$ except at $P$. For this reason
\[
i_{\partial_a}\w=\res_P\tr(\d K\, F_a)\varpi.
\]

So far, we followed literally the lines of the proof of the theorem in \cite{Klax} (and in \cite{Sh_DGr}). The remainder of the proof will be given in more general set-up of this paper, i.e. for the $M_a$ defined by an arbitrary invariant polynomial $\chi$ (\refL{Ma}). We recall that $\d\chi$ denotes the gradient of the invariant polynomial $\chi$ on the Lie algebra $\g$ with respect to the Cartan--Killing form, while $\d K$, $\d L$ denote differentials of the corresponding functions on $\L^D$.

In a neighborhood of $P$, first, by holomorphy of $\Psi$, $\Psi^{-1}$
\[
 F_a=-\partial_a\Psi\cdot\Psi^{-1}+\Psi M_a\Psi^{-1}=\Psi M_a\Psi^{-1}+O(1),
\]
and, second, by definition, $M_a=w^{-m}\d\chi(L(w))+O(1)$.  By \refE{equivar} $\Psi\d\chi(L)\Psi^{-1}=\d\chi(\Psi L\Psi^{-1})=\d\chi(K)$.
Therefore $F_a=w^{-m}\d\chi(K)+O(1)$. Since $\d K\varpi$ is also holomorphic at $P$, we obtain
\[
\begin{aligned}
i_{\partial_a}\w=&\res_P\tr(w^{-m}\d\chi(K)\d K)\varpi=\d\res_Pw^{-m}\chi(K)\varpi= \\
=&\d\res_Pw^{-m}\chi(L)\varpi= \d H_a.
\end{aligned}
\]
\end{proof}
\begin{corollary}
Hamiltonians $H_a$ are in involution.
\end{corollary}


Regarding to independence of the Hamiltonians $H_a$ for the basis invariants $\chi_i$ ($i=1,\ldots, r$, where $r=\rank\g$) as to an obvious fact, we count their number below. First, \emph{if $\chi$ is an invariant polynomial on the Lie algebra $\g$, then the function $\chi(L(P))$ has no pole on~$\Gamma$.} It immediately follows by \refL{elimin}: by invariance of the polynomial $\chi$ we have $\chi(L(z))=\chi(\Ad e^{-h\ln z}L(z))$. By \refL{elimin}, $\Ad e^{-h\ln z}L(z)$ has no pole on~$\Gamma$, which implies the required statement concerning  $\chi(L(z))$.
\begin{lemma}\label{L:nham}
Let $N$ be a number of independent Hamiltonians of the form $H_a$. Then for the semisimple Lie algebra $\g$, and nonspecial divisor $D$
\begin{equation}\label{E:nham}
 2N=\dim\g\deg D+r(\deg D-\deg\K).
\end{equation}
\end{lemma}
\begin{proof}
For the proof, we make use of the same identity which is most essential in the proof of the Hitchin theorem on integrability of his celebrated systems \cite{Hit}. Namely, let $d_1,\ldots,d_r$ be the set of degrees of the basis polynomials of the Lie algebra $\g$ ($r=\rank\g$). Then
\begin{equation}\label{E:tozh}
 \sum_{i=1}^r (2d_i-1)=\dim\g .
\end{equation}
This implies that
\begin{equation}\label{E:tozh1}
 \sum_{i=1}^r d_i= \frac{\dim\g+r}{2}.
\end{equation}
The basis Hamiltonians are obtained by means of expanding of functions of the form $\chi_i(L)$ over the basis meromorphic functions on $\Sigma$, holomorphic (as it follows from \refL{elimin}) on $\Gamma$ where $\chi_i$ runs over basis invariants. The divisors of the spaces of such functions are as follows: $d_iD$, $i=1,\ldots,r$. The total dimension of these spaces (equal to the number of Hamiltonians), by the Riemann--Roch theorem, is equal to
\begin{equation}\label{E:N_int}
\begin{aligned}
      N=\sum_{i=1}^r h^0(d_iD) &= \sum_{i=1}^r (d_i\deg D -g+1) \\
            &= (\deg D)\sum_{i=1}^r d_i-r(g-1)\\
            &= (\deg D)\frac{\dim\g+r}{2}-r(g-1),
\end{aligned}
\end{equation}
which immediately implies (taking account of the relation $\deg\K=2(g-1)$) assertion of the lemma.
\end{proof}

\subsection{Example: Hitchin systems}
\label{SS:Hitchin}

Here we show that in the case when $D=\K$, and $\g$ is one of the classical Lie algebras $\gl(n)$, $\so(2n)$, $so(2n+1)$ or $\spn(2n)$ the number of the above found integrals is exactly what is required for integrability of the system. The Lax pairs are considered in the standard representations of the corresponding classical Lie algebras. Observe that the case $D=\K$ corresponds to the Hitchin systems. We will give only dimensional check of this fact, delaying the full check until another occasion. For $\g=\gl(n)$ the fact is proved in \cite{Klax}.

We will give the space $\L^D$ by Tyurin parameters, see relations \refE{pTgln}--\refE{pTsp2n}. The elements of the space $\L^D$ are operators in the standard representation of the Lie algebra~$\g$. Observe that the set $\{h\}$, defining the gradings, is given up to an action of the group $\Ad G$ while the group acting in the standart module is different, namely, it is the classical group $G$, corresponding to the Lie algebra~$\g$, itself. The relation between $G$ and $\Ad G$ is known: $\Ad G\cong G/{\mathcal Z}(G)$, where ${\mathcal Z}(G)$ is the center og $G$. For this reason, in the cases when ${\mathcal Z}(G)$ is nontrivial (these are $\g=\gl(n)$ with the center ${\mathcal Z}(G)=\C E$, and $\g=\so(2n)$ with the center ${\mathcal Z}(G)=\{\pm E\}$) it is necessary to projectivise the standard module (the Tyurin parameters $\a$) to obtain a well-defined action of $\Ad G$.

Taking account of these remarks, we show that for the simple classical Lie algebras $\so(2n)$, $\so(2n+1)$, $\spn(2n)$ the following relation takes place:
\begin{equation}\label{E:LDtot}
     \dim\L^D=(\dim\g)(\deg D+1).
\end{equation}
The dimension of $\L^D$ is equal to the dimension of  $\L^D_{\Gamma,\{h\}}$ plus the number of Tyurin parameters $\a$ and $\ga$. We take $|\Gamma|$ to be equal to $ng$. In every case we will check that the number of Tyurin parameters is equal to $(\dim\g)g$.  Since $\dim\L^D_{\Gamma,\{h\}}=(\dim\g)(\deg D-g+1)$, the relation \refE{LDtot} will be proven.

For $\g=\so(2n)$ we have $\a\in\C^{2n}$. Taking account of the projectivization we have $\a\in\C P^{2n-1}$. These are $2n-1$ parameter at every $\ga$, and the $\ga$ itself gives one more parameter. We also have one relation: $\a^t\a=0$. In total, we have $(2n-1)\cdot ng=(\dim\g)g$ parameters.

For $\g=\so(2n+1)$ we have $\a\in\C^{2n+1}$, there is one relation $\a^t\a=0$, and $\ga$ itself gives one more parameter. In total, we have  $(2n+1)\cdot ng=(\dim\g)g$ parameters.

For $\g=\spn(2n)$ $\a\in\C^{2n}$, and $\ga$ itself gives one more parameter. In total, we have  $(2n+1)\cdot ng=(\dim\g)g$ parameters.

By this counting the relation \refE{LDtot} is proven in all three cases. Dividing $\L^D$ by $G$ defines an invariant quotient space $\L^D_0$ of dimension $\dim\g\deg D$. For $D=\K$ we have $\dim\L^\K_0=2(\dim\g)(g-1)$. Observe that the number of integrals for $D=\K$ is equal to $(\dim\g)(g-1)$, i.e. the system is integrable.

In the case $\g=\gl(n)$ we have $\dim\L^\K_0=2(n^2(g-1)+1)$, and the number of integrals is equal to the half of this dimension. Indeed, in this case $\dim\L^\K_{\Gamma,\{h\}}=n^2(g-1)+1$. The anomalous unit in this relation is due to the following. The elements $L\in\L$ taking values in $\sln(n)$ contribute $(n^2-1)(g-1)$ into the dimension, as above. The elements taking values in the subalgebra of scalar matrices possess the following property: they are holomorphic outside the divisor $D$. Indeed, for them $\tr L_{-1}=0$ (compare with \refE{pTgln}), and, for the reason they are scalar, $L_{-1}=0$. Therefore if the divisor $D$ is special, a certain "dimension anomaly"\ emerges. Namely, for $D=\K$ the dimension of the space of such elements is equal to $h^0(\K)=g$, which together with $(n^2-1)(g-1)$ exactly gives the expression for $\dim\L^\K_{\Gamma,\{h\}}$. The number of Tyurin parameters is equal to $n^2g$, as above (by the projectivization, we have $\a\in\C P^{n-1}$; this gives $n-1$ parameter for every $\ga$, taking account of the number of the points $\ga$ themselves gives $n\cdot ng$ parameters). Therefore $\dim\L^\K=n^2(g-1)+1 +n^2g=2n^2(g-1)+n^2+1$. Dividing by $GL(n)$ decreases the dimension by $n^2-1$ (the center is not counted since it does not affect the Lax equation, i.e. we actually consider $\L^\K/SL(n)$). Finally we obtain $\dim\L^\K_0=\dim\L^\K-(n^2-1)=2(n^2(g-1)+1)$. As for the number of integrals, for $\g=\gl(n)$ there additionally emerges $h^0(D)$ in \refE{N_int} (there is an invariant of degree 1, namely, trace, while for the semisimple Lie algebras the degrees of the basis invariants are always not less than 2). For $D=\K$ this exactly gives the above "anomalous"\ unit.
\subsection{Example: the Calogero--Moser systems}
\label{SS:CMoser}
We consider here one more series of examples, namely the elliptic Calogero--Moser systems for classical Lie algebras.

We start with the example considered in \cite{Klax}, namely with the elliptic Calogero--Moser system for the root system $A_n$. We define a $\gl(n)$-valued Lax operator giving its entries as follows:
\begin{equation}\label{E:matrLgl}
 L_{ij}=f_{ij}\frac{\s(z+q_j-q_i)\s(z-q_j)\s(q_i)}{\s(z)\s(z-q_i)\s(q_i-q_j)\s(q_j)}\ (i\ne j), \ \
 L_{jj}=p_j
\end{equation}
where $\s$ (and also $\wp$ below) are the Weierstra\ss\ functions,
$f_{ij}\in\C$ are constants. Up to the constants $f_{ij}$ this form of the operator $L$ is determined by the requirement that it is elliptic, and has simple poles at $z=q_i$ ($i=1,\ldots,n$) and $z=0$. The last point is the unique element of the set $\Pi$. By means of the reduction of the remaining gauge degrees of freedom it was obtained in \cite{Klax} that $f_{ij}f_{ji}=1$. For the second order Hamiltonian corresponding to the pole at $z=0$ we obtain, in accordance with \cite{Sh_DGr}, and up to a normalization,
\[
H=\res_{z=0}\,\,\, z^{-1}\!\left(-{1\over 2}\sum_{j=1}^n
p_j^2-\sum\limits_{i<j} L_{ij}L_{ji}\right).
\]
By the addition theorem for the Weierstra\ss\ functions
\[
-L_{ij}L_{ji}=
\frac{\s(z+q_i-q_j)\s(z+q_j-q_i)}{\s(z)^2\s(q_i-q_j)^2}=\wp(q_i-q_j)-\wp(z),
\]
hence
\[
H=-{1\over 2}\sum_{j=1}^n p_j^2+\sum\limits_{i<j} \wp(q_i-q_j).
\]

Next consider the cases $\g=\so(2n)$, $\g=\spn(2n)$. We set
$\Gamma=\{\pm q_1,\dots, \pm q_n\}$ ($|\Gamma|=2n$), and make use of the results of the \refEx{gr_ell} of \refSS{constrM}.

To obtain the Calogero--Moser system corresponding to the series $D_n$, let us take the Lax operator in the form
\begin{equation} \label{E:Lax_D_n}
L=\begin{pmatrix}
     A& B \\
     C& -A^t
  \end{pmatrix},
\end{equation}
where $A$, $B$, $C$ are $n\times n$ matrices, $B=-B^t$, $C=-C^t$.
It is a standard form of the elements of the Lie algebra Ли $\so(2n)$.

Let $A$ be given by the relations \refE{matrLgl}. For $i>j$ we set
\begin{equation}\label{E:matrBC}
B_{ji}=f^B_{ji}\frac{\s(z-q_j-q_i)\s(z+q_i)}{\s(z)\s(z-q_j)\s(q_i+q_j)},
\ \
C_{ij}=f^C_{ij}\frac{\s(z+q_j+q_i)\s(z-q_j)}{\s(z)\s(z+q_i)\s(q_i+q_j)}
\end{equation}
where $f^B_{ij},f^C_{ij}\in\C$ are constant. These relations completely determine the matrices $B$ and $C$ by skew-symmetry of the last. Similar to the case $\gl(n)$ we obtain $f^B_{ij}f^C_{ji}=-1$ reducing the remaining degrees of freedom. For the Hamiltonian we obtain
\begin{align*}
H=&-\res_{z=0}\,\,\, z^{-1}\left(\sum_{i=1}^n p_i^2+2\sum_{i<j}
A_{ij}A_{ji}+2\sum_{i<j}B_{ij}C_{ji}\right)\\
=&-\sum_{i=1}^n
p_i^2+2\sum_{i<j}\wp(q_i-q_j)+2\sum_{i<j}\wp(q_i+q_j)
\end{align*}
which is a conventional form of the second order Hamiltonian for the elliptic Calogero--Moser $D_n$ system.

Proceeding, further on, with the Calogero--Moser system for the root system $C_n$, we take $L$ in the form \refE{Lax_D_n} where $B=B^t$, $C=C^t$. We define the corresponding matrix entries $A_{ij}$ by the right hand sides of the relations
\refE{matrLgl}, and $B_{ij}$, $C_{ij}$ by the relations \refE{matrBC}. The relations \refE{matrBC} make sense also for $i=j$, and, in this case, their contribution into the second order Hamiltonian is equal to
\[
B_{ii}C_{ii}=f_{ii}^Bf_{ii}^C(\wp(2q_i)-\wp(z))
\]
where we can set $f_{ii}^Bf_{ii}^C=2$. Therefore
\[
H = -\sum_{i=1}^n
p_i^2+2\sum_{i<j}\wp(q_i-q_j)+2\sum_{i<j}\wp(q_i+q_j)+2\sum\limits_{i=1}^n
\wp(2q_i)
\]
which is a conventional form of the second order Hamiltonian of the elliptic Calogero--Moser system in the symplectic case.

Finally, for the root system $B_n$ we take $\Gamma=\{\pm q_1,\ldots,\pm q_n, q_0\}$ (see \refEx{gr_so2n+1} of \refSS{constr}), and the Lax operator
\begin{equation} \label{E:Lax_B_n}
L=\begin{pmatrix}
     A   & a &  B \\
    -b^t & 0 & -a^t\\
     C   & b & -A^t
  \end{pmatrix},
\end{equation}
where $A$, $B$, $C$ are $n\times n$ matrices, $B=-B^t$, $C=-C^t$, $a,b\in\C^n$. This is a standard form of an element of the Lie algebra $\so(2n+1)$. Let $A$, $B$, $C$ be given as for $D_n$, while $a$, $b$ are given by their coordinates as follows:
\begin{equation}\label{E:col_ab}
a_i=f^a_i\frac{\s(z-q_0-q_i)\s(z)}{\s(z-q_0)\s(z-q_i)\s(q_i)},
\ \
b_i=f^b_i\frac{\s(z-q_0+q_i)\s(z-q_i)}{\s(z)\s(z-q_0)\s(q_i)}.
\end{equation}
By the addition formulas, $a_ib_i=f^a_if^b_i(\wp(q_i)-\wp(z-q_0))$.
For the Hamiltonian we have
\begin{align*}
H=&-\res_{z=0}\,\,\, z^{-1}\left(\sum_{i=1}^n p_i^2+2\sum_{i<j}
A_{ij}A_{ji}+2\sum_{i<j}B_{ij}C_{ji}+2\sum_{i=1}^n a_ib_i\right) \\
=&-\sum_{i=1}^n
p_i^2+2\sum_{i<j}\wp(q_i-q_j)+2\sum_{i<j}\wp(q_i+q_j)+ 2\sum_{i=1}^n(\wp(q_i)-\wp(z-q_0)).
\end{align*}
Observe that $p_0=0$, therefore $q_0=const$ gives an invariant submanifold of the hierarchy of commuting flows (in the Lax form of the equations of motion, it corresponds to the relation $\nu_0=0$, following from the fact that $M$ is an element of $\so(2n+1)$). Taking the restriction of the system to this invariant submanifold, and omitting $\wp(q_0)$ as a nonessential constant, we arrive to the conventional form of the second order Hamiltonian of the elliptic Calogero--Moser system $B_n$.

\subsection{Holomorphy of spectra and method of deformation of Tyurin parameters}
\label{SS:Holomorphy}
Here we give proofs of Lemmas \ref{L:elimin} and \ref{L:specMhol}, and discuss the relation of the last to the equations of motion of the element $h$ which lead to the equations of deformation of Tyurin parameters in the case of classical Lie algebras.
\begin{proof}[Proof of \refL{elimin}]
We consider the adjoint action on $L$ of the local holomorphic function $e^{-h\ln z}$, where $z$ is a local coordinate in the neighborhood of a $\ga\in\Gamma$, $h$ is the element of the Cartan subalgebra used above for giving the grading on $\g$ (a choice of the branch of $\ln z$ does not matter). On the homogeneous subspace of degree $s$ of this grading, i.e. on $\g_s\subset\g$, the operator $-\ad h$ acts as multiplication by $-s$. Therefore, the operator $\Ad e^{-h\ln z}$ operates by multiplication by $z^{-s}$. In the Laurent expansion of the element $L\in\L$ at the point $\ga\in\Gamma$, the coefficient at $z^i$ is an element of  $\tilde\g_i=\bigoplus\limits_{-k\le s\le i}\g_s$. Under the action of $\Ad e^{-h\ln z}$ there emerges a sum of degrees $z^{i-s}$, $-k\le s\le i$ (with certain coefficients) at this place of the Laurent expansion. Obviously, these are nonnegative degrees.
\end{proof}
\begin{remark}\label{R:dnw}
For $M$-operators, a similar argument does not lead to any similar result because for $0\le i<k$, the components $M_i$ have nonzero projections to the homogeneous subspaces $\g_s$ with $s>i$, and there will emerge negative degrees of~$z$ at those places.
\end{remark}
A prototype of \refL{elimin} for $\g=\gl(n)$, has been formulated in \cite{Klax}. The proof given there, in this particular case is also equivalent to our proof. First, the $L_{-1}$ had been transformed to the form drown in the Figure~\ref{An} by conjugation. This corresponds to $h={\rm diag} (-1,0,\ldots,0)$. Then the conjugation by the matrix ${\rm diag} (z,0,\ldots,0)$, which is exactly the same as the $e^{-h\ln z}$ here, had been made.
\begin{proof}[Proof of \refL{specMhol}] In accordance with \refSS{LP}, we represent the grading element in the neighborhood of a $\ga\in\Gamma$ in the form $h=g^{-1}h_0g$ where $h_0$ gives the grading at the point $\ga$, and $g$ is a local holomorphic function taking values in $G$, such that $g(\ga)\in G_0$, $G_0$ being the centralizer of $h_0$ in $G$ (i.e. $Lie(G_0)=\g_0$). Then it follows from \refL{elimin} that
\begin{equation}\label{E:psi_fact}
\Psi=\Psi_0e^{-h_0\ln z}g
\end{equation}
where $\Psi_0$ is holomorphic, and has a holomorphic inverse in the neighborhood of $\ga$. By this, making use of $\partial_az=-\nu_a$ \refE{local1}, we obtain
\begin{equation}\label{E:Fahol}
\begin{aligned}
    F_a=&-\partial_a\Psi\cdot\Psi^{-1}+\Psi M_a\Psi^{-1}\\
    =&-\partial_a\Psi_0\cdot\Psi_0^{-1}-\Psi_0\frac{\nu_ah_0}{z}\Psi_0^{-1} - \Psi_0e^{-h_0\ln z}\partial_ag\cdot g^{-1}e^{h_0\ln z}\Psi_0^{-1} +\\
    &+\Psi_0e^{-h_0\ln z}g^{-1}\left( \frac{\nu_ah_0}{z}+\sum_{i=-k}^\infty M_{a,i}z^i\right)ge^{h_0\ln z}\Psi_0^{-1}.
\end{aligned}
\end{equation}
Taking into account that $\partial_a\Psi_0\cdot\Psi_0^{-1}$ is holomorphic, we obtain
\[
    F_a=\Psi_0e^{-h_0\ln z}gAg^{-1}e^{h_0\ln z}\Psi_0^{-1}+O(1)
\]
where
\[
   A=-g^{-1}\partial_ag+\frac{\nu_a}{z}(h_0-g^{-1}h_0g)+\sum_{i=-k}^\infty M_{a,i}z^i.
\]
Next, we find $g$ from the equation
\begin{equation}\label{E:dvig}
  g^{-1}\partial_ag=\frac{\nu_a}{z}(h_0-g^{-1}h_0g)+\sum_{0\le i<k}M_{a,i}z^i.
\end{equation}
Then $e^{-h_0\ln z}gAg^{-1}e^{h_0\ln z}$ is holomorphic by the same argument as $L$ (the terms mentioned in \refR{dnw} are absent).
\end{proof}
From $h=g^{-1}h_0g$ as a starting point, we find the equations of motion of $h$. To this goal we differentiate the equality, and obtain
\begin{equation}\label{E:dvih}
  \partial_ah=[h,g^{-1}\partial_ag].
\end{equation}
By the equation \refE{dvig}, setting the initial condition to $g(z)|_{z=0}=id$, we obtain in the zero approximation in $z$ that $g^{-1}\partial_ag=M_{a,0}$ (observe that  $h_0-g^{-1}h_0g=0$ for $z=0$). This results in the following form of \refE{dvih}:
\begin{equation}\label{E:dvih1}
  \partial_ah=[h,M_{a,0}],
\end{equation}
which exactly coincides with the equation \refE{local1} for $h$.
\begin{example}\label{PT_ex}
Let $\g=\gl(n)$. It easy follows from \refSS{An} that $h=\a\mu^t$ in this case where $\a=g(1,0,\ldots,0)^t$,
$\mu^t=(1,0,\ldots,0)g^{-1}$. Differentiating the relation $h=\a\mu^t$, and applying \refE{dvih1}, we obtain
\[
  \partial_a\a\cdot\mu^t+\a\partial_a\mu^t=-M_{a,0}\a\mu^t+\a\mu^tM_{a,0}.
\]
Next, multiply this equality from the right by such vector $\vartheta$ that $\mu^t\vartheta=1$. Then
\begin{equation}\label{E:dviTyur}
  \partial_a\a=-M_{a,0}\a+\l\a
\end{equation}
where $\l=(-\partial_a\mu^t+\mu^tM_{a,0})\vartheta\in\C$.

\refE{dviTyur} is nothing but the motion equation for the Tyurin parametrs found in \cite{Klax}.  In addition, the following equation is true: $\partial_az=-\mu^t\a$. Indeed, it is easy to see, that  $h=-\mu^t\a\cdot diag(-1,0,\ldots,0)$, which implies $\nu_a=-\mu^t\a$. Then the relation $\partial_az=-\nu_a$ \refE{local1} gives the required equation.

\end{example}

\section{Lax integrable systems and conformal field theory}\label{S:LaxCFT}

In this chapter, we will briefly outline the results of \cite{Sh_L_KZ} (see also \cite{Sh_DGr}) which enable one to assign each integrable system of the above discussed type with a unitary projective representation of the corresponding Lie algebra of Hamiltonian vector fields. To the family of spectral curves over the phase space of a system, we apply the technique earlier developed for the tautological bundle over the moduli space of curves \cite{WZWN1,WZWN2,Sh_DGr,Schlich_DGr}. This enables us to construct a Knizhnik--Zamolodchikov-type connection on the phase space, and represent the Hamiltonian vector fields by covariant derivatives with respect to this connection. From the physical point of view this is a Dirac-type prequantization. For the Hitchin systems, the idea of their quantization by means the Knizhnik--Zamolodchikov connection has been multiply used, or at least mentioned, in physical literature (\cite{FeW}, \cite{D_I}, \cite{Olsh_Lev, Olsh}), with different restrictions, for example, for the second order Hamiltonians, or on elliptic curves. In \cite{Sh_L_KZ,Sh_DGr} this idea is realized in full geberality, in the context of Lax equations with the spectral parameter on a Riemann surface, and for the Hamiltonians of all orders.


\subsection{Centralizer of an element, and its vacuum representation}\label{SS:aff_KN}

Let $\L$ be a Lax operator algebra, and $L\in\L$, $\Psi$ diagonalizes $L$ as it was defined in \refSS{KrPh}, i.e.
\[
   \Psi L=K\Psi
\]
where $K=diag(\k_1,\ldots,\k_n)$. The diagonal elements of the matrix $K$ are roots of the characteristic equation $\det(L(z)-\k)~=~0$. The curve given by this equation, is called the \emph{spectral curve} of the element~$L$. We denote it $\Sigma_L$. It is a $n$-sheeted branch covering of the curve $\Sigma$. \refL{elimin} immediately implies that the meromorphic function $K$ is holomorphic on~$\Gamma$, and has poles only at $P\in\Pi$.

Below, we follow the assumptions and notation of \refS{Ham_th}. In particular, we consider only classical Lie algebras, and use the Tyurin parametrization.

If we denote by $\A$ the algebra of scalar functions on $\Sigma$ holomorphic except at $P\in\Pi$, then $K\in\h\otimes\A$ where $\h\subset\g$ is a diagonal subalgebra. $\A$ is called the \emph{function Krichever--Novikov algebra}, and $\hb=\h\otimes\A$ is called the \emph{current Krichever--Novikov algebra} (only commutative current algebras emerge here, though there could be any reductive Lie algebra instead $\h$ in general). For a detailed presentation of the Krichever--Novikov algebras see \cite{Sh_DGr,Schlich_DGr}).

Vise verse: \emph{for any $h\in\hb$ we have: $\Psi^{-1} h\Psi\in\L$} \cite{Sh_L_KZ,Sh_DGr}, therefore we arrive to the assertion:
\begin{lemma}\label{L:diag1}
The Lie subalgebra commuting with an $L\in\L$, is isomorphic to  $\Psi\hb\Psi^{-1}$.
\end{lemma}
Since diagonal elements of the matrices from $K\in\hb$ correspond to the sheets of the covering $\Sigma_L\to\Sigma$, every element $K\in\hb$ can be pulled back to $\Sigma_L$, and will give a meromorphic scalar function on $\Sigma_L$, having poles in the preimage of $\Pi$ only. We denote the algebra of such function by $\A_L$. The inverse mapping $\A_L\to\hb$ is given by means construction of direct image, a function $A\in\A_L$ is assigned with the $K\in\hb$ where $K(P)=diag(A(P_1),\ldots,A(P_n))$, and $P_1,\ldots,P_n$ are preimages of the point $P\in\Sigma$.

Consider an element $A\in\A_L$, and its direct image $K\in\hb$. Let $L_A=\Psi K\Psi^{-1}$. Every sheet of the curve $\Sigma_L$ is associated with a certain row in $K$. At the branching points, a coincidence of the eigenvalues of  $K$ is possible. The order of the diagonal entries of the matrix $K$ depends on the order of the covering sheets. The ambiguity here is the same as in the definition of the matrix $\Psi$: the permutation of the rows of $\Psi$ corresponding to an element $w$ of the Weyl group descends to the transformation $\Psi\to w\Psi$ (it can be easy checked for transpositions), and $K$ transforms as follows: $K\to wKw^{-1}$. Hence $L_A=\Psi^{-1}K\Psi$ does not depend on $w$, and $L_A$ is well-defined.
\begin{lemma}\label{L:diag2}
The mapping $\A_L\to\L$ sending $A$ to $L_A$, establishes an isomorphism between $\A_L$ and the Lie subalgebra in $\L$ consisting of elements commuting with $L$.
\end{lemma}

$\L$ has a canonical representation in the space $\mathcal F$ of meromorphic vector-functions taking values in $\C^n$, holomorphic outside the sets $\Pi$ and $\Gamma$, and having the expansion of the form
\[
 \psi(z)=\nu\frac{\a}{z} +\psi_0+\ldots ,
\]
at $\ga\in\Gamma$ where $\a$ is the vector of Tyurin parameters.  $\mathcal F$ is an almost graded $\L$-module \cite{Sh_DGr,Schlich_DGr}. The invariance of the space $\mathcal F$ with respect to $\L$, for classical Lie algebras, is easy derived from \refE{pTgln}, \refE{pTso2n}, \refE{pTsp2n}, \refE{pTso2n+1}. We denote the space of semi-infinite external forms on $\mathcal F$ of a fixed charge by ${\mathcal F}^{\infty/2}$ \cite{KaRa}. It is a vacuum $\L$-module. Due to the above constructed morphism  $\A_L\to\L$ (\refL{diag2}), we can consider ${\mathcal F}^{\infty/2}$ as an $\A_L$-module as well.

We will conclude this section with \emph{Sugawara representation}. Given a vacuum module of a current Krichever--Novikov algebra, the Sugawara construction canonically gives a representation of the corresponding Lie algebra of vector fields, in the same space. We apply it to the $\A_L$-module ${\mathcal F}^{\infty/2}$ where $\A_L$ is considered as a commutative Lie algebra, and obtain a representation of the Lie algebra of meromorphic vector fields on $\Sigma_L$ holomorphic outside the preimage of $\Pi$. We denote this representation by $T$ below. We do not give any definition of the Sugawara representation here. It is presented in several places: \cite{KaRa} for the loop algebras, \cite{Sh_DGr,Schlich_DGr} for the Krichever--Novikov algebras. In particular, for the commutative Krichever--Novikov current algebras the construction has been originally proposed in \cite{KNFb}.

Observe that the Sugawara construction for the Lax operator algebras does not exist (and, perhaps, can not exist, see the discussion in \cite{Sh_DGr}), and this is one of the reasons of addressing the algebra $\A_L$ in this context.

Thus, every point of the phase space $\P^D$ is associated with an element $L\in\L$ (where $\L$ is specific for every point), its spectral curve $\Sigma_L$, and the $\L$-module ${\mathcal F}^{\infty/2}$ over it. As a result, we obtain a family of spectral curves on $\P^D$, and a bundle over it with an infinite-dimensional fiber. Symbolically, the whole picture is drown in the Figure~\ref{Spectral}. Below, we construct a finite rank quotient sheaf of this bundle, and a projective flat connection on it. In this way we transform our object to a conformal field theory in the sense of \cite{FrSh}.

\begin{figure}
\begin{picture}(350,250)
\qbezier(70,120)(105,130)(142,133) 
\qbezier(167,134)(230,140)(350,140) 
\qbezier(70,120)(20,100)(10,40)
\qbezier(10,40)(90,100)(250,0)
\qbezier(250,0)(250,100)(350,140)
\put(280,40){{\Large $\P^D$}}
\put(230,110){\circle*{3}}
\put(218,102){\Large $L$}
\put(230,110){\line(0,1){50}}
\put(230,195){\oval(40,70)}
\qbezier(220,170)(230,160)(240,180) 
\qbezier(225,167)(230,183)(235,174) 
\qbezier(215,200)(220,216)(225,207) 
\qbezier(234,206)(239,215)(244,199)
\qbezier(218,205)(220,200)(222,210) 
\qbezier(236,208)(239,200)(242,201)
\put(189,195){\Large $\Sigma_L$}
\put(239,223){\circle*{3}}
\put(236,240){$P_\infty$}
\put(238,238){\line(0,-1){12}}
\qbezier(235,230)(230,213)(247,222)
\qbezier(230,230)(225,208)(249,217)
%
\qbezier(265,250)(320,270)(330,235)%
\put(305,250){\circle*{3}}  
\qbezier(295,255)(290,240)(318,248)
\qbezier(285,255)(270,230)(326,240)
\put(275,205){\Large $d_L$}
%
%
\put(291,255){\vector(-1,-3){4}}
\put(305,240){\vector(4,1){12}}
\put(310,212){\Large $\rho(X)$}
\put(323,225){\line(-1,4){4}}
%
\put(283,210){\line(1,2){15}}
\put(270,210){\line(-3,1){25}}
%
%
\put(230,110){\vector(3,1){40}}
\put(270,120){{\Large $X$}}
\put(235,118){\vector(3,1){30}}
\put(235,105){\vector(3,1){30}}
%
\put(80,110){\circle*{3}}
\put(80,110){\line(0,1){50}}
\put(80,195){\oval(40,70)}
\qbezier(70,180)(80,160)(90,170) 
\qbezier(75,174)(80,183)(85,167) 
\qbezier(95,200)(90,216)(85,207) 
\qbezier(76,206)(71,215)(66,199)
\qbezier(92,205)(90,200)(88,210) 
\qbezier(74,208)(71,200)(68,201)
%
%
\put(155,90){\circle*{3}}
\put(155,90){\line(0,1){25}}
\put(155,133){\oval(20,35)}
\qbezier(148,130)(155,110)(163,130) 
\qbezier(152,125)(155,130)(159,125) 
\qbezier(147,137)(150,143)(154,138) 
\qbezier(157,138)(160,143)(164,137)
\qbezier(149,137)(150,133)(152,138) 
\qbezier(159,138)(160,133)(162,137)
\end{picture}
\caption{Family of spectral curves on the phase space}\label{Spectral}
\end{figure}
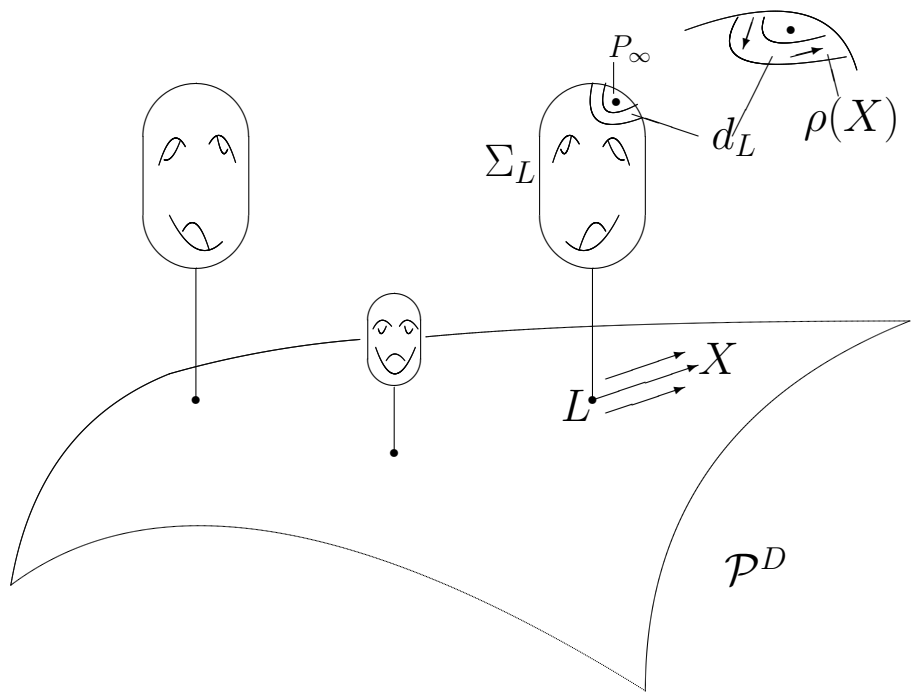


\subsection{Knizhnik--Zamolodchikov connection}
\label{SS:CFT}
Let $X$ be a tangent vector to $\P^D$ at a point $L$. Consider a deformation of the complex structure of the corresponding spectral curve in the direction of $X$. It is given by a Kodaira--Spencer class in $H^1(\Sigma_L, T\Sigma_L)$ where $T\Sigma_L$ is a tangent sheaf on $\Sigma_L$. The cocycle $\rho(X)$ representing this class can be constructed as follows. We fix a local section of the sheaf of spectral curves. On the Figure~\ref{Spectral} its intersection with the fiber over $L$ is denoted by the marked point $P_\infty$, the notation being used through. Consider the family of gluing functions (giving a smooth structure on spectral curves) defined in an annulus with the "center" at the marked point. We denote such gluing function by $d_L$. We set $\rho(X)=d_L^{-1}\partial_X d_L$. Since $d_L$ can be considered as a diffeomorphism of the annulus, $\rho(X)$ is a local field defined there.

It is shown in \cite{WZWN1,WZWN2,Sh_DGr} that the Kodaira--Spencer cocycle can be represented by a Krichever--Novikov vector field (i.e. by a global meromorphic vector field on $\Sigma_L$ holomorphic outside $\Pi$), and the ambiguity of the construction is compensated by passing to the quotient sheaf which will be defined below. According to that, below we consider $\rho(X)$ as an element of the space $\V_L^{\rm reg}\backslash\V_L/\V_L^{(1)}$ where $\V_L$ is the Lie algebra of Krichever--Novikov vector fields on $\Sigma_L$ (with $\Pi$ as the set of allowed poles), $\V_L^{(1)}$ is the direct sum of its homogeneous subspaces of nonnegative degrees, and $\V_L^{\rm reg}\subset\V_L$ is the subspace of vector fields vanishing at $P_\infty$.\footnote{It is convenient to think that $P_\infty\in\Pi$ and use the following splitting of the set $\Pi:\,\Pi=(\Pi\backslash\{P_\infty\})\cup\{P_\infty\}$ to define an almost graded structure --- compare with \refSS{constr}} Both these subspaces are Lie subalgebras in $\V_L$.

As soon as $\rho(X)$ is a Krichever--Novikov vector field, we consider  $T(\rho(X))$ where $T$ is the Sugawara representation, and then define the operators
\[  \nabla_{X}=\partial_{X}+T(\rho(X)).
\]

Next, we consider the sheaf of $\A_L$-modules ${\mathcal F}^{\infty/2}$ on $\P^D$. Let $\A_L^{reg}\subset\A_L$ be a subalgebra of functions regular at the point $P_\infty$. The sheaf of quotient spaces ${\mathcal F}^{\infty/2}/\A_L^{reg}{\mathcal F}^{\infty/2}$ on $\P^D$ is called the sheaf of \emph{coinvariants}, or the sheaf of \emph{conformal blocks} in another terminology.
\begin{theorem}\label{T:pr_flat} The operators $\nabla_X$ define a projective flat connection $\nabla$ on the sheaf of coinvariants, in particular
\[
 [\nabla_X,\nabla_Y]=\nabla_{[X,Y]}+\l(X,Y)\cdot id
\]
where $\l$ is a cocycle on the Lie algebra of tangent vector fields on $\P^D$, $id$ is the identity operator.
\end{theorem}
In \cite{WZWN2,Sh_DGr,Schlich_DGr} \refT{pr_flat} is formulated and proved for the conformal field theory on the moduli space of curves with marked points and fixed, up to a certain order, jets at those points. We claim that here (as well as in \cite{Sh_DGr}), the situation is completely similar, and the same proofs are working. By this analogy, we call the projective flat connection defined by \refT{pr_flat}, the {\it Knizhnik--Zamolodchikov connection}.

The horizontal sections of the Knizhnik--Zamolodchikov connection are also called \emph{conformal blocks}.

\subsection{Representation of the Lie algebra of Hamiltonian vector fields}\label{SS:rep}

By \refT{pr_flat} the mapping $X\to\nabla_X$ is a projective representation of the Lie algebra of vector fields on $\P^D$ in the space of sections of the sheaf of coinvariants. Denote this representation by~$\nabla$. The restriction of the representation $\nabla$ on the Lie subalgebra of Hamiltonian vector fields gives a projective representation of the last. It is remarkable for the reason it is unitary, and the representation operators of Hamiltonians being the spectral invariants (i.e. corresponding to the invariants of the Lie algebra $\g$ as described in \refSS{Ham_sys}), are mutually commuting.
\begin{theorem}[\cite{Sh_L_KZ,Sh_DGr}]
If $X$, $Y$ are Hamiltonian vector fields, and their Hamiltonians are in involution, then $[\nabla_X,\nabla_Y]=\l(X,Y)\cdot~id$. If Hamiltonians are spectral invariants then
\[ [\nabla_X,\nabla_Y]~=~0.
\]
\end{theorem}
\begin{proof} The first assertion immediately follows from \refT{pr_flat}, since $[X,Y]=0$.

The coefficients of the spectral curve are integrals for the Lax equations. Therefore, if $X$ is a Hamiltonian vector field, and its Hamiltonian is a spectral invariant, then the complex structure, as well as the gluing functions, are invariant along the phase trajectories of the vector field~$X$. Let $d_L$ be the family of the gluing functions (depending on $L$). The invariance along the trajectories immediately implies that $\partial_Xd_L=0$, hence $\rho(X)=d_L^{-1}\partial_Xd_L=0$, and therefore $\nabla_X=\partial_X$. Let $H_X$, $H_Y$ be Hamiltonians depending only on the spectrum of the Lax operator. Then, according to just said,
$[\nabla_X,\nabla_Y]=[\partial_X,\partial_Y]=\partial_{[X,Y]}$,
and $[X,Y]=0$. This implies that $[\nabla_X,\nabla_Y]~=~0$.
\end{proof}

Let $\mathcal G$ be a Lie algebra with an antilinear involution $\dag$, $T$ be its representation in the linear space $V$. The Hermitian scalar product in $V$ is called contravariant if  $T(X)^\dag=T(X^\dag)$ where, on the left hand side, $\dag$ denotes the Hermitian conjugation of operators. According to \cite{KaRa}, a pair consisting of a representation $T$, and a contravariant scalar product on $V$, is called a \emph{unitary representation} of the Lie algebra $\mathcal G$. In such case the restriction of $T$ to the Lie subalgebra of the elements $X$ such that $X^\dag=-X$ (i.e. to the real subalgebra) is unitary in the conventional sense, i.e. $T(X)^\dag=-T(X)$.

We can define a scalar product in the fibres of the sheaf of coinvariants induced by the standard scalar product of semiinfinite monomials \cite{KaRa, KNFb,Sh_DGr}, i.e. by declaring the semiinfinite monomials consisting of basis vectors as comultipliers to be orthogonal.

Let $\w$ be the symplectic form on $\P^D$, i.e. the restriction of the Krichever--Phong form to $\P^D$, and $\w^p/p!$ be the corresponding volume form on $\P^D$. Let $\L^2(\w^p/p!)$ be the  the space of sections of the sheaf of coinvariants which are quadratically integrable in the measure given by the volume form. By the square of a section we mean the square with respect to the above introduced scalar product.
\begin{theorem}[\cite{Sh_L_KZ,Sh_DGr}]\label{T:unit}
The representation $\nabla : X\to\nabla_X$ of the Lie algebra of Hamiltonian vector fields on $\P^D$ in the space of smooth sections in $\L^2(C,\w^p/p!)$ is unitary.
\end{theorem}
We refer to \cite{Sh_L_KZ,Sh_DGr} for the proof. Here, we only notice that, by the Poincar\'{e}\ theorem on absolute integral invariants, the symplectic form, and its powers, are absolute integral invariants of Hamiltonian phase flows \cite{Arn}. Hence the volume form $\w^p/p!$ defines an invariant measure on $\P^D$ with respect to the Hamiltonian phase flows, and averaging in invariant measure gives a unitary representation.



\end{document}